\documentclass[a4paper,10pt,allowtoday, accepted=2023-06-07]{quantumarticle}
\pdfoutput=1
\usepackage[breaklinks=true]{hyperref}
\usepackage{fontenc}
\usepackage{amsmath}
\usepackage{amsfonts,amssymb,amsthm, bbm,braket,mathtools}
\usepackage{mathtools}
\usepackage{xcolor}
\usepackage{graphicx}  
\usepackage{subfigure} 
\usepackage{amsbsy} 

\usepackage[numbers]{natbib}
\usepackage{silence}

\usepackage{comment}
\usepackage{algorithm}
\usepackage{tikz}
\usepackage{cancel}
\usepackage{silence}

\newtheorem{theorem}{Theorem}
\newtheorem*{theorem*}{Theorem}

\newtheorem{lemma}[theorem]{Lemma}
\newtheorem*{lemma*}{Lemma}

\newtheorem*{remark}{Remark}

\newtheorem{definition}[theorem]{Definition}

\newcommand{\SYKf}{$\text{SYK-}\,4\ $}
\newcommand{\SSSYKf}{$\text{SSYK-}\,4\ $}
\renewcommand{\exp}{\,\text{exp}}

\begin{document}

\title{Optimizing sparse fermionic Hamiltonians}

\author{Yaroslav Herasymenko}
\affiliation{QuSoft $\&$ CWI, Science Park 123
1098 XG  Amsterdam, The Netherlands}
\affiliation{QuTech, Delft University of Technology, Lorentzweg 1, 2628 CJ Delft, The Netherlands}
\author{Maarten Stroeks}
\affiliation{QuTech, Delft University of Technology, Lorentzweg 1, 2628 CJ Delft, The Netherlands}
\affiliation{EEMCS Faculty, Delft University of Technology, Van Mourik Broekmanweg 6, 2628 XE, Delft, The Netherlands}
\author{Jonas Helsen}
\affiliation{QuSoft $\&$ CWI, Science Park 123
1098 XG  Amsterdam, The Netherlands}
\author{Barbara Terhal}
\affiliation{QuTech, Delft University of Technology, Lorentzweg 1, 2628 CJ Delft, The Netherlands}
\affiliation{EEMCS Faculty, Delft University of Technology, Van Mourik Broekmanweg 6, 2628 XE, Delft, The Netherlands}
\maketitle

\begin{abstract}
\noindent We consider the problem of approximating the ground state energy of a fermionic Hamiltonian using a Gaussian state. 
In sharp contrast to the dense case \cite{SYKScaffidi, HO:approxferm}, we prove that strictly $q$-local {\em sparse} fermionic Hamiltonians have a constant Gaussian approximation ratio; the result holds for any connectivity and interaction strengths.
Sparsity means that each fermion participates in a bounded number of interactions, and strictly $q$-local means that each term involves exactly $q$ fermionic (Majorana) operators.  We extend our proof to give a constant Gaussian approximation ratio for sparse fermionic Hamiltonians with both quartic and quadratic terms.
With additional work, we also prove a constant Gaussian approximation ratio for the so-called sparse SYK model with strictly $4$-local interactions (sparse SYK-4 model).  
In each setting we show that the Gaussian state can be efficiently determined.
Finally, we prove that the $O(n^{-1/2})$ Gaussian approximation ratio for the normal (dense) SYK-$4$ model extends to SYK-$q$ for even $q>4$, with an approximation ratio of $O(n^{1/2 - q/4})$.  
Our results identify non-sparseness as the prime reason that the SYK-4 model can fail to have a constant approximation ratio \cite{SYKScaffidi, HO:approxferm}. 
\end{abstract}

\tableofcontents

\section{Introduction}

Approximating the ground state energy of a local Hamiltonian is a central problem in both physics and computer science. In computer science it plays a key role in complexity theory \cite{Gharibian2015}, while in physics ground states capture the behaviour of systems at low energy. Two common families of Hamiltonians of interest are those defined on collections of qubits and those acting on fermionic degrees of freedom. Fermionic Hamiltonians model various physical systems, such as electrons in condensed matter and quantum chemistry --- prime targets for quantum simulation. Fermions also define a model of quantum computation, equivalent to the one based on qubits \cite{BK:ferm}. Despite its practical and conceptual relevance, the general problem of approximating fermionic ground state energies is currently less well understood than its qubit counterpart.

Some rigorous progress in studying this problem – both for qubits and for fermions – was made from the perspective of optimization. In this subfield of computer science, one of the central tasks is efficiently finding problem solutions that are provably close to optimal \cite{khanna2001approximability}. The closeness is usually quantified by an approximation ratio, i.e. the ratio between the value attained by an algorithm and the optimal value for a given problem. For the classical equivalent of the ground state energy finding – Constraint Satisfaction Problems (CSPs) – such approximation ratios have been extensively studied \cite{GW:semidef}. 

For quantum Hamiltonians, an interesting question is how well the ground state energy can be approximated using ``classical" or ``mean-field'' states. For qubit Hamiltonians the natural choice of classical states are product states, while for fermionic Hamiltonians they are Gaussian states.
Gaussian states play a prominent role in fermionic optimization problems using the mean-field Hartree-Fock method, see e.g. \cite{KC:ferm}, or dynamical mean-field theory via solving impurity problems \cite{BG:impurity}, or the simulation of free fermionic computation \cite{BK:simul,MCT:gaussian}.

Formal guarantees on approximation ratios characterize numerical simulation methods using classical states and outline their limitations compared to quantum computing. For qubit Hamiltonians, it was first proved by Lieb \cite{lieb} (see \cite{BGKT:manybody} for a simplified proof) that there always exists a product state which approximates the ground state energy of a traceless $2$-local qubit Hamiltonian by a factor of $1/9$. Many more results on approximating ground state energies of many-body systems by product states can be found in \cite{BBT:approx, GK:dense, BH:dense-approx, HM:extremal,berga:improved, Parekh19}.
In \cite{BGKT:manybody} it was shown, through the Goemans-Williamson method, that for a 2-local traceless qubit Hamiltonian a product state can always be \emph{efficiently} found with approximation ratio $O(1/\log(n))$ where $n$ is the number of qubits. Ref.~\cite{BGKT:manybody} also considered fermionic Hamiltonians with quadratic ($q=2$) and quartic ($q=4$) fermionic terms.
They left as an open question whether all $4$-local fermionic Hamiltonians have a constant approximation ratio with respect to Gaussian states (a Gaussian approximation ratio).

A surprising counterexample to this conjecture was recently presented in Refs.~\cite{SYKScaffidi, HO:approxferm} --- the family of SYK-$4$ models (Sachdev-Ye-Kitaev models with quartic fermionic interactions, see Definition \ref{SYK_H}). It was shown that with high probability, SYK-$4$ Hamiltonians admits a Gaussian approximation ratio no better than $O(1/\sqrt{n})$ where $n$ is the number of fermionic modes. Contrasting this result to Refs.~\cite{lieb,GK:dense}, it means that
qubit and fermionic ground states strongly differ 
in their approximability by classical states. Moreover, this opens up the question of which fermionic Hamiltonians \emph{do} allow finite Gaussian approximation ratios.

This is the question that we aim to answer here. We do this by considering \emph{sparse} Hamiltonians, i.e. Hamiltonians where each fermionic mode participates in a bounded number of interactions. 
Sparsity holds for many physically relevant Hamiltonians, such as the Fermi-Hubbard model. It also holds for exotic Hamiltonians, such as those determined by constant-degree expander hypergraphs; notably, it does not hold for the SYK model. Sparsity of interactions has been considered in the classical CSP literature.  
It was shown in \cite{alon2006quadratic} that the MaxQP problem has an efficient constant approximation ratio algorithm on graphs of bounded chromatic number, in particular graphs with bounded degree. 
We show that a similar assumption of sparsity is enough to guarantee constant Gaussian approximation ratios for $4$-local and strictly $q$-local Hamiltonians. 
Moreover, we show that a constant Gaussian approximation ratio can be achieved for the \emph{sparse} SYK-$4$ model~\cite{sparseSYK} (which has a logarithmically growing interaction participation and is thus not sparse by our definition). Finally, we consider in more detail the optimal approximation ratio for the \emph{dense} SYK-$q$ model for $q>4$ (thus extending the work of \cite{HO:approxferm}). We show that the shortfall of Gaussian states is even more pronounced in this setting. 

To avoid confusion, we note that instead of the \textit{ground state} energy, existing works often consider approximating the \textit{maximal} eigenvalue of the Hamiltonian $\lambda_{\rm{max}}(H)$. These two optimization problems are equivalent if the family of Hamiltonians considered is invariant under a change of sign (e.g. traceless $q$-local Hamiltonians). For mathematical convenience and consistency with the literature, in the rest of the text, we will also be formulating our results in terms of approximating $\lambda_{\rm{max}}(H)$.

\begin{figure*}
    \centering
    \includegraphics[width=1.03\linewidth]{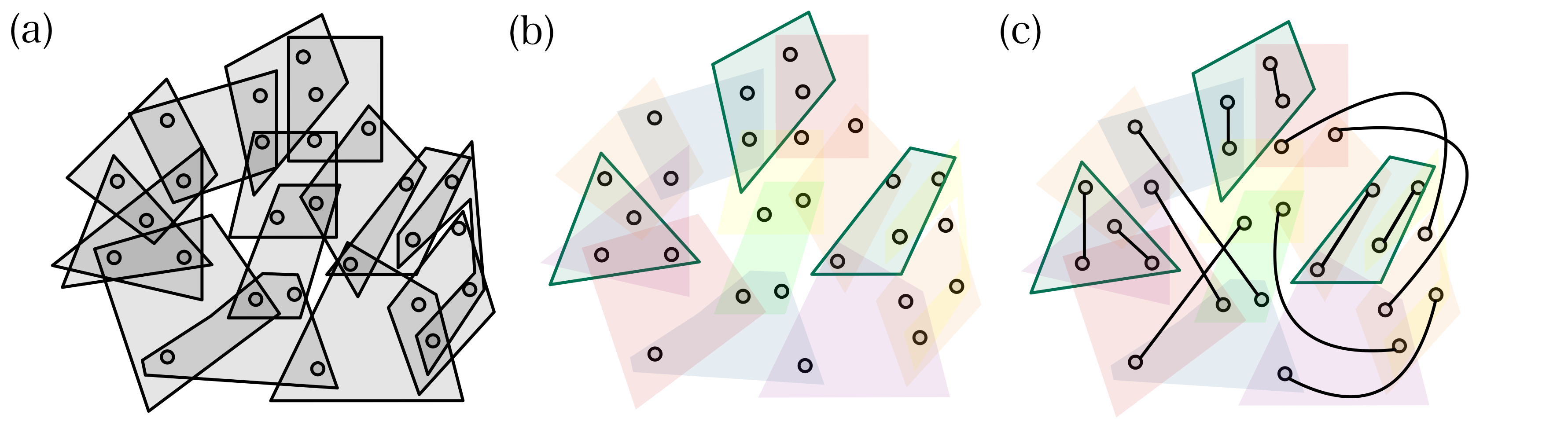}
    \caption{Illustrating the key idea of the proof of Theorem \ref{thm:sparse}. An example of a strictly $4$-local Hamiltonian is given in (a), vertices and faces representing Majorana operators and their interactions. The Hamiltonian is split into sets of terms -- different colors in (b) -- well separated from each other inside each set (so-called diffuse sets, see Definition~\ref{def:diffuse}).
    The next step is to match all Majorana operators, i.e., split the vertices into disjoint pairs, each connected by an edge (see panel (c)). We separately match the support of each term in one \textit{targeted} set of terms (the color highlighted in (b) and (c)). The remaining vertices are matched in such a way that no two vertices connected by an edge belong to the same term. The Gaussian state is then created from the resulting matching, with only terms from the \textit{targeted} set contributing to the energy. By optimizing the choice of the targeted set, a finite approximation ratio can be guaranteed.}
    \label{fig:thm_sparse}
\end{figure*}

\section{Statement of results}

\subsection{Preliminaries}

Before surveying our results, we introduce the basic setup of fermionic Hamiltonians and $q$-locality. This subsection also defines the SYK-$q$ model and spells out the previous result of a vanishing Gaussian approximation ratio for SYK-$4$. 

We consider a system of $2n$ traceless Majorana fermion operators $c_i$, $i=1,\ldots,2n$ with $c_i^2=\mathbb{I}, c_i^{\dagger}=c_i$, forming a Clifford algebra, i.e., $\{c_{j},c_{k}\} = 2\delta_{j,k}\mathbb{I}$ and representing $n$ fermionic modes. 
We denote as $I$ an ordered subset $I=\{i_1, i_2, .. i_q\} \subseteq [2n]\equiv \{1, \ldots, 2n\}$ where $i_1 < i_2 < \ldots i_q$ with $q$ even. 
We denote $C_I$ as the Hermitian Majorana monomial
\begin{align}
    C_I\equiv i^{q/2} c_{i_1}..c_{i_q},
    \label{eq:CIdef}
\end{align}
and one can verify that
\[
C_I^2=\mathbb{I}.
\]

We can think about a subset $I$ as corresponding to a term or interaction in a Hamiltonian. Indeed, it is natural to impose some form of locality:

\begin{definition}[$q$-local fermionic Hamiltonian]
Let $H$ be a fermionic Hamiltonian on $2n$ Majorana operators. We say that $H$ is $q$-local if $H$ is a sum of Hermitian traceless terms $C_I$ of weight at most $q$, i.e. each term is proportional to a product of at most $q$ operators $c_i$.
$H$ is said to be strictly $q$-local when all terms have exactly weight $q$.
\end{definition}

A local traceless fermionic Hamiltonian $H=\sum_{I \in {\mathcal{I}}} J_I C_I$ is thus characterized by an interaction set $\mathcal{I}$ and the coefficients $J_I \in \mathbb{R}$. The maximum eigenvalue of $H$ is denoted by $\lambda_{\max}(H) := \max_{\rho}{\rm Tr} (H \rho)$.
Sometimes we will refer to a collection of sets $I$ denoted as $\mathcal{I}=\{I_1, I_2, \ldots\}$. The support of $\mathcal{I}$ is defined as ${\rm Sup}(\mathcal{I})=\cup_i I_i$ and $\mathcal{I}' \subseteq \mathcal{I}$ implies that the sets in $\mathcal{I}'$ are also sets in $\mathcal{I}$.

\begin{definition}[SYK-$q$ Model]
A $q$-local (with $q$ even) SYK model on $2n$ Majoranas is defined as a family of Hamiltonians
\begin{equation}
    H = \binom{2n}{q}^{-1/2}  \sum_{I\subseteq [2n], |I|=q}J_{I}\: C_{I},
\label{SYK_H}
\end{equation}
where each $J_{I}$ is a Gaussian random variable (i.e., with zero mean and unit variance) and each $C_{I}$ is the product of the $q$ distinct Majorana operators as in Eq.~\eqref{eq:CIdef}. We normalize the model in expectation, i.e., $\mathbb{E}\left({\rm Tr}(H^2)\right)=\binom{2n}{q}^{-1}\sum_{I\subseteq [2n],|I|=q}\mathbb{E}(J_{I}^{2} ) \,{\rm Tr}(\mathbb{I}) = 2^n$.
\label{def:SYK}
\end{definition}

In \cite{SYKScaffidi} it was shown that with high probability (over the draw of $J_I$s) for the SYK-4
model, one has 
\[
\max_{\rho\;{\rm Gaussian}}{\rm Tr}(H \rho)= O(1).
\] 
In order to thus provide a counterexample to a constant Gaussian approximation ratio, one needs to prove a lower bound on $\lambda_{\rm max}(H)$ for the SYK-4 model, which holds with high probability, which was done in \cite{HO:approxferm}:

\begin{theorem}\cite{HO:approxferm}
There is a ${\rm poly}(n)$-time quantum algorithm that, given any SYK-4 Hamiltonian $H$, returns a quantum state $\rho$. With probability $1-\exp\big( -\Omega(n) \big)$ (over the draw of the $J_I$s), this state $\rho$ has ${\rm Tr}(H \rho)=\Omega(\sqrt{n})$.
\label{thm:HO}
\end{theorem}

\subsection{Sparse fermionic Hamiltonians}
Key to our work is the notion of a sparse Hamiltonian.
\begin{definition}
Let $H$ be a local traceless fermionic Hamiltonian of $2n$ Majorana operators. We say that $H$ is $k$-sparse, for an integer $k$, if no Majorana operator $c_i$ occurs in more than $k$ terms of the Hamiltonian.
\end{definition}
Using graph theoretic terminology, one may say that interactions in a $k$-sparse Hamiltonian form a hypergraph of bounded degree $k$. This condition allows us to efficiently find Gaussian states with constant approximation ratio. We have the following theorem, which is the main result of our work:

\begin{theorem}\label{thm:sparse}
Let $H$ be a traceless fermionic Hamiltonian on $2n$ Majorana operators with maximal eigenvalue $\lambda_{\max} (H)$. If $H$ is $k$-sparse and strictly $q$-local and $n>(q^2-1)k$, a Gaussian state $\rho$ can be efficiently constructed such that
\begin{equation}
    \frac{{\rm Tr}(H\rho)}{\lambda_{\rm max} (H)} \geq \frac{1}{Q},
\end{equation}
for $Q=q(q-1)(k-1)^2+q(k-1)+2$.
\end{theorem}
The proof of this theorem is given in Section~\ref{sec:sparse}; its basic idea is explained in Figure~\ref{fig:thm_sparse}.

We note that this proof only holds for Hamiltonians with terms of \emph{exactly} weight $q$. Typical physical Hamiltonians, however, have quadratic (kinetic energy of the electrons) and quartic terms (potential energy due to Coulomb interaction). Fortunately, we can also show that in the $q=4$ case we can include $q=2$ terms. For this we use a trick from \cite{BGKT:manybody} to lift such a $4$-local Hamiltonian to a strictly $4$-local Hamiltonian. This trick makes the Hamiltonian non-sparse. However, we show in Section~\ref{sec:sparse_sub_4} that, in this special case, we can circumvent the non-sparseness of the Hamiltonian and achieve a constant Gaussian approximation ratio. 

\begin{theorem}\label{thm:sparse_sub_4}
Let $H$ be a traceless fermionic Hamiltonian on $2n$ Majorana operators with maximal eigenvalue $\lambda_{\max} (H)$. If $H$ is $k$-sparse with terms of weight $2$ and $4$ and $2n>15k$, a Gaussian state $\rho$ can be efficiently constructed, such that
\begin{equation}
    \frac{{\rm Tr}(H\rho)}{\lambda_{\rm max} (H)} \geq \frac{1}{2Q}
\end{equation}
for $Q=12(k-1)^2+4(k-1)+2$.
\end{theorem}

\subsection{The sparse \texorpdfstring{$q=4$}{q=4} SYK model}\label{subsec:SYK}
In view of Theorem \ref{thm:sparse} it is worth revisiting the lack of a constant Gaussian approximation for the SYK model. The SYK-$q$ model in Definition \ref{def:SYK} is extremely non-sparse, in the sense that every Majorana operator occurs together with all other Majorana operators. This makes the SYK model somewhat unphysical, and several \emph{sparse} versions of the model have been considered \cite{sparseSYK, SYK-sparse}. Such sparse models intend to produce the same (low energy) physics, while being easier to simulate on both quantum and classical computers (see sections III and V in \cite{sparseSYK}). The sparse SYK model is generated by including terms by a Bernoulli trial with a certain probability $p$ tuned such that the expected sparsity is bounded:
\begin{definition}
The sparse \SYKf or SSYK-4 model on $2n$ Majorana operators with expected sparsity $k=O(1)$ is given as
\begin{equation}
    H = \frac{1}{\sqrt{2k n}}\sum_{I\subset [2n], |I|=4}X_I J_I C_I
    \label{eq:sparseSYKdef}
\end{equation}
where the $X_I$ are i.i.d. Bernoulli random variables with $p=\mathrm{Pr}(X_I=1) =\frac{k}{\binom{2n-1}{3}}$ and the $J_I$ are i.i.d. Gaussian random variables with mean $0$ and variance $1$.
\label{def:SSYK}
\end{definition}
Unlike the full SYK model with $\binom{2n}{4}$ terms in $H$, the sparse SYK model has a number of terms $\sim n$ in expectation. Note that the \SSSYKf model is only $k$-sparse in expectation, and with high probability there is a Majorana operator with degree $\Omega\big(\frac{\log(n)}{\log\log(n)}\big)$ (the degree distribution follows that of an Erd\H{o}s-Renyi hypergraph. See Theorem 3.4 in \cite{frieze2016introduction} for a proof of the statement for Erd\H{o}s-Renyi graphs. The hypergraph version follows by the same logic). This means that Theorem \ref{thm:sparse} does not directly apply. However, one can show, through a truncation argument, that almost all instantiations of \SSSYKf can be \emph{sparsified}, giving rise to a constant approximation ratio result that holds with high probability. 

\begin{theorem}\label{thm:sparse_SYK}
Let $H$ be a \SSSYKf Hamiltonian in Eq.~\eqref{eq:sparseSYKdef} with expected degree $k=O(1)$, such that $n>120(k+1)$. With probability at least $1-4\exp\left[-\frac{e^{-16(k+1)}k^3}{64(8k+7)} n\right]$, a Gaussian state $\rho$ can be efficiently constructed such that
\begin{equation}
    \frac{{\rm Tr}(H \rho)}{\lambda_{\max} (H)} \geq \frac{1}{Q},
\end{equation}
where $Q=1236 + 2752 k + 1536 k^2$.
\end{theorem}
Thus we arrive at the surprising conclusion that the \SSSYKf model has a constant Gaussian approximation ratio, while the dense \SYKf model does not --- even though \SSSYKf has similar physical properties as SYK-4.

\subsection{Higher-\texorpdfstring{$q$}{q} SYK models}

We investigate what Gaussian approximation ratios can be achieved for the \emph{dense} SYK model of even weight $q>4$, as this was left as an open question in \cite{HO:approxferm}. We establish an upper bound on the largest Gaussian expectation value of SYK-$q$, which behaves rather dramatically for $q>4$. We prove the following Lemma employing a method similar to the one used in \cite{SYKScaffidi}.

\begin{lemma}
Let $H$ be the \emph{dense} SYK-$q$ Hamiltonian (with even $q\geq 4$ and $q=O(1)$). With probability at least $1-\exp \big( -\Omega(n) \big)$ over the draw of SYK-$q$ Hamiltonians, the expectation value of every Gaussian state $\rho$ is bounded, more precisely:
\begin{equation}
    \max_{\rho \: {\rm Gaussian}} {\rm Tr} (H 
    \rho)  = O\big( n^{1-q/4} \big).
\end{equation}
\label{lemma:SYKexpvaluegaussians}
\end{lemma}

This Lemma is proved in Section \ref{section:gaussupperbound}. Our second result establishes a lower bound on the largest eigenvalue for SYK-$q$, essentially generalizing what was established in \cite{HO:approxferm} for $q=4$. We prove the following Lemma (its proof can be found in Section \ref{section:gaussupperbound}):

\begin{lemma}
Let $H$ be the \emph{dense} SYK-$q$ Hamiltonian with even $q\geq 4$ (and $q=O(1)$). With probability at least $1-\exp\big( -\Omega(n) \big)$ over the draw of SYK-$q$ Hamiltonians, $\lambda_{\max}(H) = \Omega(\sqrt{n})$.
\label{lemma:maxeiglowbound}
\end{lemma}

As an immediate consequence of the previous results, we see that the Gaussian approximation ratio of the dense SYK-$q$ model can be no better than $O\big(n^{1/2-q/4}\big)$: 
\begin{theorem}
Let $H$ be the \emph{dense} SYK-$q$ Hamiltonian (with even $q\geq 4$ and $q=O(1)$). With probability at least $1-\exp\big( -\Omega(n) \big)$ over the draw of SYK-$q$ Hamiltonians, we have
\begin{equation}
    \max_{\rho \: {\rm Gaussian}} \frac{{\rm Tr}(H\rho)}{\lambda_{\max}(H)} = O\big( n^{1/2-q/4} \big).
\end{equation}
\label{theorem:mainSYKqtheorem}
\end{theorem}

\begin{proof}
Theorem \ref{theorem:mainSYKqtheorem} follows from combining Lemma \ref{lemma:SYKexpvaluegaussians} and Lemma \ref{lemma:maxeiglowbound} and applying the union bound. 
\end{proof}

\section{Discussion}

The goal of this section is to place our results in a broader context and mention a few open questions. 

First, let us discuss the relation between this work and the fermion-to-qubit mapping methods. As was shown in \cite{BK:ferm}, one can map a {\em sparse} $O(1)$-local fermionic Hamiltonian onto a sparse $O(1)$-local qubit Hamiltonian (BK-superfast encoding). However, for this mapping one needs to enforce parity checks which are in general nonlocal; therefore, we cannot obtain our Theorem \ref{thm:sparse} in this way. There is also an additional obstacle: using the BK-superfast encoding, an approximating product state for the qubit Hamiltonian
does not necessarily map back to a Gaussian fermionic state.

Ref.~\cite{BK:ferm} also showed that one can map a general local fermionic Hamiltonian (like a SYK model) onto a qubit Hamiltonian with terms which are $O(\log n)$-local. Such qubit Hamiltonian is generally not expected to have a constant approximation ratio by a product state due to its $n$-dependent locality. In fact, one can easily prove that a dense model like the SYK model can only be mapped onto a qubit Hamiltonian which is $\Omega(\log n)$-local. We give the argument in Appendix \ref{app:lemAC}.

These observations suggest that approximation ratios by classical states such as Gaussian states or product states are likely to be affected by sparsity in the case of fermions, which is consistent with our new results.

Another question which is raised by our work and that of \cite{HO:approxferm} and \cite{BH:dense-approx}, is
whether studying fermionic Hamiltonians can lead to new insights into the possibility of a quantum PCP theorem \cite{AV:QPCP}. In this context it is important to mention that, besides the lower bound in Theorem \ref{thm:HO},  Ref.~\cite{HO:approxferm} also determined an upper bound on $\lambda_{\rm max}$ of the SYK-$4$ model showing that with high probability $\lambda_{\rm max}=\Theta(\sqrt{n})$. This shows that the SYK-$4$ model is extremely frustrated: the maximal average expected energy per term, {\em the energy density}, is only $\Theta(n^{-3/2})$. In contrast, our results for the sparse SYK model (see Lemma \ref{lem:H_SSYK}) show that the maximal average expected energy per term is $\Omega(1)$, which is the more `natural' physical scaling. A simple fermionic toy model in which the maximal average energy per term decreases is a model in which an extensive set of Majorana operators is mutually anti-commuting, see Lemma \ref{lem:ACspectrum} in Appendix \ref{app:lemAC}.
 The presence of many such fully-anticommuting sets in the SYK model can be seen as one of the intuitive reasons why the maximal energy density achieved is so low. 
 
 For $k$-local qubit Hamiltonians researchers have looked at the hardness of approximating the maximal energy density with constant error $\epsilon$: showing that this problem is QMA-complete would prove the quantum PCP theorem. For dense (non-sparse) $k$-local qubit Hamiltonians, it was proved in \cite{BH:dense-approx} (Theorem 13) that there is a polynomial-time classical algorithm to approximate the maximal energy density, using product state approximations. Ref.~\cite{BCGW} generalized this result and formulated an efficient classical algorithm which approximately estimates the free energy of a 2-local dense qubit Hamiltonian. 

One can similarly ask the question of approximating the maximal energy density for dense $q$-local fermionic Hamiltonians. Observe that the question is moot if the maximal energy density decreases as a function of $n$ (as in the SYK model), since for large enough $n$ (depending on $\epsilon$) the classical algorithm could always output 0 and make an error less than $\epsilon$. However, other dense $O(1)$-local fermionic Hamiltonians could exist for which this question is nontrivial and not already covered by the dense qubit case.
  
There are further open directions that are more practically-oriented. One of these is achieving finite approximation ratios for at least some classes of \textit{non-sparse} fermionic Hamiltonians (e.g., quantum chemistry or lattice systems with long-range Coulomb interactions). Furthermore, in most applications, one is interested in obtaining approximation ratios as close to 1 as possible. Although for most systems of interest one cannot expect ratios that are $\epsilon$-close to $1$, our theoretical lower bounds could still be vastly improved. For instance, for sparse SYK with $k=10$, the guaranteed ratio is only $\simeq 5\times 10^{-6}$ (cf. Theorem~\ref{thm:sparse_SYK}). This can be contrasted to the Hartree-Fock applications to quantum chemistry systems, which usually achieve approximation ratios of $>0.9$.  Improving our results to derive more realistic lower bounds could be of great value; some possible approaches are as follows.

One option is to extend the interaction subsets targeted by the constructed Gaussian state beyond the diffuse subsets considered here. If the overlapping interactions in the problem Hamiltonian are not prone to frustration, including them in the targeted set may dramatically increase the approximation ratio. The proof of Theorem~\ref{thm:sparse_sub_4} (Section~\ref{sec:sparse_sub_4}) is a special case of this approach, with the constructed Gaussian state targeting multiple overlapping terms at the same time. 

Another option for improvement is to minimize the contribution from frustration terms instead of avoiding frustration altogether. This could both improve the eventual approximation ratio by targeting a larger pool of interactions, as well as allowing to mitigate the issue of non-sparsity. An example of this approach is the proof of Theorem~\ref{thm:sparse_SYK} (Section~\ref{sec:sparse_SYK}), where the contributions from the non-sparse part of the Hamiltonian are shown to be small compared to the energy achieved by the Gaussian state. 

As a third option, one can modify the basis of fermionic modes so that non-sparsity and frustration in the Hamiltonian are minimized. In the simplest case of $q=2$, such a basis rotation can always turn all interactions into a diffuse set (simply by diagonalizing the Hamiltonian). A similar improvement may be possible for some classes of $q$-local Hamiltonians with $q\geq 4$. 

Developing these and other directions for efficient Gaussian ground state approximation are interesting possibilities for future research.

Finally, it would be interesting to provide a non-random family of fermionic Hamiltonians without a constant approximation ratio with respect to Gaussian states.

\section{Background on Gaussian states}
\label{sec:back}

In this section, we first provide some background and definitions that will be used throughout the remainder of this text.

\subsection{Gaussian states}

We define the class of fermionic Gaussian states, which are ground states and thermal states of non-interacting, quadratic ($q=2$), fermionic Hamiltonians, and give some of their useful properties.

We first note that any transformation by a real orthogonal matrix $R\in SO(2n)$, i.e.,
\begin{equation}
    \tilde{c}_{i} = \sum_{j}R_{ij}c_{j},
    \label{eq:Rtrafo}
\end{equation}
preserves the properties of Majorana operators and hence gives rise to a new set of $2n$ Majorana operators $\{\tilde{c}_{j}\}_{j=1}^{2n}$.

\begin{definition}{\textbf{Fermionic Gaussian state.}}
Given $2n$ Majorana operators denoted by $c_{1},\ldots,c_{2n}$. A fermionic Gaussian state is a -- generally mixed -- state of the form
\begin{equation}
    \rho =\frac{1}{2^n} \exp\Big( -i\sum_{i\neq j}\beta_{ij}c_{i}c_{j} \Big),
\end{equation}
where $(\beta_{ij})_{i,j=1}^{2n}$ is a real anti-symmetric matrix and the normalization is such that ${\rm Tr} \,
\rho=1$.
\end{definition}
\noindent
Fermionic Gaussian states have a number of useful properties, which we list here for future use.
\begin{enumerate}
    \item The matrix $\beta$ can be block-diagonalized by a real orthogonal matrix $R\in SO(2n)$ such that
    \begin{equation}
        \beta = R\: \bigoplus_{j=1}^{n} 
        \begin{pmatrix}
        0 & b_{j} \\
        -b_{j} & 0
        \end{pmatrix}
        \:R^T,
    \end{equation}
    with $b_j \geq 0$. Therefore, $\rho$ can be brought to the following standard form
    \begin{equation}
        \rho = \frac{1}{2^{n}}\prod_{j=1}^{n}\Big( \mathbb{I}+i\lambda_{j}\tilde{c}_{2j-1}\tilde{c}_{2j} \Big),
        \label{eq:rho}
     \end{equation}
    where $\tilde{c}_i = \sum_j R_{ij} c_j$ and $\lambda_{j} = \tanh(2b_{j}) \in [-1,+1]$.
    \item Each fermionic Gaussian state can be associated with a $2n\times 2n$ correlation matrix $\Gamma$, with
    \begin{equation}
        \Gamma_{ij} = \frac{i}{2}\text{Tr}\big( \rho [c_{i},c_{j}] \big).
    \end{equation}
    $\Gamma$ is a real anti-symmetric matrix and hence there is a real orthogonal matrix $R\in SO(2n)$ such that 
    \begin{equation}
        \Gamma = R\: \bigoplus_{j=1}^{n} 
        \begin{pmatrix}
        0 & \lambda_{j} \\
        -\lambda_{j} & 0
        \end{pmatrix}
        \:R^T,
    \end{equation}
    where the $\lambda_j$ are in Eq.~\eqref{eq:rho}.
    \item For pure fermionic Gaussian states, $\lambda_{j}\in\{-1,+1\}$ and hence for pure Gaussian states $\Gamma^{T}\Gamma = \mathbb{I}$. For mixed fermionic Gaussian states, $\Gamma^{T}\Gamma\leq\mathbb{I}$.
    \item The Pfaffian of a $2k \times 2k$ anti-symmetric matrix $A$ is defined as 
    \[
    {\rm Pf}(A)=\frac{1}{2^k k!}\sum_{\pi \in S_{2k}}{\rm sign}(\pi)\Pi_{i=1}^k A_{\pi(2i-1),\pi(2i)}.
    \]
    Alternatively, we can see the Pfaffian as a sum over perfect matchings in a graph of $2k$ vertices where an edge $(i < j)$ has weight $A_{ij}$ and each matching contributes the products of these weights to the sum.
    For a Gaussian state with correlation matrix $\Gamma$, one has for even $|I|$:
    \begin{align}
        {\rm Tr}(C_I \rho)={\rm Pf}(\Gamma_I),
        \label{eq:wick}
    \end{align}
    where $\Gamma_I$ is the $|I| \times |I|$ submatrix of $\Gamma$ restricted to rows and columns in the ordered set $I$.
\end{enumerate}

A special class of a pure Gaussian states is given by a perfect matching $M$ of Majorana operators. Such matching $M$ is specified by $n$ disjoint pairs $(m_1, m_2)$ with $m_1 < m_2$. For each pair we have a coefficient $\lambda_{(m_1,m_2)}=\pm 1$, together forming the $n$-dimensional vector $\vec{\lambda}$. The class of states are of the form
\begin{equation}
    \rho(M,\vec{\lambda})=\frac{1}{2^{n}}\Pi_{(m_1,m_2)\in M} (\mathbb{I}+i \lambda_{(m_1,m_2)} c_{m_1} c_{m_2}).
\end{equation}

It is useful to introduce a notion of {\em consistency} between this class of Gaussian states specified by a matching $M$ and an interaction subset $I$. 
\begin{definition}
An (even) interaction subset $I \subseteq [2n]$ and a perfect matching $M$ on $[2n]$ are called consistent if $M$ contains a perfect matching of the elements of $I$. Given a set of interactions $\mathcal{I}$, we say that $M$ is consistent (resp. inconsistent) with $\mathcal{I}$ if $M$ is consistent (resp. inconsistent) with \textbf{each} interaction in $\mathcal{I}$.
\end{definition}

The following Lemma is straightforward

\begin{lemma}
\label{lem:consistent_expectation}
\begin{enumerate} Consider a matching $M$ and an interaction $I=\{i_1,i_2,..i_q\}$.
    \item If $M$ is consistent with interaction $I$, let the perfect matching on the subset $I$ be given by pairs $(i_{\pi(2l-1)},i_{\pi(2l)})$ for $l=1,\ldots,q/2$ and a permutation $\pi \in S_q$ where $i_{\pi(2l-1)} < i_{\pi(2l)}$. Then, the following holds: 
    \begin{equation*}
        {\rm Tr} (C_I \rho(M, \vec{\lambda}))  ={\rm sign}(\pi) \hspace{-1em}\prod_{l\in\{1,\ldots,q/2\}} \hspace{-1em}\lambda_{(i_{\pi(2l-1)}, i_{\pi(2l)})}.
    \end{equation*}
    where ${\rm sign}(\pi)=\pm 1$.  
    \item If $M$ is inconsistent with $I$, then
    \begin{equation}
        {\rm Tr} (C_I \rho(M, \vec{\lambda}))  =0.
             \end{equation}
\end{enumerate}
\end{lemma}
\begin{proof}
In order for the trace to be nonzero, one needs to exactly match the Majorana operators in $C_I$ with some in the expansion of $\rho(M,\vec{\lambda})$ since 
$\mathrm{Tr}(C_{I'})=0$ for any non-empty subset $I'$. If $M$ is inconsistent, there is no term in the expansion of $\rho$ which precisely matches $C_I$, so the expectation vanishes. If $M$ is consistent, we have
\begin{align}
    &{\rm Tr} (C_I \rho(M, \vec{\lambda}))\notag\\
    &\hspace{3em}=
    \frac{1}{2^n}\mathrm{Tr}\big(C_I \Pi_{\scriptscriptstyle(m_1,m_2)\in M} (\mathbb{I}+i \lambda_{\scriptscriptstyle(m_1,m_2)} c_{m_1} c_{m_2})\big) \notag \\
    &\hspace{3em}=\frac{1}{2^n}\mathrm{Tr}\big(C_I \Pi_{\scriptscriptstyle(m_1,m_2)\in M, (m_1,m_2)\cap I\neq\emptyset)}\notag\\
    &\hspace{13em}\times(i \lambda_{\scriptscriptstyle(m_1,m_2)} c_{m_1} c_{m_2})\big)\notag\\
    &\hspace{3em}={\rm sign}(\pi) \prod_{l\in[1,..,q/2]} \lambda_{(i_{\pi(2l-1)}, i_{\pi(2l)})}.
\end{align}
Here we have used that one can first reorder $C_I$ such that the pairs in the perfect matching are adjacent, i.e. $C_I={\rm sign}(\pi) i^{q/2} c_{i_{\pi(1)}} c_{i_\pi(2)}\ldots c_{i_{\pi(q)}}$, then one can commute through each pair to its matching pair in $\rho$ and use $(c_i c_j)^2=-\mathbb{I}$, $i^q=(-1)^{q/2}$ and $\mathrm{tr}(\mathbb{I})=2^n$. 
\end{proof}

\section{Approximation ratios for sparse fermionic Hamiltonians}\label{sec:sparse}

In this section we prove Theorem \ref{thm:sparse}. We begin by setting up needed definitions and stating several technical Lemmas (which are proved in the Appendices). 

The key auxiliary notion in the proof of Theorem \ref{thm:sparse} is that of a \textit{diffuse} subset of Hamiltonian terms. Intuitively, the terms in a diffuse subset are well separated from each other while covering only a limited part of the system. This idea is formalized as follows:

\begin{definition}
\label{def:diffuse}
Consider a set of $q$-local interactions ${\cal I}$ on $2n$ Majorana operators. A subset of these interactions ${\cal I}'\subset {\cal I}$ is \emph{diffuse} with respect to ${\cal I}$, if the following three conditions apply:
\begin{enumerate}
    \item $\forall I_1, I_2 \in {\cal I}'$, $I_1$ and $I_2$ don't share any Majorana operators, i.e. $I_1\cap I_2 = \emptyset$.
    \item $\forall I_1, I_2 \in {\cal I}'$, there exists no $I_3 \in {\cal I}$ which shares Majorana operators with both $I_1$ and $I_2$ (if $I_3 \cap I_1\neq \emptyset$ then $I_3 \cap I_2=\emptyset$ and vice versa).
    \item The size of support of ${\cal I}'$, i.e. $|\mathrm{Sup}({\cal I}')|$, is smaller than $\frac{2qn}{q+1}$.
\end{enumerate}

\end{definition}

In the setting of Theorem~\ref{thm:sparse}, diffuse sets of terms appear naturally due to the following Lemma.

\begin{lemma}
\label{lem:diffuse_splitting}
Consider a $k$-sparse strictly $q$-local fermionic Hamiltonian $H$ on $2n$ Majoranas. The interaction set ${\cal I}$ of $H$ can be split into $Q$ disjoint subsets ${\cal I}_\alpha$ ($\alpha \in [Q])$ all of which are diffuse with respect to ${\cal I}$ such that
\begin{align}
    {\cal I}=\bigcup^{Q}_{\alpha=1} {\cal I}_\alpha.
\end{align}
The parameter $Q$ is given as $Q=q(q-1)(k-1)^2+q(k-1)+2$ and does not depend on $n$. The construction of this splitting can be done efficiently, in time $\mathrm{poly}(n)$.
\end{lemma}

Lemma \ref{lem:diffuse_splitting} is a special case of Lemma \ref{lem:diffuse_splitting_general}, which is proven in Appendix \ref{sec:diffuse_splitting_general}. The proof relies on a combinatorial argument on a graph that takes Hamiltonian terms as vertices and connects them with an edge if the pair violates conditions 1 or 2 of Definition~\ref{def:diffuse}. By the sparsity assumption, this graph has an efficiently constructable coloring with a bounded number of colors, from which the split ${\cal I}=\bigcup^{Q}_{\alpha=1} {\cal I}_\alpha$ can be constructed. 

The usefulness of diffuse sets comes from Lemma \ref{lem:diffuse_matching_general}, see its proof in Appendix \ref{sec:diffuse_matching_general}. Here we state its corollary, relevant to proving Theorem \ref{thm:sparse}:

\begin{lemma}
\label{lem:diffuse_matching}
Let the interaction set $\cal{I}'$ be diffuse w.r.t. $\cal{I}\supset \cal{I}'$ ($\cal{I}'$ and $\cal{I}$ are strictly $q$-local and $k$-sparse). 
If $n> (q^2-1) k$, one can efficiently construct a matching $M$ of the set $[2n]$ that is consistent with each interaction in $\cal{I}'$ and inconsistent with each interaction in ${\cal I}\backslash{\cal{I}}'$. 
\end{lemma}

With matchings introduced above, one can construct useful Gaussian states. The tool to do so is given by the following statement:

\begin{lemma}
\label{lem:matching_gaussian}
Let $H=\sum_{I\in {\cal I}} J_I C_I$ be strictly $q$-local and ${\cal I}'$ be a diffuse subset of ${\cal I}$. Let $M$ be a matching of $[2n]$ as guaranteed by Lemma \ref{lem:diffuse_matching}. One can efficiently construct a Gaussian state $\rho_{{\cal I}'}$ with the property:
\begin{equation}
    {\rm Tr}(H \rho_{{\cal I}'})=\sum_{I\in {\cal I}'} |J_I|.
\end{equation}
\end{lemma}

Lemma \ref{lem:matching_gaussian} is a specific case of a slightly more general Lemma \ref{lem:matching_gaussian_general}, which is stated and proven in Appendix \ref{sec:matching_gaussian_general}. We denote
\begin{equation}
\mathbf{J}\left({\cal{I}'}\right)\equiv\sum_{I\in{\cal I}'} |J_I|.
\end{equation}
As shown below, Theorem \ref{thm:sparse} can be proven by constructing a diffuse ${\cal I}'\subset{\cal I}$ and a corresponding Gaussian state $\rho_{{\cal I}'}$ with large enough ${\rm Tr} (H \rho_{{\cal I}'})=\mathbf{J}\left({\cal{I}'}\right)$.

\begin{theorem*}[Repetition of Theorem \ref{thm:sparse}]
Let $H$ be a traceless fermionic Hamiltonian on $2n$ Majoranas with maximal eigenvalue $\lambda_{\max} (H)$. If $H$ is $k$-sparse and strictly $q$-local and $n>(q^2-1)k$, a Gaussian state $\rho$ can be efficiently constructed, such that
\begin{equation}
    \frac{{\rm Tr}(H\rho)}{\lambda_{\rm max} (H)} \geq \frac{1}{Q},
\end{equation}
for $Q=q(q-1)(k-1)^2+q(k-1)+2$.
\end{theorem*}

\begin{proof}
For a Hamiltonian $H=\sum_{I\in {\cal I}} J_I C_I$, we construct the splitting of ${\cal I}$ into diffuse subsets ${\cal I}=\cup_\alpha {\cal I}_\alpha $ as guaranteed by Lemma \ref{lem:diffuse_splitting}. Next, find $\alpha=\mathrm{argmax}_{\alpha'} \mathbf{J}({\cal I}_{\alpha'})$; since $Q$ in Lemma \ref{lem:diffuse_splitting} is constant, $\alpha$ can be found efficiently. Next, use Lemma \ref{lem:diffuse_matching} to construct a matching $M({\cal I}_\alpha)$ (the condition $n>(q^2-1)k$ is satisfied by assumptions of Theorem \ref{thm:sparse}). 
Since ${\cal I}_\alpha$ is diffuse with respect to ${\cal I}$, the Gaussian state $\rho_{{\cal I}_\alpha}$ can be efficiently constructed from $M({\cal I}_\alpha)$ via Lemma \ref{lem:matching_gaussian}. Using ${\rm Tr} (H \rho_{{\cal I}_\alpha})=\mathbf{J}({\cal I}_\alpha)$, the following inequality can be obtained for the resulting approximation ratio:
\begin{align}
\label{eq:sparse_ratio}
\frac{{\rm Tr}(H \rho_{{\cal I}_\alpha})}{\lambda_{\rm max}(H)}\geq \frac{\mathbf{J}({\cal I}_\alpha)}{  \sum_{\alpha'} \mathbf{J}({\cal I}_{\alpha'})} \geq \frac{1}{  Q} .
\end{align} 
For the first inequality, note that $\lambda_{\rm max}(H)\leq \sum_{I\in{\cal I}} |J_I| =\sum_{\alpha} \mathbf{J}({\cal I}_\alpha)$. The second inequality comes from a pigeonhole-type argument: if $\mathbf{J}({\cal I}_\alpha) = \max_{\alpha'} \mathbf{J}({\cal I}_{\alpha'}) $, it directly follows that $\mathbf{J}({\cal I}_\alpha) \geq \frac{1}{Q} \sum_{\alpha'} \mathbf{J}({\cal I}_{\alpha'})$. Inequality \eqref{eq:sparse_ratio} concludes the proof, as it asserts the approximation ratio bound claimed in the Theorem.
\end{proof}

\section{Sparse Hamiltonians with terms of weight \texorpdfstring{$2$}{2} and \texorpdfstring{$4$}{4}}
\label{sec:sparse_sub_4}

In this section we prove Theorem \ref{thm:sparse_sub_4}. We will again need to use the concept of diffuse subsets in Definition \ref{def:diffuse}. The proof of Theorem \ref{thm:sparse_sub_4} is similar in its basic idea to that of Theorem \ref{thm:sparse}. The main obstacle in this case is the presence of terms of different weight, which does not allow one to use Lemmas \ref{lem:diffuse_splitting}-\ref{lem:matching_gaussian} directly. This can be resolved by a slightly more elaborate construction and applying the more general Lemmas \ref{lem:diffuse_splitting_general}-\ref{lem:matching_gaussian_general} which are proved in the Appendices and Lemmas \ref{lem:diffuse_splitting}-\ref{lem:matching_gaussian} directly follow as special cases.

\begin{lemma}[Generalization of Lemma \ref{lem:diffuse_splitting}]
\label{lem:diffuse_splitting_general}
Let ${\cal I}$ be the interaction set of a $k$-sparse $q$-local Hamiltonian on the set of Majorana fermions $[2n]$. The set ${\cal I}$ can be split into $(qQ)/2$ disjoint, strictly $2q'$-local subsets ${\cal I}^{(2q')}_\alpha$ (with $\alpha \in [Q]$ and $q'\in [q/2]$) each of which is diffuse with respect to ${\cal I}$:
\begin{align}
    {\cal I}=\bigcup^{q/2}_{q'=1} \bigcup^{Q}_{\alpha=1} {\cal I}^{(2q')}_\alpha.
\end{align}
The parameter $Q=q(q-1)(k-1)^2+q(k-1)+2$ does not grow with $n$. The construction of this splitting can be done efficiently, in time $\mathrm{poly}(n)$.
\end{lemma}

\begin{lemma}[Generalization of Lemma \ref{lem:diffuse_matching}]
\label{lem:diffuse_matching_general}
Let a strictly $q'$-local $\cal{I}'$ be diffuse w.r.t. $q$-local $k$-sparse $\cal{I}$ on $[2n]$, such that $n>(q^2-1)k$. One can efficiently construct a matching $M$ of $[2n]$ that is consistent with $\cal{I}'$ and inconsistent with all interactions $I \in \mathcal{I}\backslash \mathcal{I}'$ such that (1) $|I|\geq q'$ or (2) $I\not\subset \mathrm{Sup}(\cal{I}')$. 
\end{lemma}

\begin{lemma}[Generalization of Lemma \ref{lem:matching_gaussian}]
\label{lem:matching_gaussian_general}
Let $H=\sum_{I\in {\cal I}} J_I C_I$ on $[2n']$ be $q$-local and ${\cal I}'$ be a diffuse subset of ${\cal I}$. Consider a matching $M$ of $[2n']$. If $M$ is consistent with ${\cal I}'$ and inconsistent with ${\cal I}\backslash{\cal I}'$, one can efficiently construct a Gaussian state $\rho_{{\cal I}'}$ with the property:
\begin{equation}
    {\rm Tr} (H \rho_{{\cal I}'})=\sum_{I\in {\cal I}'} |J_I|.
\end{equation}
\end{lemma}

In Lemma \ref{lem:matching_gaussian_general}, we use $n'$ instead of $n$ to avoid confusion, as it will also be used for $n'\neq n$. The Lemmas above are proven in Appendices \ref{sec:diffuse_splitting_general}-\ref{sec:matching_gaussian_general}. With these in hand, we are ready to proceed with the proof of Theorem \ref{thm:sparse_sub_4}.

\begin{theorem*}[Repetition of Theorem \ref{thm:sparse_sub_4}]
Let $H$ be a traceless fermionic Hamiltonian on $[2n]$ with maximal eigenvalue $\lambda_{\max} (H)$. If $H$ is $k$-sparse with terms of weight $2$ and $4$ and $2n>15k$, a Gaussian state $\rho$ can be efficiently constructed, such that
\begin{equation}
    \frac{{\rm Tr}(H\rho)}{\lambda_{\rm max} (H)} \geq \frac{1}{2Q}
\end{equation}
with $Q=12(k-1)^2+4(k-1)+2$.
\end{theorem*}

\begin{proof}

We make use of the construction in Ref.~\cite{BGKT:manybody} which relates a Hamiltonian with weights $2$ and $4$ on a set of fermionic modes $[2n]$, that is, 
\begin{align}
H=\sum_{I\in{\cal I}^{(2)}} J_I C_I + \sum_{I\in{\cal I}^{(4)}} J_I C_I,
\end{align}
to a strictly $4$-local Hamiltonian $\tilde{H}$ on an extended set of fermions $[2n+2]$:
\begin{align}
\label{eq:tilde_H_decomposed_form}
\tilde{H}=\sum_{I\in{\cal I}^{(2)}} (-i c_{2n+1} c_{2n+2}) J_I C_I + \sum_{I\in{\cal I}^{(4)}} J_I C_I.
\end{align}
Introducing $\tilde{\cal I}^{(2)}\equiv\{(2n+1,2n+2)\cup I\,|\, I\in {\cal I}^{(2)}\}$, $\tilde{H}$ can be also written as: 
\begin{align}
\label{eq:tilde_H_merged_form}
\tilde{H}=-\sum_{I\in\tilde{\cal I}^{(2)}} J_I C_I + \sum_{I\in{\cal I}^{(4)}} J_I C_I.
\end{align}

The relation between $\tilde{H}$ and $H$ is via the following property:

\begin{lemma}[Lemma 6 of \cite{BGKT:manybody}]
\label{lem:bravyi}
For $H$ and $\tilde{H}$ introduced above, $\lambda_{\rm max}(H)=\lambda_{\rm max}(\tilde{H})$. Moreover, for any Gaussian state $\tilde{\rho}$ of ${2n+2}$ Majorana modes, one can efficiently compute a Gaussian state $\rho$ of $2n$ Majorana modes s.t. ${\rm Tr} (H\rho )\geq {\rm Tr} (\tilde{H}\tilde{\rho})$.
\end{lemma}

Although strictly $4$-local, Hamiltonian $\tilde{H}$ is no longer sparse since the operators $c_{2n+1}$ and $c_{2n+2}$ participate in $|{\cal I}^{(2)}|$ terms (which is generally $O(n)$). This prevents a direct application of Lemma \ref{lem:diffuse_splitting} to $\tilde{H}$. We resolve the issue as follows.

Similarly to the proof of Theorem \ref{thm:sparse}, we start by splitting each set of the original interactions ${\cal I}^{(2,4)}$ in $H$ into subsets diffuse w.r.t. ${\cal I}^{(2)}\cup{\cal I}^{(4)}$: ${\cal I}^{(2)}=\cup_\alpha {\cal I}^{(2)}_\alpha$, ${\cal I}^{(4)}=\cup_\alpha {\cal I}^{(4)}_\alpha$. Each of the two splittings exists and can be done efficiently, as guaranteed by Lemma \ref{lem:diffuse_splitting} (since the original $H$ is sparse). Since ${\cal I}^{(2)}\cup{\cal I}^{(4)}$ is $k$-sparse and $4$-local, we can bound $|\{{\cal I}^{(2)}_\alpha\}|<Q$, $|\{{\cal I}^{(4)}_\alpha\}|<Q$ for $Q=12(k-1)^2+4(k-1)+2$. In what follows, we will use the splittings ${\cal I}=\bigcup_{\alpha}{\cal I}^{(2)}_\alpha\cup\bigcup_{\alpha}{\cal I}^{(4)}_\alpha$ to construct two Gaussian states $\tilde{\rho}({\cal I}^{(2)}_\alpha)$ and $\tilde{\rho}({\cal I}^{(4)}_\alpha)$ on $[2n+2]$ with good properties relative to $\tilde{H}$, that is,
\begin{align}
    {\rm Tr}(\tilde{H}\tilde{\rho}({\cal I}^{(2,4)}_\alpha))=\sum_{I\in {\cal I}^{(2,4)}_\alpha} |J_I|\equiv \mathbf{J}({\cal I}^{(2,4)}_\alpha).
\end{align}
With these Gaussian states, we will then show that the Gaussian state $\tilde{\rho}({\cal I}^{(2,4)}_\alpha)$ for $q,\,\alpha=\mathrm{argmax}_{q',\,\alpha'}(\mathbf{J}({\cal I}^{(q')}_{\alpha'}))$ is efficiently constructable and yields the desired approximation ratio for $\tilde{H}$. We will then apply Lemma \ref{lem:bravyi} and extend the statement to the original Hamiltonian $H$, thus finishing the proof.

Following the outline above, we now move to construct the Gaussian state $\tilde{\rho}({\cal I}^{(2)}_\alpha)$. Consider an ansatz of the form $\tilde{\rho}\equiv\rho_{[2n]}\sigma_{\{2n+1,2n+2\}}$, where $\rho_{[2n]}$ is itself a Gaussian state of $[2n]$. To construct $\rho_{[2n]}$, note that each ${\cal I}^{(2)}_\alpha$ is $2-$local and diffuse w.r.t. ${\cal I}^{(2)}\cup{\cal I}^{(4)}$ which is $4$-local. Since $2n>15k$ by assumptions of Theorem \ref{thm:sparse_sub_4}, 
we can apply Lemma \ref{lem:diffuse_matching_general} with $q=4$ to construct a matching $M({\cal I}^{(2)}_\alpha)$ that is consistent with ${\cal I}^{(2)}_\alpha$. Since ${\cal I}^{(2)}_\alpha$ is $2$-local, Lemma \ref{lem:diffuse_matching_general} also implies that the matching $M({\cal I}^{(2)}_\alpha)$ is inconsistent with the entirety of ${\cal I}^{(4)}\cup{\cal I}^{(2)}\backslash {\cal I}^{(2)}_\alpha$. We then use $M({\cal I}^{(2)}_\alpha)$ in Lemma \ref{lem:matching_gaussian_general} (substituting $n'=n$ for $n'$ used in the Lemma)
to construct $\rho_{[2n]}$ in $\tilde{\rho}=\rho_{[2n]}\sigma_{\{2n+1,2n+2\}}$. This implies the following expression (using Eq.~\eqref{eq:tilde_H_decomposed_form} for $\tilde{H}$):
\begin{equation}
    {\rm Tr} (\tilde{H}\tilde{\rho}) =
    {\rm Tr}(\sigma (-i c_{2n+1} c_{2n+2})) \sum_{I\in{\cal I}^{(2)}_\alpha}  |J_I|.
\end{equation}
By choosing $\sigma$ to be the $+1$ eigenstate projector of operator $-i c_{2n+1} c_{2n+2}$, we arrive at the desired outcome:
\begin{equation}
    {\rm Tr}(\tilde{H} \tilde{\rho})=\sum_{I\in{\cal I}^{(2)}_\alpha}  |J_I|=\mathbf{J}({\cal I}^{(2)}_\alpha).
\end{equation}
The constructed Gaussian state $\tilde{\rho}$ we will denote as $\tilde{\rho}({\cal I}^{(2)}_\alpha)$.

\begin{figure}
    \centering
    \includegraphics[width=0.99\linewidth]{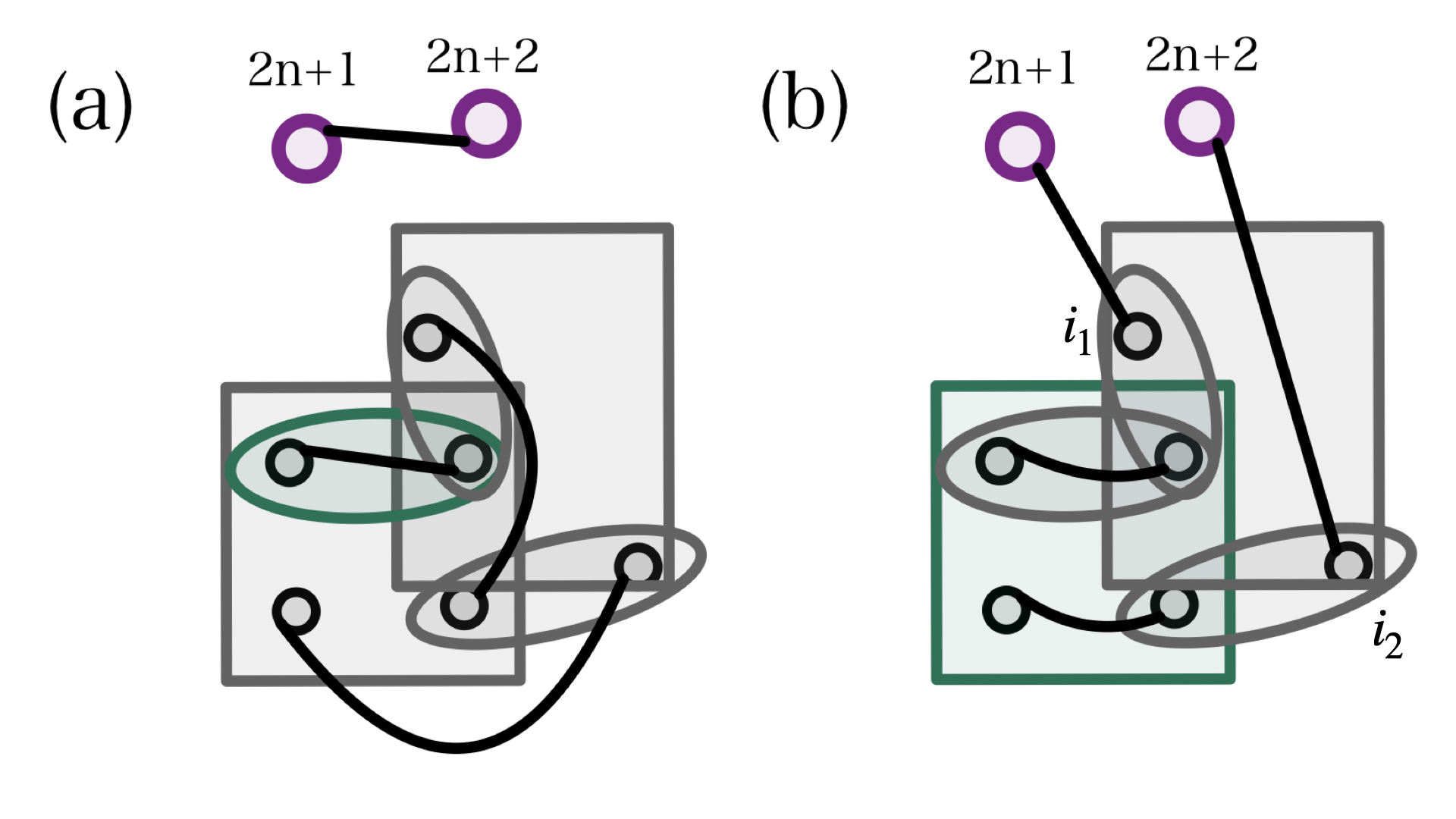}
    \caption{Demonstration of the method in the proof of Theorem~\ref{thm:sparse_sub_4}. (a) Matching $M(\mathcal{I}^{(2)}_\alpha)$ for $\mathcal{I}^{(2)}_\alpha$, here comprised of a single term (shown in green). To ensure consistency with $\mathcal{I}^{(2)}_\alpha$ in $\tilde{H}$, $M(\mathcal{I}^{(2)}_\alpha)$ perfectly matches these terms and the pair $(2n+1, 2n+2)$. The rest of the vertices are matched so that each pair does not belong to the same term in $\mathcal{I}\backslash\mathcal{I}^{(2)}_\alpha$ (grey). (b) Matching $M(\mathcal{I}^{(4)}_\alpha)$ for $\mathcal{I}^{(4)}_\alpha$ shown in green. Vertices $i_1, i_2$ are chosen not to belong to the same term in $\mathcal{I}^{(2)}$, ensuring no accidental consistency with a term in $\tilde{H}$. Note the special status of the term from $\mathcal{I}^{(2)}$ that is a subset of the $\mathcal{I}^{(4)}_\alpha$ term. From the perspective of $\tilde{H}$, it is not consistent with $M(\mathcal{I}^{(4)}_\alpha)$ although it coincides with an edge from $M(\mathcal{I}^{(4)}_\alpha)$. This is due to the intentional absence of the edge $(2n+1, 2n+2)$ in $M(\mathcal{I}^{(4)}_\alpha)$.}
    \label{fig:sub_4}
\end{figure}

For a diffuse ${\cal I}^{(4)}_\alpha\subset {\cal I}^{(2)}\cup{\cal I}^{(4)}$, we construct the Gaussian states $\tilde{\rho}({\cal I}^{(4)}_\alpha)$ in a different way. First we use Lemma \ref{lem:diffuse_matching_general} to construct a matching $M({\cal I}^{(4)}_\alpha)$ of $[2n]$. This matching is guaranteed to be consistent with ${\cal I}^{(4)}_\alpha$. However, since ${\cal I}^{(2)}$ is $2$-local and ${\cal I}^{(4)}$ is 4-local, while in general $\rm{Sup}({\cal I}^{(2)})\cap \rm{Sup} ({\cal I}^{(4)}_\alpha)\neq\emptyset$, Lemma \ref{lem:diffuse_matching_general} implies that $M({\cal I}^{(4)}_\alpha)$ is inconsistent with ${\cal I}^{(4)}\backslash{\cal I}^{(4)}_\alpha$ but may be consistent with some terms in ${\cal I}^{(2)}$ (as those $I$ don't obey the $|I|\geq q'=4$ condition). At the same time, we aim to achieve ${\rm Tr}(\tilde{H} \tilde{\rho}({\cal I}^{(4)}_\alpha))=\mathbf{J}({\cal I}^{(4)}_\alpha)$ which excludes contributions from ${\cal I}^{(2)}$. 
Thus we cannot extend $M({\cal I}^{(4)}_\alpha)$ to the extended set $[2n+2]$ directly, as it was done for ${\cal I}^{(2)}_\alpha$. 
Instead, we will create a matching of $[2n+2]$ using a reduced version of $M({\cal I}^{(4)}_\alpha)$ which inherits its beneficial properties, and then complete the matching by making it inconsistent with $\tilde{\cal I}^{(2)}$ -- eliminating the difficulty described above.

To enable this, we find and mark an edge $(i_1,i_2) \in M({\cal I}^{(4)}_\alpha)$, such that $i_1\not\in \mathrm{Sup}({\cal I}^{(4)}_\alpha)$. This is always possible since ${\cal I}^{(4)}_\alpha$ is diffuse and thus $[2n]\backslash\rm{Sup}({\cal I}^{(4)}_\alpha)$ is non-empty (cf. Condition 3 in Definition \ref{def:diffuse}).
Note that $M({\cal I}^{(4)}_\alpha)$ is constructed via Lemma \ref{lem:diffuse_matching_general} and $\{i_1, i_2\}\not\subset\mathrm{Sup}({\cal I}^{(4)}_\alpha)$. This implies that as a two-fermion interaction, $\{i_1, i_2\}$ is guaranteed not to belong to $\mathcal{I}^{(2)}$. The latter statement is the key property of the marked edge $(i_1,i_2)$ that we will employ momentarily.

We construct a matching $\tilde{M}({\cal I}^{(4)}_\alpha)$ of $[2n+2]$ in two stages. First we construct an intermediate matching $M'({\cal I}^{(4)}_\alpha)$ of $[2n+2]\backslash\{i_1,i_2,2n+1,2n+2\}$ by removing the edge from $M({\cal I}^{(4)}_\alpha)$:
\begin{align}
    M'({\cal I}^{(4)}_\alpha)=M({\cal I}^{(4)}_\alpha)\backslash(i_1,i_2).
\end{align}
Since $\{i_1,i_2\}\not\subset\mathrm{Sup}({\cal I}^{(4)}_\alpha)$, we are guaranteed that ${M}'({\cal I}^{(4)}_\alpha)$ is consistent with ${\cal I}^{(4)}_\alpha$ and inconsistent with ${\cal I}^{(4)}\backslash{\cal I}^{(4)}_\alpha$ (from the construction of $M({\cal I}^{(4)}_\alpha)$). 
In the second stage, we complete $\tilde{M}({\cal I}^{(4)}_\alpha)$ to the entire set of $2n+2$ modes by adding two edges: $(i_1,2n+1)$ and $(i_2,2n+2)$:

\begin{align}
    \tilde{M}({\cal I}^{(4)}_\alpha)=M({\cal I}^{(4)}_\alpha)\cup\{(i_1,2n+1), (i_2,2n+2)\}.
\end{align}

These new edges render $\tilde{M}({\cal I}^{(4)}_\alpha)$ inconsistent with $\tilde{\mathcal{I}}^{(2)}$. To see it, note that all interactions in $\tilde{\mathcal{I}}^{(2)}$ take the form $I=\{j_1, j_2, 2n+1,2n+2\}$ where $\{j_1,j_2\} \in  {\cal I}^{(2)}$. By construction $\{i_1,i_2\}\notin {\cal I}^{(2)}$, thus we have $\{j_1,j_2\}\neq\{i_1, i_2\}$. As a result, matching $\tilde{M}({\cal I}^{(4)}_\alpha)$ of $[2n+2]$ is consistent with ${\cal I}^{(4)}_\alpha$ and inconsistent with $\tilde{\cal I}^{(2)}\cup {\cal I}^{(4)}\backslash{\cal I}^{(4)}_\alpha$. We continue by applying Lemma \ref{lem:matching_gaussian_general} to such $\tilde{M}({\cal I}^{(4)}_\alpha)$ and $\tilde{H}$ (substituting $n'=n+1$ for $n'$ used in the Lemma). This efficiently constructs a Gaussian state $\tilde{\rho} ({\cal I}^{(4)}_\alpha)$ that yields:
\begin{align}
    {\rm Tr}(\tilde{H} \tilde{\rho}({\cal I}^{(4)}_\alpha))=\sum_{I\in {\cal I}^{(4)}_\alpha} |J_I|\equiv \mathbf{J}({\cal I}^{(4)}_\alpha),
\end{align}
as desired.

The Gaussian state claimed in Theorem \ref{thm:sparse_sub_4} is to be chosen among the states $\tilde{\rho} ({\cal I}^{(2,4)}_\alpha)$ whose existence we've proven above. We make the choice by identifying the highest energy in the respective Gaussian state: $(q,\alpha)=\mathrm{argmax}_{(q,\alpha)} \mathbf{J}^{(q)}_\alpha$. As we showed, the respective Gaussian state $\tilde{\rho}({\cal I}^{(q)}_\alpha)$ can be efficiently constructed and the following is guaranteed:
\begin{align}
\frac{{\rm Tr}(\tilde{H} \tilde{\rho} ({\cal I}^{(q)}_\alpha))}{\lambda_{\rm max}(\tilde{H})}\geq \frac{\mathbf{J}^{(q)}_\alpha}{  \sum_{q',\alpha'} \mathbf{J}^{(q')}_{\alpha'}} \geq \frac{1}{2Q} .
\end{align} 
Here we used that $\lambda_{\rm max}(\tilde{H})\leq \sum_{q',\alpha'} \mathbf{J}^{(q')}_{\alpha'}$ and that $\mathbf{J}^{(q)}_\alpha=\rm{max}_{(q',\alpha')} \mathbf{J}^{(q')}_{\alpha'}$. 

With the state $\tilde{\rho}({\cal I}^{(q)}_\alpha)$ on $[2n+2]$ fermions at hand, we finalize the proof by an application of Lemma \ref{lem:bravyi}. This relates $\lambda_{\rm max}(H)$ to $\lambda_{\rm max}(\tilde{H})$ and allows us to efficiently construct the Gaussian state $\rho({\cal I}^{(q)}_\alpha)$ of $[2n]$, with the desired property:

\begin{align}
\frac{{\rm Tr} (H \rho({\cal I}^{(q)}_\alpha))}{\lambda_{\rm max}(H)}\geq\frac{{\rm Tr} (\tilde{H} \tilde{\rho}({\cal I}^{(q)}_\alpha))}{\lambda_{\rm max}(\tilde{H})}\geq \frac{1}{2Q} .
\end{align}

\end{proof}

\section{The sparse SYK-\texorpdfstring{$4$}{4} model}\label{sec:sparse_SYK}

\begin{theorem*}[Repetition of Theorem \ref{thm:sparse_SYK}]
Let $H$ be a \SSSYKf Hamiltonian in Eq.~\eqref{eq:sparseSYKdef} with expected degree $k=O(1)$, such that $n>120(k+1)$. With probability at least $1-4\exp\left[-\frac{e^{-16(k+1)}k^3}{64(8k+7)} n\right]$, a Gaussian state $\rho$ can be efficiently constructed such that
\begin{equation}
    \frac{{\rm Tr}(H \rho)}{\lambda_{\max} (H)} \geq \frac{1}{Q},
\end{equation}
where $Q=1236 + 2752 k + 1536 k^2$.
\end{theorem*}

\begin{proof}
In what follows we will omit the normalization $1/\sqrt{2k n}$ in Eq.~\eqref{eq:sparseSYKdef}, of course this normalization is irrelevant for lowerbounding the Gaussian approximation ratio.
We split $H$ as $H=H^{(k')}+h^{(k')}$, s.t. the Hamiltonian $H^{(k')}$ is $k'$-sparse and the residual Hamiltonian $h^{(k')}$ contains the rest of $H$. The term sets are denoted as follows:
\begin{align}
    H^{(k')}=\sum_{I\in \mathcal{I}^{(k')}} J_I C_I,~~
    h^{(k')}=\sum_{I\in \bar{\mathcal{I}}^{(k')}} J_I C_I,
\end{align}
i.e. $\mathcal{I}=\mathcal{I}^{(k')}\cup \bar{\mathcal{I}}^{(k')}$. To define such a split, we use the following deterministic algorithm. For every given Majorana, we list the interactions $I\in\mathcal{I}$ which involve that Majorana using a lexicographical order for the words $I=\{i_1,i_2,i_3,i_4\}$. For each Majorana where such a list is longer than $k'$, we mark all elements except for the first $k'$. All terms of $H$ which were marked this way at least once, we include into $h^{(k')}$. The rest of the terms enter $H^{(k')}$, which by this construction is $k'$-sparse. To continue the proof we need a pair of Lemmas. The first lower bounds the total interaction strength of the \SSSYKf Hamiltonian:

\begin{lemma}
\label{lem:H_SSYK}
With probability at least $1-2e^{-\frac{kn}{32}}$, we have
\begin{equation}
    \sum_{I\in \mathcal{I}} |J_I|\geq kn/8.
\end{equation}
\end{lemma}

This statement is proven in Appendix \ref{app:SSYK_concentration}, by splitting the problem into upper bounding $|\mathcal{I}|$ separately from $\sum_{I} |J_I|$, and then applying the Chernoff bound for both. 

The second lemma shows that the total interaction strength of the residual Hamiltonian $h^{(k')}$ is bounded from above with high probability:

\begin{lemma}
\label{lem:h_kappa}
If $k'\geq e^2 k+1$, we have with probability at least $1- 2\exp\left[-\frac{e^{-2k'}k^3}{64(k'-1)} n\right]$ that
\begin{equation}
\sum_{I\in \bar{\mathcal{I}}^{(k')}} |J_I| \leq \frac{4k^2}{\sqrt{k'-1}}e^{-k'}n.
\end{equation}
\end{lemma}

Lemma \ref{lem:h_kappa} is proven in Appendix \ref{app:SSYK_concentration}. The key technical difficulty is bounding the random variable $\bar{\mathcal{I}}^{(k')}$, which does not reduce to a sum of independent variables and thus a simple Chernoff bound cannot be applied. Instead, we apply an exponential version of Efron-Stein inequality \cite{bouch:concentration}.\\

To build a Gaussian state with finite approximation ratio, we apply the construction of Theorem \ref{thm:sparse} to $H^{(k')}$, which is $k'$-sparse and strictly $4$-local. If $n$ is large enough (i.e. $n>(q^2-1)k'$ for $q=4$), this state $\rho$ is guaranteed to yield energy ${\rm Tr} (H^{(k')}\rho)> \frac{1}{Q'} \sum_{I\in {\mathcal{I}}^{(k')}} |J_I|$ for $Q'=12k'^2 -20k' +10$ (see Eq.~\eqref{eq:sparse_ratio} in the proof of Theorem \ref{thm:sparse}).
At the same time, with high probability $|{\rm Tr} (h^{(k')}\rho)|\leq\sum_{I\in \bar{\mathcal{I}}^{(k')}} |J_I|\leq \frac{4k^2}{\sqrt{k'-1}}e^{-k'}n$ and $\sum_{I\in {\mathcal{I}}} |J_I| \geq \frac{kn}{8}$ (Lemmas \ref{lem:H_SSYK} and \ref{lem:h_kappa}). The resulting approximation ratio is then:
\begin{align}
    \frac{{\rm Tr}(H\rho)}{\lambda_{\mathrm{max}}(H)}&\geq\frac{{\rm Tr} (H^{(k')}\rho)-|{\rm Tr} (h^{(k')}\rho)|}{\sum_{I\in {\mathcal{I}}} |J_I|} \notag \\
    &\geq \frac{\frac{1}{Q'}\sum_{I\in {\mathcal{I}}} |J_I|-(1+\frac{1}{Q'})\sum_{I\in \bar{\mathcal{I}}^{(k')}} |J_I|}{\sum_{I\in {\mathcal{I}}} |J_I|} \notag \\
    &\geq \frac{1}{Q'} - \frac{32(Q'+1)ke^{-k'}}{\sqrt{k'-1} Q'} \label{eq:Q'_SSYK_bound}.
\end{align}
Crucially, the second term decays exponentially with $k'$ and the first term only algebraically (note here the definition of $Q'$). We now fix $k'=8(k+1)$, consistent with the requirement $k'\geq e^2k+1$ of Lemma \ref{lem:h_kappa}. In this case $\frac{32(Q'+1)ke^{-k'}}{\sqrt{k'-1} Q'}$ as a function of $k$ is always smaller than $\frac{1}{2Q'}$. This allows us to bound the right hand side of Eq.~\eqref{eq:Q'_SSYK_bound} as $\frac{1}{2Q'}$, and substituting $k'=8(k+1)$ we obtain the bound claimed in the Theorem:

\begin{align}
    \frac{{\rm Tr}(H\rho)}{\lambda_{\mathrm{max}}(H)}\geq \frac{1}{1236 + 2752 k + 1536 k^2} \label{eq:Q_SSYK_bound}.
\end{align}
The earlier assumed condition $n>(q^2-1)k'$ for $q=4$ and $k'$ translates into $n>120(k+1)$.
Given the conditions of Lemmas \ref{lem:H_SSYK} and \ref{lem:h_kappa}, the bound in Eq.~\eqref{eq:Q_SSYK_bound} holds with the probability:

\begin{align}
    &\left(1- 2\exp\left[-\frac{e^{-16(k+1)}k^3}{64(8k+7)} n\right]\right)\left(1-2e^{-\frac{kn}{32}}\right)\notag\\
    &\geq 1- 4\exp\left[-\frac{e^{-16(k+1)}k^3}{64(8k+7)} n\right].
\end{align}

\end{proof}

\section{Upper bound on Gaussian approximation ratio for SYK-\texorpdfstring{$q$}{q} Hamiltonians}
\label{section:gaussupperbound}

\subsection{Gaussian upper bound for SYK-\texorpdfstring{$q$}{q} models}
\label{sec:gauss-upper}

We consider the expectation value of a SYK-$q$ Hamiltonian $H$ with respect to fermionic Gaussian states and we obtain an upper bound on its expectation value, with high probability over the random couplings $J_I$. 

\begin{lemma*}[Repetition of Lemma \ref{lemma:SYKexpvaluegaussians}]
Let $H$ denote a Hamiltonian drawn from the $q$-local SYK Hamiltonians (with $q\geq 4$ even and $q=O(1)$), i.e. the coupling strengths $J_I$ are drawn according to their distribution. With probability at least $1-\exp(-\Omega(n))$, $H$ has the property that, for any fermionic Gaussian state $\rho$
\begin{align}
    \text{\normalfont{Tr}}(H \rho ) &\leq (q\!-\!1)!!\:2^{1/2-q/4}q^{1/2+q/2}\notag\\
    &\hspace{4em}\times\!\sqrt{\log[q/\log(3/2)]}\:(2n)^{1-q/4}\!.\!
\end{align}
\end{lemma*}

\begin{proof}
We first use Wick's theorem on the expectation of a product of Majorana operators w.r.t. a fermionic Gaussian state $\rho$ characterized by a correlation matrix $\Gamma$, see Eq.~\eqref{eq:wick}. Note that the correlation matrix $\Gamma_{i<j}$ can be viewed as a real $d:=(2n^{2}-n)$-dimensional vector. We note that $\sum_{i < j}\Gamma_{ij}^2=\frac{1}{2}{\rm Tr}(\Gamma^T \Gamma)\leq \frac{1}{2}{\rm Tr}(\mathbb{I})=n$ so that $\|\Gamma\|\leq n^{1/2}$.

Let $M(I)$ be a perfect matching of the indices in $I$ ($|I|$ even), there are $(q-1)!!$ such matchings.
We have 
\begin{align}
    \text{Tr}(C_I \rho ) &=   i^{q/2}\sum_{M(I)}\text{sign}\big(M(I)\big)\:\text{Tr}\big( c_{i_{1}(M)}c_{i_{2}(M)}\rho)\notag\\
    &\hspace{2em}\times\text{Tr}(c_{i_{3}(M)}c_{i_{4}(M)}\rho) \ldots \text{Tr}( c_{i_{q\!-\!1}(M)}c_{i_{q}(M)}\rho).
\label{eq:wicks}
\end{align}
Here we have assumed that for each matching $M(I)$; $i_{1}(M)<i_{2}(M)$, $i_{3}(M)<i_{4}(M)$, $\ldots$, $i_{q-1}(M)<i_{q}(M)$, i.e. any sign arising from getting the expression to this form is absorbed in $\text{sign}\big(M(I)\big)$).

The expectation of $H$ in Eq.~\eqref{SYK_H} w.r.t. fermionic Gaussian states $\rho$ can be written as:
\begin{center}
\begin{align}
    \text{Tr}(H\rho) &=\: \binom{2n}{q}^{-1/2}  i^{q/2}\!\!\!\!\!\sum_{I\subseteq [2n],\,|I|=q}\!\!\!\!\!J_{I}\: \Big[ \sum_{M(I)}\text{sign}\big(M(I)\big)\notag\\
    &\hspace{7em}\times \prod_{t=1}^{q/2}\text{Tr}\big( c_{i_{2t-1}(M)}c_{i_{2t}(M)}\rho\big) \Big] \nonumber \\ 
    &=\: \binom{2n}{q}^{-1/2}   \!\!\!\!\!\sum_{I\subseteq [2n],\,|I|=q}\!\!\!\!\!J_I \sum_{M(I)} \text{sign}\big(M(I)\big) \notag\\
    &\hspace{7em}\times\prod_{t=1}^{q/2}\: \Gamma_{i_{2t-1}(M),i_{2t}(M)} . 
\label{eq:expupperbound}
\end{align}
\end{center}
We note that we can view ${\rm Tr}(H\rho)$ as a sum of $(q-1)!!$ terms, one for each matching $M$ of some subset of indices $I$, i.e. ${\rm Tr}(H \rho)=\sum_{M} {\rm Tr}(H_M \rho)$ where $H_M=\sum_I \tilde{J}(M,I) \prod_{t=1}^{q/2}\: \Gamma_{i_{2t-1}(M),i_{2t}(M)}$. We have defined the $q/2$-way, $d\times d\times\ldots\times d$, tensor $\tilde{J}(M,I)$, whose entries are equal to either zero (when the indices coincide or are not ordered properly) or to a standard Gaussian random variable. Each $J_{I}$ appears only once in $\text{Tr}(H_{M}\rho)$ and therefore all entries of $\tilde{J}(M,I)$ are statistically independent. We note that ${\rm sign}(M)$ does not depend on which (ordered) subset $I$ one chooses. To bound each term ${\rm Tr}(H_M \rho)$, with high probability, we invoke the following Lemma:

\begin{lemma}{(Theorem 1 in \cite{spectralnormrandomtensors}.)}
Let $A$  be a random $K$-way tensor $\in \mathbb{R}^{d_{1}\times d_{2} \times \ldots \times d_{K}}$ and $w_{i}$ be vectors $\in \mathbb{R}^{d_{i}}$ and 
\begin{equation*}
    A(w_{1},w_{2},\ldots,w_{K}) :=\!\!\!\! \sum_{k_{1},\ldots,k_{K}}\!\!\!\! A_{k_{1},\ldots,k_{K}}(w_{1})_{k_{1}}\ldots (w_{q/2})_{k_{K}}.
\end{equation*}
If we have for each fixed unit vector $w_{i}/\|w_{i}\|$ ($i \in \{1,\ldots,K\}$):
\begin{align}
    &\text{\normalfont{Pr}}\Big( |A(w_{1}/\|w_{1}\|,\ldots,w_{K}/\|w_{K}\|)|\geq t \Big)\notag\\ 
    &\hspace{10em}\leq 2\exp\Big(-t^{2}/(2\sigma^{2})\Big),
\end{align}
then the spectral norm $\|\!|A\|\!|:=\max_{w_{1},\ldots,w_{K}}A\:\big(w_{1}/\|w_{1}\|,\ldots,w_{K}/\|w_{K}\|\big)$ (with $w_{i}\in \mathbb{R}^{d_{i}}$) can be bounded as follows:
\begin{equation*}
    \|\!|A\|\!| \!\leq\!\! \bigg[ 8\sigma^{2}\Big[ \Big( \sum_{i=1}^{K} \!d_{i}\! \Big) \log\!\big[2K/\log(3/2)\big] + \log\big(\frac{2}{\delta}\big) \Big] \bigg]^{\frac{1}{2}},
\end{equation*}
with probability at least $1-\delta$.
\label{lemma:randomtensors}
\end{lemma}
To apply the Lemma, note that the vectors $w_i$ correspond to $\Gamma_{i< j}$ viewing $i < j$ as a single index and we can use their norm $\|\Gamma\|\leq n^{1/2}$.
In addition, for each entry in the tensor we have $\mathbb{E}\big[\exp \big(t\tilde{J}(M,I)_{k_{1},\ldots,k_{q/2}}\big)\big] \leq \exp\big( t^{2}/2 \big)$ (for $t\geq 0$) as the entry is zero or a Gaussian variable with variance 1 and mean zero. Using Chernoff's bound and the fact that all entries of $\tilde{J}(M,I)$ are statistically independent, we conclude that for any set of real vectors $w_1,\ldots,w_{q/2}$ one has
\begin{align}
    &\text{Pr}\Big[ \Bigl\lvert\sum_{k_1,\ldots, k_{q/2}}\!\!\!\!\!\! \tilde{J}(M,I)_{k_1,\ldots,k_{q/2}} \!\frac{(w_{1})_{k_1}}{\|w_{1}\|} \!\ldots\! \frac{(w_{q/2})_{k_{q/2}}}{\|w_{q/2}\|} \Bigr\rvert \!\geq\! t \Big]\notag\\
    &\hspace{13em}\leq 2\exp(-t^{2}/2).
\end{align}
Therefore, for each term $H_M$ we can apply Lemma \ref{lemma:randomtensors} and, using $K=q/2$ and $\sigma=1$, obtain
\begin{align}
    &\bigl\lvert\!\bigl\lvert\!\bigl\lvert \tilde{J}(M,I) \bigr\rvert\!\bigr\rvert\!\bigr\rvert \leq \bigg[ 4q(2n^{2}\!-\!n) \log\big[q/\log(3/2)\big]\notag\\ 
    &\hspace{11em}+ 8\log\big(2\delta^{-1}\big) \bigg]^{1/2}\!\!\!\!,
\label{eq:specnormupperbound}
\end{align}
with probability at least $1-\delta$. Then we can first bound
\begin{center}
\begin{align}
    \max_{\rho \text{ Gaussian}}\text{Tr}(H\rho) \leq&\: \binom{2n}{q}^{-1/2}\|\Gamma\|^{q/2}\sum_{M}\bigl\lvert\!\bigl\lvert\!\bigl\lvert \tilde{J}(M,I) \bigr\rvert\!\bigr\rvert\!\bigr\rvert \nonumber \\
    \leq&\: \big(q/\sqrt{2}\big)^{q/2}(2n)^{-q/4}\:\sum_M \bigl\lvert\!\bigl\lvert\!\bigl\lvert \tilde{J}(M,I) \bigr\rvert\!\bigr\rvert\!\bigr\rvert,
\label{eq:upperboundFGS}
\end{align}
\end{center}
where we have used that $\binom{2n}{q}\geq (2n/q)^{q}$. We can now combine the upper bound in Eq.~\eqref{eq:specnormupperbound} and Eq.~ \eqref{eq:upperboundFGS}. Applying the union bound, we have with probability at least $1-(q-1)!!\:\delta$, that
\begin{align}
    &\max_{\rho \text{ Gaussian}} \text{Tr}(H\rho )\leq (q-1)!!\:\bigg[ 2^{1-q/2}q^{q+1}\notag\\
    &\hspace{4em}\times\!\big[(2n)^{2-q/2}-(2n)^{1-q/2}\big]\log\big[q/\log(3/2)\big] \notag\\
    &\hspace{6em}+ 2^{3-q/2}q^{q}(2n)^{-q/2}\log\big(2\delta^{-1}\big) \bigg]^{1/2}\!\!\!\!.\!\!
\end{align}
Therefore, we can take $\delta = \exp\big(-\Omega(n)\big)$ such that, asymptotically, we have (assuming $q=O(1)$):
\begin{align}
    &\max_{\rho \text{ Gaussian}} \text{Tr}(H\rho) \leq (q-1)!!\:2^{1/2-q/4}q^{1/2+q/2}\notag\\
    &\hspace{6em}\times\sqrt{\log[q/\log(3/2)]}\:(2n)^{1-q/4},
\end{align}
with probability at least $1-\delta$. Note that in deriving this upper bound we only use the norm of the correlation matrix $\Gamma$, hence this upper bound is not necessarily achievable by a Gaussian state as the constraint $\Gamma^T \Gamma \leq \mathbb{I}$ imposes more conditions on $\Gamma$ than just an upper bound on its norm.
\end{proof}

\subsection{Maximum eigenvalue lower bound for \texorpdfstring{$q$}{q}-local SYK Hamiltonians}
\label{sec:maxbound}

To show that fermionic Gaussian states cannot achieve a constant approximation ratio for $q\geq 4$ SYK models, we derive a lower bound on the maximum eigenvalue of the Hamiltonians $H$ in Eq.~\eqref{SYK_H}:
\begin{lemma*}[Repetition of Lemma \ref{lemma:maxeiglowbound}]
For the class of $q$-local SYK Hamiltonians (with even $q\geq4$) in Eq.~\eqref{SYK_H}, $\lambda_{\max}(H) = \Omega(\sqrt{n})$ with probability at least $1-\exp\big( -\Omega(n) \big)$ over the draw of Hamiltonians.
\end{lemma*}
The remainder of this section will be devoted to proving this Lemma. The techniques used are similar to those used in Section 6 of Ref.~\cite{HO:approxferm}. We note that throughout this section, we shall use $C$ to denote a quantity that is constant in $n$ or is bounded from above and below by a constant in $n$, and it will generally differ from appearance to appearance (for the sake of clarity). Importantly, $C$ can contain factors of $q$ (note that $q=O(1)$).

We start by obtaining a lower bound on the maximum eigenvalue of a so-called $2$-colored SYK model and will use this to prove Lemma \ref{lemma:maxeiglowbound}. The Hamiltonian of such a $2$-colored SYK model is slightly different from the standard SYK model Hamiltonian in Eq.~\eqref{SYK_H}. We divide the $2n$ Majorana operators into two subsets, with sizes $n_{1}$ and $n_{2}$ ($n_2\leq n_1$), and denote the operators in the first set by $\phi_{1},\ldots,\phi_{n_1}$ and the ones in the second set by $\chi_{1},\ldots,\chi_{n_2}$. The Hamiltonian is now given by\footnote{We denote Hamiltonians from the class of $2$-colored SYK Hamiltonians by $H^{(2)}$.}: 
\begin{equation}
    H^{(2)} = \frac{i}{\sqrt{n_{2}}}\sum_{j=1}^{n_{2}}\tau_{j}\chi_{j},
\label{SYK_2}
\end{equation}
where
\begin{equation}
    \tau_{j} = \binom{n_{1}}{q-1}^{-1/2}i^{q/2-1}\sum_{\substack{S\subseteq [n_{1}]\\ |S| = q-1}}J_{S,j}\: \phi^{S}.
\end{equation}
Here $\phi^{S}$ the product of $q-1$ of the $\phi$ Majorana operators in subset $S$, and $J_{S,j}$ are independent Gaussian random variables. The subset $S$ labels an ordered subset of $q-1$ Majorana operators (note that these are different from the subsets $I$ defined before that correspond to ordered subsets of $q$ Majorana operators). We note that the (Hermitian) $\tau_{j}$ operators do not necessarily obey $\{\tau_{j},\tau_{k}\} = 2\delta_{jk}\mathbb{I}$, but instead satisfy $\mathbb{E}(\{\tau_{j},\tau_{k}\}) = -i^{q-2}\delta_{jk}\mathbb{I}$. 

\begin{lemma}
Let $\{\phi_{i}\}_{i=1}^{n_{1}}$ and $\{\chi_{i}\}_{i=1}^{n_{2}}$ be $n_{1}+n_{2}$ Majorana operators. For the class of $q$-local $2$-colored SYK Hamiltonians (with even $q\geq4$) in Eq. \eqref{SYK_2} defined in terms of these Majorana operators, the maximum eigenvalue of the Hamiltonian $\lambda_{\max}(H)$ is lower bounded by $C\sqrt{n}$ (with $C$ a constant) with probability at least $1-\exp\big( -\Omega(n) \big)$ over the draw of Hamiltonians.
\label{lemma:maxeiglowbound2color}
\end{lemma}

\begin{proof}
We introduce a new set of Majorana operators (again of size $n_{2}$) $\sigma_{1},\ldots,\sigma_{n_{2}}$ (which do obey $\{\sigma_{j},\sigma_{k}\} = 2\delta_{jk}\mathbb{I}$) and we define the \textit{quadratic} Hamiltonian $H'$:
\begin{equation}
    H' = \frac{i}{\sqrt{n_{2}}}\sum_{j=1}^{n_{2}}\sigma_{j}\chi_{j}.
\end{equation}
This quadratic Hamiltonian $H'$ is optimized by the fermionic Gaussian state $\rho_{0} = \frac{1}{2^{n_{2}+n_{1}/2}}\prod_{j=1}^{n_{2}}\big(\mathbb{I} + i\sigma_{j}\chi_{j}\big)$, which achieves $\text{Tr}(H'\rho_{0}) = \sqrt{n_{2}}$. The idea is now to construct a new state $\rho_{\theta}$ obtained from $\rho_{0}$ by applying a unitary transformation to $\rho_{0}$, and to find a lower bound for the expectation value of $H^{(2)}$ w.r.t. $\rho_{\theta}$. 
\begin{equation}
    \rho_{\theta} := e^{-\theta \zeta}\rho_{0}e^{+\theta \zeta}\:,\quad \text{where }\zeta := \sum_{j=1}^{n_{2}}\tau_{j}\sigma_{j}\text{ and }\theta \in \mathbb{R}.
\label{eq:defrhotheta}
\end{equation}
The expectation value of $H^{(2)}$ w.r.t. $\rho_{\theta}$ is:
\begin{multline}
    \text{Tr}(H^{(2)}\rho_{\theta}) = \text{Tr}(H^{(2)}_{\theta}\rho_{0})\:,\\ \quad \text{where }H^{(2)}_{\theta} := e^{+\theta \zeta}H^{(2)}e^{-\theta \zeta}.
\end{multline}
Using the BCH expansion of $H_{\theta}$ and $\text{Tr}(H^{(2)}\rho_{0})=0$ , we obtain:

\begin{align}
    \text{Tr}(H^{(2)\rho_{\theta}}) \!&= \theta \: \text{Tr}([\zeta,H^{(2)}]\rho_{0})\notag\\
    &\hspace{2em}+\theta^{2}\int_{0}^{1}(1-s)\text{Tr}([\zeta,[\zeta,H^{(2)}]]\rho_{s\theta})ds \nonumber \\
    &=\: \theta \: \text{Tr}([\zeta,H^{(2)}]\rho_{0}) \notag\\
    &\hspace{2em}+ \theta^{2}\!\!\!\mathop{\mathbb{E}}\limits_{s\sim[0,1]}\Big[ (1\!-\!s)\text{Tr}([\zeta,[\zeta,H^{(2)}]]\rho_{s\theta})\Big] \nonumber \\
    &\geq \!\theta\, \text{Tr}([\zeta,\!H^{(2)}]\rho_{0}) \!-\! \theta^{2}\|[\zeta,\![\zeta,H^{(2)}]]\|,\!\!
\label{eq:BCH}
\end{align}
where we have used the triangle inequality and $\|\cdot\|$ denotes the spectral norm. To lower bound $\text{Tr}(H^{(2)}\rho_{\theta})$, one now has to (i) lower bound $\theta \: \text{Tr}([\zeta,H^{(2)}]\rho_{0})$ and (ii) upper bound $\theta^{2}\:\|\:[\zeta,[\zeta,H^{(2)}]]\:\|$. This proof technique is similar in spirit to the proof in \cite{ImprovedApproxAlg}, although their proof is for qubit Hamiltonians with bounded-degree interactions.

First, we find a lower bound for $\theta \: \text{Tr}([\zeta,H^{(2)}]\rho_{0})$ which holds with high probability:
\begin{center}
\begin{align}
    \text{Tr}([\zeta, H^{(2)}]\rho_0) =&\: \frac{i}{\sqrt{n_{2}}}\sum_{j,k=1}^{n_{2}}\text{Tr}([\tau_{j}\sigma_{j},\tau_{k}\chi_{k}]\rho_0) \nonumber \\
    =&\: \frac{i}{\sqrt{n_{2}}}\sum_{j=1}^{n_{2}}\text{Tr}([\tau_{j}\sigma_{j},\tau_{j}\chi_{j}]\rho_0) \nonumber \\
    =&\: \frac{2i}{\sqrt{n_{2}}}\sum_{j=1}^{n_{2}}\text{Tr}(\sigma_{j}\chi_{j}\:\tau_{j}^{2}\rho_0) \nonumber \\
    =&\: \frac{2}{\sqrt{n_{2}}}2^{-(n_{2}+n_{1}/2)}\sum_{j=1}^{n_{2}}\text{Tr}(\mathbb{I}_{n_{2}}\:\tau_{j}^{2}) \nonumber \\
    =&\: \frac{2(-1)^{q/2}}{\sqrt{n_{2}}\binom{n_{1}}{q-1}} \sum_{j=1}^{n_{2}}\sum_{\substack{S\subseteq [n_{1}]\\ |S| = q-1}}\big(J_{S,j}\big)^{2},
\end{align}
\end{center}
where we have used that $\text{Tr}([\tau_{j}\sigma_{j},\tau_{k}\chi_{k}]\rho_0)$ is non-zero only for $j=k$, and the definition of $\tau_{j}$. The quantity
$\text{Tr}([\zeta, H^{(2)}]\rho_0)$ is thus a chi-squared random variable (up to normalization factors and potentially a sign) with $n_{2}\binom{n_{1}}{q-1}$ degrees of freedom and its expectation value is given by:
\begin{center}
\begin{align}
    \mathbb{E}\big[\text{Tr}([\zeta, H^{(2)}]\rho_0)\big] =&\: \frac{2(-1)^{q/2}}{\sqrt{n_{2}}\binom{n_{1}}{q-1}} \sum_{j=1}^{n_{2}}\sum_{\substack{S\subseteq [n_{1}]\\ |S| = q-1}}\!\!\mathbb{E}\big[\big(J_{S,j}\big)^{2}\big] \nonumber \\
    =&\: 2\sqrt{n_{2}}\:(-1)^{q/2},
\label{eq:firstorderexp}
\end{align}
\end{center}
where we have used that $\mathbb{E}\big[\big(J_{S,j}\big)^{2}\big] = 1$. We note that in order to obtain a positive first-order contribution to $\text{Tr}(H^{(2)}\rho_{\theta})$, one should take $\theta$ positive for $q/2$ even, and one should take $\theta$ negative for $q/2$ odd. Since $\text{Tr}([\zeta, H^{(2)}]\rho_0)$ is a chi-squared random variable with $n_{2}\binom{n_{1}}{q-1}$ degrees of freedom, the following tail bounds can be obtained \cite{chisquaredtailbounds}:
\begin{equation}
    \text{Pr}\big[ \text{Tr}([\zeta,H^{(2)}]\rho_0)\!\!\leq\! \sqrt{n_{2}} \big] \leq \!\exp\!\Big(\!\! -\!\Omega \big(n_{2}\:n_{1}^{q-1}\big)\! \Big),
\label{eq:firstorderI}
\end{equation}
    for $q/2$ even, and
\begin{equation}
    \text{Pr}\big[ \text{Tr}([\zeta,H^{(2)}]\rho_0)\!\!\geq\! -\sqrt{n_{2}} \big] \leq \!\exp\!\Big(\!\! -\!\Omega \big(n_{2}\:n_{1}^{q-1}\big)\! \Big),
\label{eq:firstorderII}
\end{equation}
for $q/2$ odd.
The random variable $\text{Tr}([\zeta,H^{(2)}]\rho_0)$ is thus equal to $2\sqrt{n_{2}}(-1)^{q/2}$ in expectation and the probability that -- for any even $q\geq 4$ -- its norm is smaller than half the norm of this expectation is at most exponentially small in the system size.

In order to upper bound $\theta^{2}\:\|\:[\zeta,[\zeta,H^{(2)}]]\:\|$, we first evaluate $[\zeta,[\zeta,H^{(2)}]]$:
\begin{center}
\begin{align}
    &[\zeta,[\zeta,H^{(2)}]]\notag\\ &\hspace{1em}=\frac{i^{q/2}}{\sqrt{n_{2}\binom{n_{1}}{q-1}}}\sum_{j=1}^{n_{2}}\sum_{\substack{S\subseteq [n_{1}]\\ |S| = q-1}}J_{S,j}[\zeta,[\zeta,\phi^{S}\chi_{j}]] \nonumber \\
    &\hspace{1em}=\: \frac{i^{q/2}}{\sqrt{n_{2}\binom{n_{1}}{q-1}}}\sum_{j,k,l=1}^{n_{2}}\sum_{\substack{S\subseteq [n_{1}]\\ |S| = q-1}}J_{S,j}[\tau_{k}\sigma_{k},[\tau_{l}\sigma_{l},\phi^{S}\chi_{j}]] \nonumber \\
    &\hspace{1em}=\: \frac{i^{3q/2-2}}{\sqrt{n_{2}\binom{n_{1}}{q-1}^3}}\sum_{j,k,l=1}^{n_{2}}
    \sum_{S,S',S''}
    J_{S,j}J_{S',k}J_{S'',l}\notag\\
    &\hspace{9.5em}\times[\phi^{S'}\sigma_{k},[\phi^{S''}\sigma_{l},\phi^{S}\chi_{j}]],
\end{align}
\end{center}
where the final sum over $S,S',S''$ is over all $S,S',S''\subseteq [n_{1}]$ with $|S|=|S'|=|S''|=q-1$ (all sums over $S,S',S''$ will implicitly have this constraint from now on). The nested commutator in this expression simplifies as follows (note that the product of $i^{3q/2-2}$ and the nested commutator is Hermitian):

\begin{align}
    &i^{3q/2-2}[\phi^{S'}\sigma_{k},[\phi^{S''}\sigma_{l},\phi^{S}\chi_{j}]] \vphantom{\sum_{S}}=\notag\\
    &\begin{cases}
      C(\phi^{K}\sigma_{k}\sigma_{l}\chi_{j})_{H}, & \text{if}\hspace{1em} (|S''\cap S| \text{ is odd}) \\
      & \hspace{1em}\wedge\:(|S'\cap (S''\triangle S)|+\delta_{k,l} \text{ is odd})\\
      & \hspace{2em}\wedge\:(|S| = |S'| = |S''| = q-1), \\
      &\\
      0, & \text{otherwise,}
    \end{cases}  
\end{align}
where $(\phi^{K}\sigma_{k}\sigma_{l}\chi_{j})_{H}$ denotes a Hermitian version of $\phi^{K}\sigma_{k}\sigma_{l}\chi_{j}$ (i.e., $\phi^{K}\sigma_{k}\sigma_{l}\chi_{j}$ up to potential integer powers of $i$) and $K := (S\triangle S' \triangle S'')\cup (S\cap S'\cap S'')$ (note that $|K|$ is odd). We therefore have:

\begin{align}
    [\zeta,[&\zeta,H^{(2)}]] =\notag\\
    &C\frac{1}{\sqrt{n_{2}}}\binom{n_{1}}{q-1}^{-3/2}\!\!\!\!\sum_{j,k,l=1}^{n_{2}}\sum_{S,S',S''}\!\!\!J_{S,j}J_{S',k}J_{S'',l}\notag\\
    &\hspace{3em}\times\!\!\Big((\phi^{K}\sigma_{k}\sigma_{l}\chi_{j})_{H}f(S,S',S''\!,j,k,l)\Big),
\label{eq:A}
\end{align}
where we have defined
\begin{align}
    &f(S,S',S'',j,k,l)  
\vphantom{\sum_{S}}:=\notag\\
    &\begin{cases}
      1, & \text{if}\hspace{1em} (|S''\cap S| \text{ is odd}) \\
      & \hspace{1em}\wedge\:(|S'\cap (S''\triangle S)|+\delta_{k,l} \text{ is odd})\\
      & \hspace{2em}\wedge\:(|S| = |S'| = |S''| = q-1), \\
      &\\
      0, & \text{otherwise.}
      \end{cases}
\label{eq:indexfunction}
\end{align}

We now wish to find an upper bound on the expected value of the spectral norm of $[\zeta,[\zeta,H^{(2)}]]$. And in addition, we would like to show that the spectral norm exceeds twice the value of this upper bound with probability that is at most exponentially small in the system size. To establish this, we will have to show the following:
\begin{equation}
    \mathbb{E}\Big( \|\:[\zeta,[\zeta,H^{(2)}]]\:\|^{k} \Big) \leq \alpha^{k},
\label{expeqref}
\end{equation}
for \textit{even} $k$ proportional to the system size and for some $\alpha$. Eq. \eqref{expeqref} implies two things: First, since $\mathbb{E}\big( \|\:[\zeta,[\zeta,H^{(2)}]]\:\| \big)^{k} \leq \mathbb{E}\big( \|\:[\zeta,[\zeta,H^{(2)}]]\:\|^{k} \big)$ (using Jensen's inequality), it implies $\mathbb{E}\big( \|\:[\zeta,[\zeta,H^{(2)}]]\:\| \big) \leq \alpha$ (i.e., $\alpha$ is the upper bound on the expected value of the spectral norm). Second, applying Markov's inequality to the random variable $\|\:[\zeta,[\zeta,H^{(2)}]]\:\|$ and using Eq. \eqref{expeqref} yields
\begin{align}
    &\text{Pr}\Big[ \|\:[\zeta,[\zeta,H^{(2)}]]\:\| \geq \alpha' \Big] \!=\! \text{Pr}\Big[ \|\:[\zeta,[\zeta,H^{(2)}]]\:\|^{k}\notag\\
    &\hspace{19em}\geq (\alpha')^{k} \Big]\notag\\
    &\hspace{10em}\leq\vphantom{\sum^S} \big(\alpha/\alpha')^{k},
\end{align}
with $\alpha'\geq \alpha$.
So taking $\alpha' = 2\alpha$ and $k$ equal to the system size $2n$ ($=2n_{2}+n_{1}$) yields the desired result of the probability of the spectral norm exceeding twice the value of the upper bound being at most exponentially small in the system size.

For convenience, we define $A:=[\zeta,[\zeta,H^{(2)}]]$. Since $A$ is Hermitian (by direct calculation), the spectrum of $A^{2}$ is non-negative and therefore we have $\|A\|^{k}=\lambda_{\max}(A^{2})^{k/2}\leq \text{Tr}(A^{k})$ (for even $k$). Using Eq. \eqref{eq:A}, we express $A$ as $C\sum_{\tilde{S}\subseteq[2n_{2}+n_{1}]}Q_{\tilde{S}}C_{\tilde{S}}$ for convenience, where $C$ is a non-negative constant, $Q_{\tilde{S}}$ are real random variables, and $C_{\tilde{S}}$ denotes a Hermitian (even) Majorana monomial. In addition, we define the random variable (which is obtained by replacing Majorana monomials in $A$ with $1$)
\begin{align}
    &A(1) := C\sum_{\tilde{S}\subseteq[2n_{2}+n_{1}]}Q_{\tilde{S}}\notag\\
    &=C \frac{1}{\sqrt{n_2}}\binom{n_1}{q-1}^{-3/2}\sum_{j,k,l=1}^{n_{2}}\sum_{S,S',S''}J_{S,j}J_{S',k}J_{S'',l}\notag\\
    &\hspace{9em}\times f(S,S',S'',j,k,l).
\label{eq:A(1)}
\end{align}
If we now \textit{assume} that

\begin{align}
    \mathbb{E}\big( Q_{\tilde{S}_{1}}\ldots Q_{\tilde{S}_{k}} \big)&\geq 0 \quad \text{and} \nonumber \\
    \mathbb{E}\Big( A(1)^{k} \Big)/\alpha^{k}&\leq 1/2^{n_{2}+n_{1}/2},
\label{eq:conditions}
\end{align}
both hold for some even $k$ and some constant $\alpha$ (note that the first condition will automatically be satisfied since $\{J_{S,j}\}$ is a collection of independent standard Gaussian random variables), then for even $k$ we can establish
\begin{center}
\begin{align}
    \mathbb{E}\big( \|A\| \big)^{k} &\leq \mathbb{E}\big( \|A\|^{k} \big) \nonumber \\
    &\leq \mathbb{E}\big(\text{Tr}(A^{k})\big) \nonumber \\ 
    &=\: C^{k}\hspace{-1.0em}\sum_{\substack{\tilde{S}_{1},\ldots,\tilde{S}_{k}\\ \subseteq[2n_{2}+n_{1}]}}\hspace{-1.0em}\mathbb{E}\big(Q_{\tilde{S}_{1}}... Q_{\tilde{S}_{k}}\big) \text{Re}\big[\text{Tr}\big( C_{\tilde{S}_{1}}... C_{\tilde{S}_{k}} \big)\big] \nonumber \\
    &\leq\: 2^{n_{2}+n_{1}/2}\:C^{k}\hspace{-1.0em}\sum_{\substack{\tilde{S}_{1},\ldots,\tilde{S}_{k}\\ \subseteq[2n_{2}+n_{1}]}}\hspace{-1.0em}\mathbb{E}\big(Q_{\tilde{S}_{1}}... Q_{\tilde{S}_{k}}\big) \nonumber \\
    &=\: 2^{n_{2}+n_{1}/2}\:\mathbb{E}\big( A(1)^{k} \big) \leq \alpha^{k},
\end{align}
\end{center}
where the first inequality is again Jensen's inequality and we have also used that $\mathbb{E}\big(\text{Tr}(A^{k})\big)$ is real (since $A$ is Hermitian) and that $\text{Re}\big[\text{Tr}\big( C_{\tilde{S}_{1}}... C_{\tilde{S}_{k}} \big)\big]$ is always at most $2^{n_{2}+n_{1}/2}$ (note that $\text{Tr}\big( C_{\tilde{S}_{1}}... C_{\tilde{S}_{k}} \big)$ equals $2^{n_{2}+n_{1}/2}$ up to integer powers of $i$, but imaginary contributions vanish in the sum). This establishes Eq. \eqref{expeqref}, and thereby the desired result. Therefore, what is left is to show that the second condition in Eq. \eqref{eq:conditions} is satisfied.

From this point onward, we shall take $n_{1}$ and $n_{2}$ proportional to $n$, where $2n = 2n_{2}+n_{1}$ denotes the total number of Majorana operators. We now show that the second condition in Eq. \eqref{eq:conditions} is satisfied for $k = 2n$ and $\alpha = C\sqrt{n}$. In order to do so, we show that $\mathbb{E}\big( A(1)^{2n} \big)\leq (C\sqrt{n})^{2n}$ (where the factor of $2^{n_{2}+n_{1}/2}$ is absorbed in $C^{2n}$). To that end, we thus need to find an upper bound on the ($2n$)th moment of the random variable $A(1)$ in Eq. \eqref{eq:A(1)}. 

In Appendix \ref{app:momentbound}, we derive this upper bound and indeed show that $\mathbb{E}\big( A(1)^{2n} \big)\leq (C\sqrt{n})^{2n}$. Therefore,
\begin{multline}
\mathbb{E}\big( \|\: [\zeta,[\zeta,H^{(2)}]] \:\| \big) \leq C\sqrt{n} \quad \text{and}\\ \text{Pr}\Big[ \|\: [\zeta,[\zeta,H^{(2)}]] \:\| \geq 2C\sqrt{n} \Big] \leq \exp\big(\!-\!\Omega(n)\big),
\label{eq:higherorder}
\end{multline}
which is the desired result.

Combining Eq.~\eqref{eq:BCH}, Eqs.~ \eqref{eq:firstorderI},\eqref{eq:firstorderII} and Eq. \eqref{eq:higherorder}, we conclude that there exists a $\theta = O(1)$ such that
\begin{equation}
    \text{Tr}(H^{(2)}\rho_{\theta}) \geq C\sqrt{n},
\end{equation}
with probability at least $1-\exp\big(\!-\!\Omega(n)\big)$.
\end{proof}

What is left is to show that this result also holds for the standard SYK Hamiltonian. This translation from $2$-colored SYK Hamiltonian to standard SYK Hamiltonian is given in Lemma \ref{lem:2colouredtostandard} below, and its proof is given in Appendix \ref{app:2coltostandard}.

\begin{lemma}
For the class of $q$-local SYK Hamiltonians (with even $q\geq 4$) in Eq.~\eqref{SYK_H}, $\rho_{\theta}$ (defined in Eq.~\eqref{eq:defrhotheta}) achieves $\text{\normalfont{Tr}}(H\rho_{\theta})\geq C\sqrt{n}$ with probability at least $1-\exp\big(-\Omega(n)\big)$ over the draw of Hamiltonians, provided that $\rho_{\theta}$ achieves $\text{\normalfont{Tr}}(H^{(2)}\rho_{\theta})\geq C\sqrt{n}$ (with $H^{(2)}$ the $2$-coloured SYK Hamiltonian defined in Eq.~\eqref{SYK_2}) with probability at least $1-\exp\big(-\Omega(n)\big)$ over the draw of $2$-coloured Hamiltonians.
\label{lem:2colouredtostandard}
\end{lemma}

This also concludes the proof of Lemma \ref{lemma:maxeiglowbound}, i.e., that $\lambda_{\max}(H) = \Omega(\sqrt{n})$ with probability at least $1-\exp\big(\!-\Omega(n)\big)$ over the draw of standard SYK Hamiltonians.

\section*{Acknowledgements}
J.H. and Y.H acknowledge support from the QSC Zwaartekracht grant (NWO). Y.H., M.S. and B.M.T acknowledge support by QuTech NWO funding 2020-2024 – Part I “Fundamental Research”, project number 601.QT.001-1, financed by the Dutch Research Council (NWO). We thank V. Cheianov, R. O'Donnell, M. Hastings, J. Klassen, J. Liebert and T.E. O’Brien for discussions, as well as N. Ju and J. Jiang for spotted typos.

\newpage

\appendix
\onecolumn

\section{Extensive sets of all anti-commuting terms}
\label{app:lemAC}

One can easily prove that when one maps a dense, non-sparse, fermionic model such as the SYK model onto a qubit Hamiltonian, the locality of the resulting Hamiltonian has to grow as some function of $n$, due to the following Lemma:

\begin{lemma}
Any set of all-mutually anti-commuting Pauli strings $\{Q_i\}_{i=0}^{m-1}$, each of weight at most $k$, on $n$ qubits has cardinality $m$ bounded as
\begin{equation}
    m \leq 3 \times 2^{k (3k-1)},
    \label{eq:card}
\end{equation}
assuming that $k(k-1) < n$.
\label{lem:AC}
\end{lemma}

\begin{proof}
Take $Q_0$ of weight at most k and let $m-1$ Paulis $Q_i$ anticommute with it. We can represent each Pauli string as a $2n$-bit string $y$, say $Q_0=y_x y_z$ where the Hamming weight $|y_x|\leq k, |y_z| \leq k$. Any other $Q_i$ in the set has to anti-commute with $Q_0$ on the support of the string $y$. First, note that the set of strings of length at most $2k$ which have symplectic inner product equal to 1 (so anti-commute) to a given string of length $2k$ is at most $2^{2k-1}$. 
Now we pick the largest subset ${\cal M}_1$ of the set of elements $Q_1, \ldots Q_{m-1}$ such that all elements in the subset act identically on the support of $Q_0$, i.e. are represented by the same string of length at most $2k$ while differing beyond the support of $Q_0$.
Let the cardinality of this set be $|{\cal M}_1|=m_1 \leq m-1$ and $m_1 \geq \frac{m-1}{2^{2k-1}}\geq \frac{m}{2^{2k}}$ as the largest set should at least be a fraction $1/2^{2k-1}$ of the total.
So now we consider this set ${\cal M}_1$ and their action on the remaining $n-k$ qubits (outside the support of $Q_0$), where these elements all have to anti-commute. In addition, each element has Pauli weight at most $k-1$ (as we had to overlap with at least one Pauli with $Q_0$).
We then reapply this argument on this set, leading to a new set ${\cal M}_2$ with $|{\cal M}_2|=m_2\geq \frac{m_1-1}{2^{2(k-1)-1}}$ acting on $n-2k$ qubits and having weight $k-2$ etc. We can reiterate this process $l$ times so that the remaining weight of the set of Pauli strings ${\cal M}_l$ has $k-l=1$. This implies that ${\cal M}_l$ can contain at most 3 elements since they all need to anti-commute on a single qubit (assuming that $n- kl > 0$ or $n-k (k-1) > 0$). So we have 
\begin{equation}
3 \geq |{\cal M}_{l=k-1}|=m_{k-1} \geq \frac{m}{4^{k+(k-1)+\ldots +l}}=
\frac{m}{2^{k(3k-1)}}.
\end{equation}
\end{proof}

The SYK-$4$ model contains large (of size $n$) sets of mutually anti-commuting terms. An example is the set of all terms which only overlap on one fixed Majorana. Lemma \ref{lem:AC} then shows that any fermion-to-qubit mapping (an encoding possibly using more qubits) will require the weight of some of the resulting Pauli terms to grow as a function of $n$. Note that the actual mapping by Bravyi and Kitaev \cite{BK:ferm} with $k=O(\log n)$ shows that the upper bound in Eq.~\eqref{eq:card} is not completely tight. 

Another straightforward observation on the energy scaling of a model where all terms anti-commute is that $\lambda_{\rm max}$ does not necessarily scale with the number of terms, as captured by the following Lemma

\begin{lemma}
Let $H=\sum_{i\in \mathcal{I}} J_I C_I$ where the $\{C_I\}$ are a set of all-mutually anti-commuting Majorana operators on $[2n]$ (each $C_I$ has even support). Then \begin{equation}
    \lambda_{\rm max}(H)=\sqrt{\sum_I J_I^2}.
\end{equation}
\label{lem:ACspectrum}
\end{lemma}

\begin{proof}
We have $H=\sum_I J_I C_I=\sqrt{\sum_I J_I^2} \sum_I \beta_I C_I$ with $\sum_I \beta_I^2=1$. Take the state $\rho=\frac{1}{2^{n}}(\mathbb{I}+\sum_I \beta_I C_I)$ and thus ${\rm Tr}(H \rho)=\sqrt{\sum_I J_I^2} \sum_I \beta_I^2=\sqrt{\sum_I J_I^2}$. This is the maximal eigenvalue that can be reached since one can map each $c_I$ onto a single Majorana operator $c_{i(I)}$ as these sets form identical algebras. Then we can use the normalization of $\beta_I$ to view $\sum_i \beta_I c_{i(I)}=\tilde{c}_1$ with single Majorana operator $\tilde{c}_1$ (this is an example of the transformation in Eq.~\eqref{eq:Rtrafo}).  A single Majorana $\tilde{c}_1$ has spectrum $\pm 1$ and hence the (hugely degenerate) spectrum of $H$ is simply $\pm \sqrt{\sum_I J_I^2}$.\end{proof}

Thus, if all $J_I$ are of similar strength, we observe that the overall maximal energy scales as $\sqrt{|\mathcal{I}|}$ rather than $|\mathcal{I}|$.

\section{Splitting sparse Hamiltonians into diffuse interaction sets}
\label{sec:diffuse_splitting_general}

\begin{lemma*}[Repetition of Lemma \ref{lem:diffuse_splitting_general}]

Let ${\cal I}$ be the interaction set of a $k$-sparse $q$-local Hamiltonian on the set of fermions $[2n]$. The set ${\cal I}$ can be split into $(qQ)/2$ disjoint, strictly $2q'$-local subsets ${\cal I}^{(2q')}_\alpha$ (with $\alpha \in [Q]$ and $q'\in [q/2]$) each of which is diffuse with respect to ${\cal I}$:
\begin{align}
    {\cal I}=\bigcup^{q/2}_{q'=1} \bigcup^{Q}_{\alpha=1} {\cal I}^{(2q')}_\alpha.
\end{align}
The parameter $Q=q(q-1)(k-1)^2+q(k-1)+2$ does not grow with $n$. The construction of this splitting can be done efficiently, in time $\mathrm{poly}(n)$.
\end{lemma*}

\begin{proof}

Consider a graph ${\cal G}$ with vertices corresponding to interaction sets $I\in{\cal I}$, where two interaction sets $I_1, I_2$ are connected with an edge if either 1. they share at least one Majorana operator or 2. $I_1$ and $I_2$ both share Majorana operators with another set $I'\neq I_1, I_2$. For a $q$-local $k$-sparse Hamiltonian, ${\cal G}$ has maximal degree $Q'$ with $Q'=q(q-1)(k-1)^2+q(k-1)$. Here $q(k-1)$ is the maximal number of interactions $I_2$ directly sharing a Majorana fermion with any given interaction $I_1$, and $q(q-1)(k-1)^2$ is the maximal number of interactions satisfying condition 2. Since a $Q'$-sparse graph is vertex-colorable by at most $(Q'+1)$ colors \cite{Bollobas}, we can split ${\cal I}$ into $(Q'+1)$ subsets ${\cal I}_\alpha$, s.t. any two interactions $I_1,I_2$ from a set ${\cal I}_\alpha$ are not connected by an edge in ${\cal G}$. By definition of ${\cal G}$, this amounts to sets ${\cal I}_\alpha$ satisfying the first two conditions of Definition \ref{def:diffuse}. A greedy algorithm can be used to assign the vertices ${\cal G}$ with $(Q'+1)$ colors, so ${\cal I}_\alpha$ can be constructed efficiently. 

Each interaction set ${\cal I}_\alpha$ can contain terms of different weight.  For each value of $\alpha$ we define strictly $2q'$-local sets ${\cal I}^{(2q')}_\alpha$ (for $q'=1,..,q/2$) by restricting to the  strictly $2q'$-local part of ${\cal I}_\alpha$. This gives a splitting of ${\cal I}$ into efficiently constructable subsets  ${\cal I}^{(2q')}_\alpha$:
\begin{align}
    {\cal I}=\bigcup^{q/2}_{q'=1} \bigcup^{Q'+1}_{\alpha=1} {\cal I}^{(2q')}_\alpha,
\end{align}
where all sets ${\cal I}^{(2q')}_\alpha$ satisfy conditions $1$ and $2$ in Definition \ref{def:diffuse}.

The rest of the proof is concerned with the third condition of a diffuse set in Definition \ref{def:diffuse}, for all sets ${\cal I}^{(2q')}_\alpha$. This means ensuring that for all values of $\alpha$ and $q'$, the support size $|\mathrm{Sup}({\cal I}^{(2q')}_\alpha)|$ is smaller than $2n\frac{q}{q+1}$ . Fix $q'$ and consider sets ${\cal I}^{(2q')}_\alpha$ for  $\forall \alpha \in [Q'+1]$. 

Consider the case where $|\mathrm{Sup}({\cal I}^{(2q')}_\alpha)|< 2n\frac{q}{q+1}$ does not hold for at least one value of $\alpha$, which we set to be $\alpha=Q'+1$ without loss of generality.

Let us prove that the violation $|\mathrm{Sup}({\cal I}_\beta^{(2q')})|\geq 2n\frac{q}{q+1}$ cannot hold for any $\beta\neq Q'+1$. 
Firstly, no interaction $I$ from ${\cal I}^{(2q')}_\beta$ can be a strict subset of an interaction in ${\cal I}^{(2q')}_{Q'+1}$ or share Majoranas with two terms in ${\cal I}^{(2q')}_{Q'+1}$ simultaneously. The first scenario is excluded since ${\cal I}^{(2q')}_{Q'+1}$ and ${\cal I}^{(2q')}_{\beta}$ are both strictly $2q'$-local and the second scenario is excluded because ${\cal I}^{(2q')}_{Q'+1}$ satisfies condition 2 of Definition \ref{def:diffuse}. 
From these two facts it follows that each interaction in ${\cal I}^{(2q')}_\beta$ must involve at least one Majorana from $[2n] \backslash \mathrm{Sup}({\cal I}^{(2q')}_{Q'+1})$. This implies
\begin{align}
|\mathrm{Sup}({\cal I}^{(2q')}_\beta)|\leq2q'|[2n]\backslash\mathrm{Sup}({\cal I}^{(2q')}_{Q'+1})|,
\end{align}
This can be further bounded as $|\mathrm{Sup}({\cal I}^{(2q')}_\beta)|\leq q |[2n]\backslash\mathrm{Sup}({\cal I}^{(2q')}_{Q'+1})|$, because $2q'\leq q$. Since we assumed $|\mathrm{Sup}({\cal I}^{(2q')}_{Q'+1})|\geq 2n q/(q+1)$ and thus $|[2n]\backslash\mathrm{Sup}({\cal I}^{(2q')}_{Q'+1})|<2n /(q+1)$, it follows that $|\mathrm{Sup}({\cal I}^{(2q')}_\beta)|\leq q|[2n]\backslash\mathrm{Sup}({\cal I}^{(2q')}_{Q'+1})|<2n q /(q+1)$. Thus we have shown that for a given $q'$, the condition 3 of Definition \ref{def:diffuse} -- indeed cannot be violated by more than one ${\cal I}^{(2q')}_{\alpha}$.

Consider all $q'$ for which there exists a violation $|{\rm Sup}({\cal I}^{(2q')}_{Q'+1})|\geq 2n\frac{q}{q+1}$. Since $\frac{q}{q+1}>\frac{1}{2}$ for any $q$, this violation can be fixed by splitting ${\cal I}^{(2q')}_{Q'+1}$ in half. Introduce non-overlapping sets ${\tilde{\cal I}}^{(2q')}_{Q'+1}$ and $\tilde{\tilde{\mathcal{I}}}^{(2q')}_{Q'+1}$ of sizes $\lfloor|{\cal I}^{(2q')}_{Q'+1}|/2\rfloor$ and $\lceil|{\cal I}^{(2q')}_{Q'+1}|/2\rceil$: ${\cal I}^{(2q')}_{Q'+1}={\tilde{\cal I}}^{(2q')}_{Q'+1}\cup \tilde{\tilde{\mathcal{I}}}^{(2q')}_{Q'+1}$. By implication, $|\mathrm{Sup}(\tilde{\tilde{\mathcal{I}}}^{(2q')}_{Q'+1})| \leq  2n/2\leq 2n q/(q+1)$ and similarly $|\mathrm{Sup}({\tilde{\mathcal{I}}}^{(2q')}_{Q'+1})| \leq 2nq/(q+1)$. We conclude the construction by modifying the set ${\cal I}^{(2q')}_{\alpha}$ for the considered $q'$: we redefine ${\cal I}^{(2q')}_{Q'+1}\equiv{\tilde{\mathcal{I}}}^{(2q')}_{Q'+1}$, and introduce one extra interaction set ${\cal I}^{(2q')}_{Q'+2}\equiv\tilde{\tilde{\mathcal{I}}}^{(2q')}_{Q'+1}$. 

The proof can now be finalized. Performing the above procedure for all $q'$ where a violation was present, and completing the $\{{\cal I}^{(2q')}_{\alpha}\}$ without such violations with ${\cal I}^{(2q')}_{\alpha=Q'+2}=\emptyset$, we arrive at the splitting
\begin{align}
    {\cal I}=\bigcup^{q/2}_{q'=1} \bigcup^{Q}_{\alpha=1} {\cal I}^{(2q')}_\alpha,
\end{align}
where $Q=Q'+2=q(q-1)(k-1)^2+q(k-1)+2$. Interaction sets ${\cal I}^{(2q')}_\alpha$ are diffuse (satisfying all three conditions of Definition \ref{def:diffuse}) with respect to ${\cal I}$ for all $q'$ and $\alpha$. The construction of ${\cal I}^{(2q')}_\alpha$ is efficient, because each step can be implemented in time $\rm{poly}(n)$.

\end{proof}

\begin{figure}[t]
    \centering
    \includegraphics[width=\linewidth]{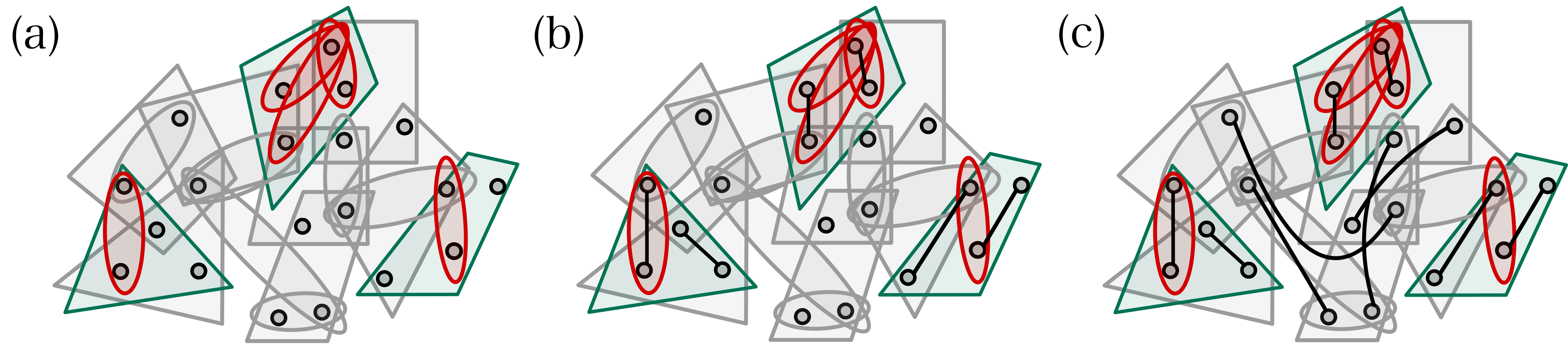}
    \caption{Example of the construction from the proof of Lemma \ref{lem:diffuse_matching_general}. (a) The $q$-local set of interactions $\mathcal{I}$ ($q=4$). Highlighted in green is diffuse and strictly $q'$-local $\mathcal{I}'$ ($q'=4$), in red are the interactions in $\rm{Sup(\mathcal{I}')}$ with weight less than $q'$, in grey are the rest of interactions in $\mathcal{I}$. The goal is to create a matching $M$ consistent with green-colored terms, inconsistent with grey-colored terms, and with no guaranteed relation to the red-colored terms. (b) Matching $M'$ on $\rm{Sup(\mathcal{I}')}$, consistent with $\mathcal{I}'$ by construction. Ensuring inconsistency with all red-colored terms is in general impossible. For example, consider the three overlapping red-colored terms at the top center. (c) Completing $M=M'\cup M''$ with a matching $M''$ on $[2n]\backslash\rm{Sup(\mathcal{I}')}$, ensuring inconsistency with all grey-colored terms. For this, the vertices are matched only if they belong to different interactions.}
    \label{fig:illustration_matching}
\end{figure}

\section{Majorana matchings from diffuse interaction sets}
\label{sec:diffuse_matching_general}

\begin{lemma*}[Repetition of Lemma \ref{lem:diffuse_matching_general}]
Let a strictly $q'$-local $\cal{I}'$ be diffuse w.r.t. $q$-local $k$-sparse $\cal{I}$ on $[2n]$, such that $n>(q^2-1)k$. One can efficiently construct a matching $M$ of $[2n]$ that is consistent with $\cal{I}'$ and inconsistent with all interactions $I \in \mathcal{I}\backslash \mathcal{I}'$ such that (1) $|I|\geq q'$ or (2) $I\not\subset \mathrm{Sup}(\cal{I}')$. 
\end{lemma*}

\begin{proof}

We first note that for $I\in \mathcal{I}\backslash{{\cal I}'}$ the condition $|I|\geq q'$ implies $ I\not \subset\rm{Sup}({\cal I}')$. Indeed, there are two possible options for $I\in \mathcal{I}\backslash{{\cal I}'}$ such that $I\subset\rm{Sup}({\cal I}')$.
The first option is that $I$ is a strict subset of a single interaction from ${\cal I}'$. However, this is not possible given $|I|\geq q'$, because ${\cal I}'$ is $q'$-local.
The second option is for $I$ to share Majorana modes with two or more interactions in ${\cal I}'$. This is ruled out because ${\cal I}'$ is diffuse with respect to $\mathcal{I}$ (cf. Condition 2 in Definition \ref{def:diffuse}).
The above implies that it is sufficient to construct the matching $M$ that is consistent with ${\cal I}'$ and inconsistent with $\{I\in \mathcal{I}\backslash{{\cal I}'}|I\not\subset\rm{Sup}({\cal I}')\}$.

We construct $M$ in two steps. First we construct a matching ${M'}$ of $\rm{Sup}({\cal I}')$ (note $|\rm{Sup}({\cal I}')|$ is always even). Next, we construct a matching ${M''}$ of the remaining Majorana modes $[2n]\backslash \rm{Sup}({\cal I}')$. The desired matching of $[2n]$ is the union $M=M' \cup M''$.

To construct $M'$, we match vertices of each $I\in{\cal I}'$ in an arbitrary way: for every such $I=\{i_1,..i_{q'}\}$, $\{i_{2l-1}, i_{2l}\}\in M'$ for $l\in[1,..q'/2]$. This matching is always possible, since ${\cal I}'$ is diffuse and thus different interactions from ${\cal I}'$ do not overlap. 
Thus constructed $M'$ (and therefore also $M=M'\cup M''$) is explicitly consistent with all $I\in{\cal I}'$ .

To construct a matching $M''$ of $[2n]\backslash\rm{Sup}({\cal I}')$, we aim to ensure that no $(m_1,m_2)\in M''$ is a subset of any interaction in $\mathcal{I}$. For this, consider a `permitted edge’ graph ${\cal P}$ with vertices $[2n]\backslash\rm{Sup}({\cal I}')$, and edges inserted between every pair $(i_1,i_2)$ \textit{unless} they belong to the same interaction in $\mathcal{I}$. We aim to construct $M''$ as a perfect matching of ${\cal P}$. Note that since $\mathcal{I}$ is $q$-local and $k$-sparse, the graph ${\cal P}$ has degree bounded from below as $|[2n]\backslash{\rm Sup}({\cal I}')|-(q-1)k$. At the same time, since ${\cal I}'$ is diffuse, we're guaranteed by Condition 3 in Definition~\ref{def:diffuse} that $|[2n]\backslash{\rm Sup}({\cal I}')|\geq \frac{2n}{q+1}$. Therefore, since $n>(q^2-1)k$ by assumption, the degree of the vertices in ${\cal P}$ is lower bounded as $|[2n]\backslash{\rm Sup}({\cal I}')|-(q-1)k\geq |[2n]\backslash{\rm Sup}({\cal I}')|/2 +\frac{n}{q+1}-(q-1)k>|[2n]\backslash{\rm Sup}({\cal I}')|/2$. Given this lower bound, we apply Dirac's theorem~\cite{dirac1952some}, which yields an efficiently constructable Hamiltonian cycle in the graph ${\cal P}$. Matching $M''$ is then obtained by pairing the sequential vertices in this cycle, making it a perfect matching of ${\cal P}$. By definition of ${\cal P}$, $M''$ is guaranteed to contain at least one outgoing edge from every interaction in $\{I\in \mathcal{I}\backslash{{\cal I}'}|I\not\subset\rm{Sup}({\cal I}')\}$. This makes $M=M'\cup M''$ inconsistent with $\{I\in \mathcal{I}\backslash{{\cal I}'}|I\not\subset\rm{Sup}({\cal I}')\}$, as desired.

\end{proof}

Lemma~\ref{lem:diffuse_matching}, which is used in the proof of Theorem \ref{thm:sparse}, is a special case of Lemma~\ref{lem:diffuse_matching_general}. To obtain Lemma~\ref{lem:diffuse_matching}, one sets $q'=q$ and considers strictly $q$-local $\mathcal{I}$ instead of simply $q$-local. In this case all terms in $\mathcal{I}\backslash \mathcal{I}'$ satisfy the first condition of the Lemma, and therefore the constructed $M$ is inconsistent with the entirety of $\mathcal{I}\backslash \mathcal{I}'$.

\section{Matchings and Gaussian states}
\label{sec:matching_gaussian_general}

\begin{lemma*}[Repetition of Lemma \ref{lem:matching_gaussian_general}]
Let $H=\sum_{I\in {\cal I}} J_I C_I$ on $[2n']$ be $q$-local and ${\cal I}'$ be a diffuse subset of ${\cal I}$. Consider a matching $M$ of $[2n']$. If $M$ is consistent with ${\cal I}'$ and inconsistent with ${\cal I}\backslash{\cal I}'$, one can efficiently construct a Gaussian state $\rho_{{\cal I}'}$ with the property:
\begin{equation}
    {\rm Tr} (H \rho({\cal I}'))=\sum_{I\in {\cal I}'} |J_I|.
    \label{eq:targ}
\end{equation}
\end{lemma*}
\begin{proof}
For the given matching $M$, consider its associated Gaussian state pure $\rho (M, \vec{\lambda})$ of the form:
\begin{align}
\rho (M, \vec{\lambda})=\frac{1}{2^{n}}\Pi_{\scriptscriptstyle\{m_1,m_2\}\in M} (\mathbb{I}+i \lambda_{(m_1,m_2)} c_{m_1} c_{m_2}).
\end{align}
Lemma \ref{lem:consistent_expectation} implies that the contribution to ${\rm Tr}(H\rho (M, \vec{\lambda}))$ from inconsistent interactions ${\cal I}\backslash{\cal I}'$ vanishes and contributions from ${\cal I}'$ yield:
\begin{equation}
\label{eq:H_expectation_lambda}
    {\rm Tr}(H\rho(M, \vec{\lambda}))= \sum_{I\in {\cal I}'} J_I {\rm sign}(\pi) \prod_{l\in \{1,..,|I|/2\}} \lambda_{(i_{\pi(2l-1)}, i_{\pi(2l)})}.
\end{equation}
The proof is completed by choosing an appropriate value for $\vec{\lambda}$. Since ${\cal I}'$ is diffuse, by Condition 1 of Definition~\ref{def:diffuse}, distinct interactions from ${\cal I}'$ do not share Majorana fermions. This means that the values $\lambda_{(m_1,m_2)}$ for different $I$ in Eq.~\eqref{eq:H_expectation_lambda} can be chosen independently. In particular, by picking appropriate $\lambda_{(m_1,m_2)}=\pm 1$, one can eliminate the sign of $J_I{\rm sign}(\pi)$ and achieve a contribution $|J_I|$ for each $I \in {\mathcal I}'$. Note that this procedure can be done efficiently, as it is simply a matter of choosing at most $n$ $\pm 1$ values 
by checking the sign of most $|\mathcal{I}'|$ terms. Denoting the thus chosen $\rho(M, \vec{\lambda})$ as $\rho({\cal I}')$, this yields Eq.~\eqref{eq:targ}.
\end{proof}

A special case of Lemma \ref{lem:matching_gaussian_general} is Lemma \ref{lem:matching_gaussian} used in the proof of Theorem \ref{thm:sparse}.

\section{Concentration bounds for sparse SYK-\texorpdfstring{$4$}{4}}
\label{app:SSYK_concentration}

Here we derive the concentration bounds for the \SSSYKf Hamiltonian that were used in the proof of Theorem~\ref{thm:sparse_SYK} (Section \ref{sec:sparse_SYK}). We first prove an auxiliary Lemma that will be used later in this Section, allowing to separate the statistics of interaction selection and interaction strength:

\begin{lemma}
\label{lem:bern_gaus_split}
For $a\in[D]$, let $X_a$ be i.i.d. Bernoulli random variables $X_a\sim {\rm Bern}(p)$ and $J_a$ i.i.d. Gaussian random variables $J_a\sim {\rm N}(0,1)$. Then for any integer $d\in [D]$

\begin{align}
    \mathbb{P}\left[\sum^D_{a=1} X_a |J_a|< y \right]&\leq  \mathbb{P}\left[\sum^D_{a=1} X_a < d \right]+\mathbb{P}\left[\sum^d_{a=1} |J_a| < y \right]\label{eq:bern_gaus_split_1},\\
    \mathbb{P}\left[\sum^D_{a=1} X_a |J_a|> y \right]& \leq \mathbb{P}\left[\sum^D_{a=1} X_a > d \right]+\mathbb{P}\left[\sum^d_{a=1} |J_a| > y \right].\label{eq:bern_gaus_split_2}
\end{align}
\end{lemma}
\begin{proof}

To prove Eq.~\eqref{eq:bern_gaus_split_1}, first show

\begin{align}
    \mathbb{P}\left[\sum^D_{a=1} X_a |J_a|\geq y \right]&=\sum^D_{d'=1}\left(\mathbb{P}\left[\sum^D_{a=1} X_a = d' \right]\mathbb{P}\left[\sum^{d'}_{a=1} |J_a| \geq y \right]\right)\notag \\
    &{\geq } \sum^D_{d'=d}\left(\mathbb{P}\left[\sum^D_{a=1} X_a = d' \right]\mathbb{P}\left[\sum^{d'}_{a=1} |J_a| \geq y \right]\right)\notag\\
    &{\geq } \left(\sum^D_{d'=d}\mathbb{P}\left[\sum^D_{a=1} X_a = d' \right]\right)\mathbb{P}\left[\sum^{d}_{a=1} |J_a| \geq y \right]\notag\\
    &= \mathbb{P}\left[\sum^D_{a=1} X_a \geq d \right]\mathbb{P}\left[\sum^d_{a=1} |J_a| \geq y \right] \label{eq:bern_gaus_presplit_1}.
\end{align}
It follows that
\begin{align}
    \mathbb{P}\left[\sum^D_{a=1} X_a |J_a|< y \right]&=1-\mathbb{P}\left[\sum^D_{a=1} X_a |J_a|\geq y \right]\notag\\
    & {\leq } 1-\mathbb{P}\left[\sum^D_{a=1} X_a \geq d \right]\mathbb{P}\left[\sum^d_{a=1} |J_a| \geq y \right]\notag\\
    &=1-\left(1-\mathbb{P}\left[\sum^D_{a=1} X_a < d \right]\right)\left(1-\mathbb{P}\left[\sum^d_{a=1} |J_a| < y \right]\right)\notag\\
    & {\leq } \mathbb{P}\left[\sum^D_{a=1} X_a < d \right]+\mathbb{P}\left[\sum^d_{a=1} |J_a| < y \right]. \label{eq:bern_gaus_split_1_derived}
\end{align}

This ends the proof of Eq.~\eqref{eq:bern_gaus_split_1}. In the same vein, one derives Eq.~\eqref{eq:bern_gaus_split_2}. Namely, we first have (cf. Eq.~\eqref{eq:bern_gaus_presplit_1}):
\begin{align}
    \mathbb{P}\left[\sum^D_{a=1} X_a |J_a|\leq y \right]&\geq \sum^d_{d'=1}\left(\mathbb{P}\left[\sum^D_{a=1} X_a = d' \right]\mathbb{P}\left[\sum^{d'}_{a=1} |J_a| \leq y \right]\right)\notag\\
    &\geq
    \mathbb{P}\left[\sum^D_{a=1} X_a \leq d \right]\mathbb{P}\left[\sum^d_{a=1} |J_a| \leq y \right] \label{eq:bern_gaus_presplit_2}.
\end{align}

Similarly to  Eq.~\eqref{eq:bern_gaus_split_1_derived}, one obtains Eq.~\eqref{eq:bern_gaus_split_2} from Eq.~\eqref{eq:bern_gaus_presplit_2}:

\begin{align}
    \mathbb{P}\left[\sum^D_{a=1} X_a |J_a|> y \right]&=1-\mathbb{P}\left[\sum^D_{a=1} X_a |J_a|\geq y \right]\notag\\
    &\leq 1-\left(1-\mathbb{P}\left[\sum^D_{a=1} X_a > d \right]\right)\left(1-\mathbb{P}\left[\sum^d_{a=1} |J_a| > y \right]\right)\notag\\
    &\leq \mathbb{P}\left[\sum^D_{a=1} X_a > d \right]+\mathbb{P}\left[\sum^d_{a=1} |J_a| > y \right]. \label{eq:bern_gaus_split_2_derived}
\end{align}

\end{proof}

We proceed with the proof of Lemmas \ref{lem:H_SSYK} and \ref{lem:h_kappa}, which were used in Section \ref{sec:sparse_SYK} to prove Theorem \ref{thm:sparse_SYK}.

\begin{lemma*}
[Repetition of Lemma \ref{lem:H_SSYK}]
Let interactions $\mathcal{I}$ and interaction strengths $\{J_I\}$ be those of the \SSSYKf model with average degree $k$. 
With probability at least $1-2e^{-\frac{kn}{32}}$ we have
\begin{equation}
    \sum_{I\in \mathcal{I}} |J_I|\geq kn/8.
\end{equation}
\end{lemma*}
\begin{proof}
The random variable $\sum_{I\in \mathcal{I}} |J_I|$ is a function of two sets of random variables. The first set is $X_I\in\{0,1\}$ for all possible 4-Majorana interactions $I\subset [2n],\, |I|=4$, indicating the presence of $I$ in $\mathcal{I}$. Denoting
\begin{align}
D\equiv\begin{pmatrix}2n-1\\ 3\end{pmatrix},
\end{align} 
$X_I$ is drawn from a Bernoulli distribution with probability $p=k D^{-1}$, i.e.
$X_I\sim {\rm Bern}\left(k D^{-1}\right)$. 
The second set is $J_I$ for all $I\in \mathcal{I}$, distributed normally $J_I\sim {\rm N}(0,1)$.
We introduce auxiliary variables $J_a$ for $a\in[\lceil kn/4\rceil ]$ and $J_a\sim {\rm N}(0,1)$. Then by Lemma \ref{lem:bern_gaus_split}:
\begin{align}
\label{eq:total_two_parts}
    \mathbb{P}\left[\sum_{I\in \mathcal{I}} |J_I|<kn/8\right]\leq\mathbb{P}\left[\sum_{I\subset [2n],\, |I|=4} X_I < \lceil\frac{kn}{4}\rceil\right]+\mathbb{P}\left[ \sum_{a\in[\lceil kn/4\rceil ]} |J_a|<kn/8\right].
\end{align}
We can bound the first term using the Chernoff bound for sums of Bernoulli random variables. Substituting $\binom{2n}{4}=\binom{2n-1}{3}\frac{n}{2}$, we get
\begin{align}
\label{eq:total_degree}
    \mathbb{P}\left[\sum_{I\subset [2n],\, |I|=4} X_I < \lceil\frac{kn}{4}\rceil\right]=\mathbb{P}\left[\sum_{I\subset [2n],\, |I|=4} X_I < \frac{kn}{4}\right] \leq \exp\left( -\frac{kn}{4} (1-\log 2) \right).
\end{align}
On the other hand, standard concentration properties of Gaussian random variables imply, see Lemma~\ref{lem:abs_gauss} at the end of this Appendix,
\begin{align}
\label{eq:total_weight}
    \mathbb{P}\left[\sum^{\lceil kn/4\rceil}_{a=1} |J_a| < \frac{kn}{8}\right] < e^{-kn/32}.
\end{align}
Since $\exp\left( -\frac{kn}{4} (1-\log 2) \right)\leq e^{-kn/32}$, the bound in Eq.~\eqref{eq:total_two_parts} yields
\begin{align}
    \mathbb{P}\left[\sum_{I\in \mathcal{I}} |J_I|<kn/8\right]<2e^{-kn/32},
\end{align}
as desired.

\end{proof}

\begin{lemma*}[Repetition of Lemma \ref{lem:h_kappa}]
If $k'\geq e^2 k+1$, we have with probability at least $1- 2\exp\left[-\frac{e^{-2k'}k^3}{64(k'-1)} n\right]$ that
\begin{equation}
\sum_{I\in \bar{\mathcal{I}}^{(k')}} |J_I| \leq \frac{4k^2}{\sqrt{k'-1}}e^{-k'}n.
\end{equation}
\end{lemma*}

\begin{proof}
The random variable $\sum_{I\in \bar{\mathcal{I}}^{(k')}} |J_I|$ is a function of random variables $X_I\sim {\rm Bern}\left(k D^{-1}\right)$ for $I\subset [2n],\, |I|=4$ and $J_I\sim {\rm N}(0,1)$ for all $I\in \mathcal{I}$.
We introduce auxiliary random variables $J'_a\sim {\rm N}(0,1)$ for $a\in[K]$ where 
\begin{align}
K\equiv \lfloor\frac{4k^2}{\sqrt{k'-1}}e^{-k'}n\rfloor.
\end{align}
By Lemma \ref{lem:bern_gaus_split}, one can upperbound 
\begin{align}
\label{eq:tail_two_parts}
    \mathbb{P}\Big[\sum_{I\in \bar{\mathcal{I}}^{(k')}} |J_I|> \frac{4k^2}{\sqrt{k'-1}}e^{-k'}n\Big]\;\leq\; \mathbb{P}\Big[|\bar{\mathcal{I}}^{(k')}|>\frac{4k^2}{\sqrt{k'-1}}e^{-k'}n\Big]+\mathbb{P}\Big[\sum_{a\in[K]} |J'_a|> \frac{4k^2}{\sqrt{k'-1}}e^{-k'}n\Big].
\end{align}
We now proceed with upper bounding $\mathbb{P}\big[|\bar{\mathcal{I}}^{(k')}|>\frac{4k^2}{\sqrt{k'-1}}e^{-k'}n\big]$ and then $\mathbb{P}\big[\sum_{a\in[K]} |J'_a|> \frac{4k^2}{\sqrt{k'-1}}e^{-k'}n\big]$.

To bound $\mathbb{P}[|\bar{\mathcal{I}}^{(k')}|>\frac{4k^2}{\sqrt{k'-1}}e^{-k'}n]$, we introduce the Majorana degree function $k_i=k_i(\{X_I\})\in[D]$, which is a random variable that counts the number of interactions in $\mathcal{I}$ involving a given Majorana $c_i$. Since $X_I\sim {\rm Bern}\left(k D^{-1}\right)$, $k_i$ follows the binomial distribution ${\rm Bin}(D, k D^{-1})$ (note however that different $k_i$ and $k_j$ are not necessarily independent). Given the construction of $h^{(k')}$, it is clear that $|\bar{\mathcal{I}}^{(k')}|$ can be bounded by the `excess degree' summed over all Majoranas. Concretely, using the Majorana degree function $k_i$ we define a random variable
\begin{align}
Z\equiv Z(\{X_I\})\equiv \frac{1}{2n} \sum^{2n}_{i=1} (k_i-k') \,\mathbb{I}_{k_i>k'},
\end{align}
which has the immediate property
\begin{align}
|\bar{\mathcal{I}}^{(k')}|\leq 2n Z. \label{eq:excess_degree_bound}
\end{align}
Here we used the indicator function $\mathbb{I}_{k_i>k'}=1$ when $k_i > k'$ and $0$ otherwise.
Given Eq.~\eqref{eq:excess_degree_bound}, $\mathbb{P}[Z>\frac{2k^2}{\sqrt{k'-1}}e^{-k'}]\geq\mathbb{P}[|\bar{\mathcal{I}}^{(k')}|> \frac{4k^2}{\sqrt{k'-1}}e^{-k'} n]$ and thus it suffices to bound the former. We begin by calculating its mean:
\begin{align}
\label{eq:EZ_via_tail}
\mathbb{E}[Z] = \frac{1}{2n} \sum_{i=1}^{2n} \mathbb{E} [ (k_i-k') \mathbb{I}_{k_i>k'}] = \mathbb{E} [ (k_1-k') \mathbb{I}_{k_1>k'} ],
\end{align}
where we used linearity of $\mathbb{E}(.)$ and the permutation symmetry of the SSYK ensemble. Hence we now need to calculate $\mathbb{E} [ (k_1-k') \mathbb{I}_{k_1>k'} ]$ for a single Majorana (w.l.o.g. $c_1$). Since the associated degree $k_1\sim {\rm Bin}(D, k D^{-1})$, we calculate directly (denoting $p=k D^{-1}$):
\begin{align}
\label{eq:binom_tail}
    \mathbb{E} [ (k_1-k') \mathbb{I}_{k_1>k'} ]&\leq\mathbb{E} [ k_1 \mathbb{I}_{k_1>k'} ]=\sum^{D}_{k_1=k'+1} p^{k_1}(1-p)^{D-k_1}\frac{D!}{(D-k_1)! k_1!} k_1\notag\\
    &=Dp\sum^{D-1}_{k_1=k'} p^{k_1}(1-p)^{D-1-k_1}\frac{(D-1)!}{(D-1-k_1)! k_1!}.
\end{align}
The following identity holds \cite{SA:book}:
\begin{align}
\label{eq:binom_cdf}
    \sum^{z-1}_{x=y} w^{x}(1-w)^{z-1-x}\frac{(z-1)!}{(z-1-x)! x!}=\beta_{w}(y, z-y),
\end{align}
where $\beta_{w}(y, z-y)$ is the regularized incomplete beta function. For integer $y, z>y$ it is defined as \[\beta_{w}(y, z-y)=\frac{(z-1)!}{(y-1)!(z-y-1)!}\int^w_{t=0} t^{y-1} (1-t)^{z-y-1} dt. \] Using the Stirling bound $x!\geq \sqrt{2\pi} x^{x+1/2} e^{-x}, x\in\mathbb{N}$, one bounds $\beta_{w}(y, z-y)$ as:
\begin{align}
\label{eq:beta_bound}
    \beta_{w}(y, z-y) < \frac{w(z-1)}{\sqrt{2\pi (y-1)}}\left(e\frac{w(z-1)}{(y-1)}\right)^{y-1}.
\end{align}
Substituting $p=k D^{-1}$ and using Eqs.~\eqref{eq:binom_cdf}, \eqref{eq:beta_bound} in Eq.~\eqref{eq:binom_tail} for $k'>e^2k+1$ we obtain
\begin{align}
    \mathbb{E} [ (k_1-k') \mathbb{I}_{k_i>k'} ]&< D p \frac{p (D-1)}{\sqrt{2\pi (k'-1)}}\left(e\frac{p (D-1)}{k'-1}\right)^{k'-1} \notag\\
    &<\sqrt{\frac{ e^2 k^4}{2\pi (k'-1)}}e^{-k'} \notag\\
    \Rightarrow \mathbb{E} [Z]&<\sqrt{\frac{ e^2 k^4}{2\pi (k'-1)}}e^{-k'}. \label{eq:E_Z_bound}
\end{align}
We now aim to apply the Efron-Stein inequality \cite{bouch:concentration} to bound deviations from the mean $\mathbb{E}(Z)$. For this, we introduce an additional set of independent random variables $\{X'_I\}$ such that $X'_I\sim {\rm Bern}\left(k D^{-1}\right)$. This allows to define auxiliary functions $Z'_I$
\begin{align}
   Z'_I \equiv Z\big|_{X_I\rightarrow X'_I} 
\end{align}
where for a single interaction $I$ only, the variable $X_I$ is replaced by $X'_I$. Using the indicator function $\mathbb{I}_{Z>Z'_I}$, a further auxiliary function $V=V(\{X_I\})$ can be defined:
\begin{align}
\label{eq:efron_V}
    V\equiv \mathbb{E}_{\{X'_I\}} \left[\sum_{I\subset [2n],\; |I|=4} (Z-Z'_I)^2 \mathbb{I}_{Z>Z'_I}\right],
\end{align}
where the averaging is performed over the additional random variables $\{X'_I\}$ alone.
An exponential version of the Efron-Stein inequality (Theorem 2 of \cite{bouch:concentration}) states for all $\theta>0$ and $\lambda\in(0,\theta^{-1})$:
\begin{align}
\label{eq:exp_efron}
    \log \mathbb{E}[\exp(\lambda (Z-\mathbb{E}[Z]))]\leq 
    \frac{\lambda\theta}{1-\lambda\theta} \log \mathbb{E}\left[\exp \left(\frac{\lambda V}{\theta}\right)\right].
\end{align}
To employ Eq.~\eqref{eq:exp_efron}, we have to bound $\mathbb{E}\left[\exp \left(\frac{\lambda V}{\theta}\right)\right]$. First we upper bound $V(\{X_I\})$ as a function. For all interactions $I$ we claim, independent of $\{X_I\}$ and $\{X_I'\}$:
\begin{align}
\label{eq:efron_term_bound}
    (Z-Z'_I)^2 \mathbb{I}_{Z>Z'_I} \leq \frac{4}{n^2} ~\mathbb{I}_{X_I=1}.
\end{align}
To show this, we will go through four possible cases: $(X_I, X'_I)=(0,0)$, $(1,0)$, $(0,1)$, or $(1,1)$. If $X_I=X_I'$, the left hand side of Eq.~\eqref{eq:efron_term_bound} vanishes, reproducing Eq.~\eqref{eq:efron_term_bound} for the cases $(X_I, X'_I)=(0,0)$ and $(1,1)$. For $(X_I, X'_I)=(0,1)$, $Z$ is smaller than $Z'_I$, because replacing $X_I=0$ by $X'_I=1$ cannot decrease the excess degree for any Majorana (cf. definition of $k_i$ and $2n Z=\sum^{2n}_{i=1} (k_i-k') \mathbb{I}_{k_i>k'}$). Due to the factor $\mathbb{I}_{Z>Z'_I}$, $(Z-Z'_I)^2 \mathbb{I}_{Z>Z'_I}$ in this case is zero, in agreement with Eq.~\eqref{eq:efron_term_bound}. The last case is $(X_I, X'_I)=(1,0)$. As any interaction $I$ only involves $4$ fermions, the reduction of total excess degree $2n(Z\big|_{X_I=1}-Z\big|_{X_I=0})$ is at most equal to $4$, independent of the rest of the variables $\{X_I\}$. Therefore $(Z-Z'_I)^2 \mathbb{I}_{Z>Z'_I}$ for $(X_I, X'_I)=(1,0)$ is at most equal to $\frac{4}{n^2}$, proving Eq. \eqref{eq:efron_term_bound}. From Eq.~\eqref{eq:efron_term_bound} it follows that $\mathbb{E}_{\{X'_I\}} [(Z-Z'_I)^2 \mathbb{I}_{Z>Z'_I}] \leq \frac{4}{n^2} ~\mathbb{I}_{X_I=1}$, which we can use to bound $V(\{X_I\})$. From the definition stated in Eq.~\eqref{eq:efron_V} we get:
\begin{align}
    V(\{X_I\})\leq \frac{4}{n^2} \sum_{I\subset [2n],\;|I|=4} X_I.
\end{align}
Since $X_I\sim {\rm Bin}(1, k D^{-1})$, we have
\begin{align}
    \mathbb{E}\left[\exp \left(\frac{\lambda V}{\theta}\right)\right]&\leq \mathbb{E}\left[\exp \left(\frac{4 \lambda}{\theta n^2} \sum_{I\subset [2n],\;|I|=4} X_I\right)\right] \notag\\
    &=\left(\mathbb{E}\left[\exp \left(\frac{4 \lambda}{\theta n^2} X_1\right)\right]\right)^{\binom{2n}{4}} \notag \\
    &=\left((1-k D^{-1})+k D^{-1} \exp \left(\frac{4 \lambda}{\theta n^2}\right)\right)^{\binom{2n}{4}} \notag \\
    &\leq \exp\left(k {\binom{2n}{4}} D^{-1}\left(\exp \left(\frac{4 \lambda}{\theta n^2}\right)-1\right)\right) \notag\\
    &=\exp\left(\frac{k n}{2} \left(\exp \left(\frac{4 \lambda}{\theta n^2}\right)-1\right)\right).
\end{align}
We further assume a constraint $\lambda<\frac{n^2\theta}{4}$, which implies the inequality $\exp(\frac{4\lambda}{\theta n^2})-1<\frac{8\lambda}{\theta n^2}$. This allows to further bound $\mathbb{E}\left[\exp \left(\frac{\lambda V}{\theta}\right)\right]$:
\begin{align}
\label{eq:efron_expV_bound}
    \mathbb{E}\left[\exp \left(\frac{\lambda V}{\theta}\right)\right] \leq \exp \left[ \frac{4\lambda k}{\theta n} \right].
\end{align}
We now assume an additional constraint $\lambda<\frac{1}{2\theta}$, which strengthens the condition $\lambda < \theta^{-1}$ of Eq.~\eqref{eq:exp_efron}. With this constraint, using Eq.~\eqref{eq:efron_expV_bound} in Eq.~\eqref{eq:exp_efron}, we obtain:
\begin{align}
\label{eq:exp_efron_solved}
    \log \mathbb{E}[\exp(\lambda (Z-\mathbb{E}[Z]))] &\leq \frac{4\lambda^2}{1-\lambda\theta}\frac{k}{n} \notag \\
    &\leq \frac{8\lambda^2 k}{n}.
\end{align}
This inequality is true regardless of $\theta$ and $\lambda$, insofar both numbers are positive and satisfy the constraints we introduced: 
\begin{align}
    \frac{4\lambda}{n^2}<\theta<\frac{1}{2\lambda}.
\end{align}
For a valid $\theta$ to exist, it's necessary and sufficient that $\lambda$ belongs to the interval $(0, \frac{n}{2\sqrt{2}})$. For such $\lambda$, Eq.~\eqref{eq:exp_efron_solved} holds, and combined with a Markov inequality it implies for any $t>0$:
\begin{align}
    \mathbb{P}[Z>\mathbb{E}[Z]+t] <\exp\left[\frac{8\lambda^2 k}{n}-\lambda t\right].
\end{align}
We next choose the value of $\lambda\in (0, \frac{n}{2\sqrt{2}})$ that optimizes the right hand side. If $\frac{t}{2\sqrt{2} k}<1$, this is achieved with $\lambda = \frac{t n}{16 k}$. This yields the result
\begin{align}
    \mathbb{P}[Z>\mathbb{E}[Z]+t] <\exp\left[-\frac{n t^2}{32 k}\right].
\end{align}
We choose $t=\sqrt{\frac{k^4}{2 (k'-1)}}e^{-k'}$, which automatically ensures the desired condition $\frac{t}{2\sqrt{2} k}<1$ because of the constraint $k'>e^2 k+1$ that we assumed in the Lemma statement. We obtain: 
\begin{align}
    \mathbb{P}\left[Z>\mathbb{E}[Z]+\sqrt{\frac{k^4}{2 (k'-1)}}e^{-k'} \right] <\exp\left[-\frac{e^{-2k'}k^3}{64(k'-1)} n\right].
\end{align}
Since $\mathbb{E}[Z]<\sqrt{\frac{e^2 k^4}{2\pi (k'-1)}}e^{-k'}$ (Eq.~\eqref{eq:E_Z_bound}) and $\sqrt{\frac{e^2}{2\pi}}+\sqrt{\frac{1}{2}}<2$, we arrive at an upper bound for the probability $\mathbb{P}\left[|\bar{\mathcal{I}}^{(k')}| > \frac{4k^2}{\sqrt{k'-1}}e^{-k'}n\right]$:
\begin{align}
\label{eq:tail_degree}
    &\mathbb{P}\left[Z > \frac{2k^2}{\sqrt{k'-1}}e^{-k'}\right] <\exp\left[-\frac{e^{-2k'}k^3}{64(k'-1)} n\right] \notag \\
    &\Rightarrow \mathbb{P}\left[|\bar{\mathcal{I}}^{(k')}| > \frac{4k^2}{\sqrt{k'-1}}e^{-k'}n\right] <\exp\left[-\frac{e^{-2k'}k^3}{64(k'-1)} n\right].
\end{align}

To bound $\mathbb{P}\left[\sum_{a\in[K]} |J'_a|> \frac{4k^2}{\sqrt{k'-1}}e^{-k'}n\right]$, we use the concentration properties of Gaussian random variables (see Lemma \ref{lem:abs_gauss} at the end of this Appendix). Using $K=\lfloor\frac{4k^2}{\sqrt{k'-1}}e^{-k'}n\rfloor$ in Lemma \ref{lem:abs_gauss}.1:
\begin{align}
\label{eq:tail_weight}
    \mathbb{P}\left[\sum_{a\in[K]} |J'_a|> \frac{4k^2}{\sqrt{k'-1}}e^{-k'}n\right]\leq \mathbb{P}\left[\sum_{a\in[K]} |J'_a| > K \right] < e^{-K/20}.
\end{align}
Note that our bound for $\mathbb{P}\left[|\bar{\mathcal{I}}^{(k')}| > \frac{4k^2}{\sqrt{k'-1}}e^{-k'}n\right]$ in Eq.~\eqref{eq:tail_degree} is always greater than our bound for $\mathbb{P}\left[\sum_{a\in[K]} |J'_a|> \frac{4k^2}{\sqrt{k'-1}}e^{-k'}n\right]$ in Eq.~\eqref{eq:tail_weight}. This allows us to conclude the proof of the Lemma, as Eqs.~\eqref{eq:tail_two_parts} and ~\eqref{eq:tail_degree} imply:
\begin{align}
    \mathbb{P}\left[\sum_{I\in \bar{\mathcal{I}}^{(k')}} |J_I|> \frac{4k^2}{\sqrt{k'-1}}e^{-k'}n\right]
    &\leq 2\exp\left[-\frac{e^{-2k'}k^3}{64(k'-1)} n\right].
\end{align}
\end{proof}

\begin{lemma}
\label{lem:abs_gauss}
Let $J_a\sim {\rm N}(0,1)$ for $a\in [A]$, $A\in \mathbb{N}$. Then
\begin{enumerate}
    \item $\mathbb{P}\left[\sum^A_{a=1} |J_a| \geq A\right] \leq e^{-A/20}.$
    \item $\mathbb{P}\left[\sum^A_{a=1} |J_a| \leq \frac{A}{2}\right] \leq e^{-A/8}.$
\end{enumerate}
\end{lemma}
\begin{proof}
For $J\sim {\rm N}(0,1)$, we have
\begin{align}
    \mathbb{E}\left[\exp\left(\lambda |J|\right)\right]=e^{\frac{\lambda^2}{2}}
    \left(1+\rm{Erf}\left(\frac{\lambda}{\sqrt{2}}\right)\right).
\end{align}
The Chernoff bound then implies:
\begin{align}
\mathbb{P}\left[\sum^A_{a=1} |J_a| \geq A \right] \leq \left( e^{\frac{\lambda^2}{2}-\lambda}
    \left(1+\rm{Erf}\left(\frac{\lambda}{\sqrt{2}}\right)\right)\right)^A,  \notag  \\
\mathbb{P}\left[\sum^A_{a=1} |J_a| \leq A/2 \right] \leq \left( e^{\frac{\lambda^2}{2}+\frac{\lambda}{2}}
    \left(1-\rm{Erf}\left(\frac{\lambda}{\sqrt{2}}\right)\right)\right)^A.    
\end{align}
Evaluating the two expressions at $\lambda=\frac{1}{2}$ and $\lambda=1$ respectively and using basic inequalities for the resulting constants, we obtain the two bounds claimed in the Lemma.
\end{proof}

\section{Moment bound for dense SYK-\texorpdfstring{$q$}{q}}
\label{app:momentbound}
In this Appendix, we establish the moment bound $\mathbb{E}\big( A(1)^{2n} \big) \leq \big( C\sqrt{n} \big)^{2n}$, where $A(1)$ is defined as (in Eq. \eqref{eq:A(1)}):
\begin{equation}
    A(1) = C\frac{1}{\sqrt{n}}\binom{n}{q-1}^{-3/2}\sum_{j,k,l=1}^{n}\sum_{S,S',S''}J_{S,j}J_{S',k}J_{S'',l}\: f(S,S',S'',j,k,l).
\label{eq:A(1)2}
\end{equation}
The function $f$ in this expression is defined as (in Eq. \eqref{eq:indexfunction}):
\begin{equation}
    f(S,S',S'',j,k,l)  
\vphantom{\sum_{S}}=
    \begin{cases}
      1, & \text{if}\hspace{1em} (|S''\cap S| \text{ is odd}) \\
      & \hspace{1em}\wedge\:(|S'\cap (S''\triangle S)|+\delta_{k,l} \text{ is odd})\\
      & \hspace{2em}\wedge\:(|S| = |S'| = |S''| = q-1), \\
      &\\
      0, & \text{otherwise.}
      \end{cases}
\label{eq:indexfunction2}
\end{equation}

We classify the terms in the sum in Eq. \eqref{eq:A(1)2} into five classes whose total contributions to the sum are denoted by $D_{0}$, $D_{1}$, $D_{2}$, $D_{3}$ and $D_{4}$. $D_{0}$ comprises of all terms for which the three $J$'s are distinct. We shall therefore call the call the $D_{0}$ contribution the \textit{diagonal-free} contribution. $D_{1}$ comprises of all terms for which the three $J$'s are equal. $D_{2}$, $D_{3}$ and $D_{4}$ comprise of all terms for which exactly two out of three $J$'s are equal. 

Taking $f$ into account, and thereby the terms that actually appear in $A(1)$, we conclude that the terms appearing in each class $D_{0}$, $D_{1}$, $D_{2}$, $D_{3}$ and $D_{4}$ correspond to the index sets given in Table \ref{table:contributions_to_A(1)}. An illustration of \textit{examples} of the index sets $(S,j)$, $(S',k)$ and $(S'',l)$ associated with these different classes of contributions to $A(1)$ is given in Figure \ref{fig:SetsDiag}.

\begin{figure*}
    \centering
    \includegraphics[width=1.0\linewidth]{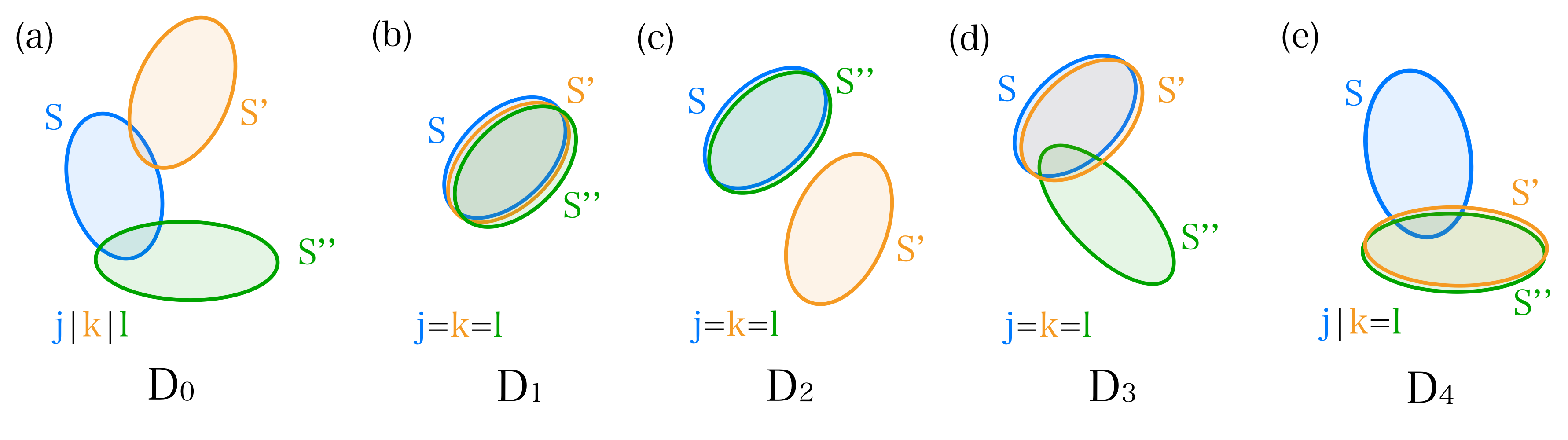}
    \caption{Illustration of \textit{examples} of index sets $(S,j)$, $(S',k)$ and $(S'',l)$ (corresponding to non-zero values of $f$ in Eq. \eqref{eq:indexfunction2}) associated with the different classes of contributions to $A(1)$ in Eq. \eqref{eq:A(1)2}. The $D_{0}$ contribution in (a) is the diagonal-free contribution (i.e., $(S,j)$, $(S',k)$ and $(S'',l)$ are unequal). The $D_{1}$, $D_{2}$, $D_{3}$ and $D_{4}$ contributions in resp. (b), (c), (d) and (e) are the diagonal contributions (i.e., at least two of $(S,j)$, $(S',k)$ and $(S'',l)$ are equal).}
    \label{fig:SetsDiag}
\end{figure*}

\begin{table}[h]
\centering
\begin{tabular}{||c c c||} 
 \hline
 class & associated index sets of terms & associated index sets of terms in $A(1)$ \\ [0.5ex] 
 \hline\hline
 $D_0$ & $(S,j)\neq (S',k) \neq (S'',l)\neq (S,j)$ & $(S,j)\neq (S',k) \neq (S'',l)\neq (S,j)$ \\
 $D_1$ & $(S,j) = (S',k) = (S'',l)$ & $(S,j) = (S',k) = (S'',l)$ \\ 
 $D_2$ & $(S,j)=(S'',l)\neq (S',k)$ & $S = S'' \neq S'$, $j=l=k$ \\ 
 $D_3$ & $(S,j)=(S',k)\neq (S'',l)$ & $S=S'\neq S''$, $j=l=k$ \\ 
 $D_4$ & $(S',k) = (S'',l) \neq (S,j)$ & $(S',k) = (S'',l) \neq (S,j)$ \\ [1ex]
 \hline
\end{tabular}
\caption{The index sets associated with each class of terms, and the index sets associated with each class of terms that appear in the expression for $A(1)$ (i.e., taking $f$ into account).}
\label{table:contributions_to_A(1)}
\end{table}

To upper bound the $(2n)$th moment of $A(1)_{\min}$, we upper bound the $r$th moments (for even $r\leq 16\cdot 2n$) of $D_{0}$, $D_{1}$, $D_{2}$, $D_{3}$, $D_{4}$ separately. In particular, if $\mathbb{E}\big( (D_{i})^{r} \big) \leq \big( C\sqrt{n} \big)^{r}$ for $i=0,1,...,4$ and all even $r\leq 16\cdot 2n$, then $\mathbb{E}\big( A(1)^{2n} \big)\leq \big( C\sqrt{n} \big)^{2n}$. Note that through the multinomial expansion and successive application of Cauchy-Schwarz inequality these former bounds indeed give an upper bound on the $(2n)$th moment of $A(1)$:
\begin{center}
\begin{align}
    \mathbb{E}\big( A(1)^{2n} \big) =&\: \mathbb{E}\Big( (D_{0}+D_{1}+D_{2}+D_{3}+D_{4})^{2n} \Big) = \sum_{k_{0}+...+k_{4}=2n}\frac{2n!}{k_{0}!...k_{4}!}\mathbb{E}\Big( D_{0}^{k_{0}}D_{1}^{k_{1}}...D_{4}^{k_{4}} \Big) \nonumber \\
    \leq&\: \sum_{k_{0}+...+k_{4}=2n}C^{2n}\:\mathbb{E}\big( D_{0}^{2k_{0}} \big)^{1/2} \mathbb{E}\big( D_{1}^{4k_{1}} \big)^{1/4} \mathbb{E}\big( D_{2}^{8k_{2}} \big)^{1/8} \mathbb{E}\big( D_{3}^{16k_{3}} \big)^{1/16} \mathbb{E}\big( D_{4}^{16k_{4}} \big)^{1/16} \nonumber \\
    \leq&\: \sum_{k_{0}+...+k_{4}=2n}C^{2n}\: \big( C\sqrt{n} \big)^{k_{0}+...+k_{4}} 
    \leq \big( C\sqrt{n} \big)^{2n},
\label{eq:momentboundsplit}
\end{align}
\end{center}
where we have used that the multinomial coefficient can be upper bounded by $C^{2n}$ and that the number of $5$-tuples of non-negative integers whose sum equals $2n$ is upper bounded by $Cn^{4}$ (which is smaller than $C^{2n}$ for some constant $C$). Although clearly the $r$th moments of e.g. $D_{0}$ have to only be bounded for even $r\leq 2\cdot 2n$, we bound -- for the sake of clarity -- the $r$th moments for even $r\leq 16\cdot 2n$ for \textit{all} $D_{i}$'s. We first deal with the case of $D_{0}$, since the fact that this contribution is diagonal-free allows one to employ a decoupling technique. Afterwards, we will consider the $D_{1}$, $D_{2}$, $D_{3}$ and $D_{4}$ contributions. 

First, we state the following lemma, which will be useful throughout this appendix.
\begin{lemma}
Let $P$ and $P'$ be two polynomials of centered Gaussian random variables (i.e., the monomials are formed by products of elements from a sequence of independent centered Gaussian random variables, and each variable is allowed to appear in a monomial multiple times) with non-negative coefficients. Then, for any even $r$, $\mathbb{E}\big( |P+P'|^r \big)\geq \mathbb{E}\big( |P|^r \big)$.
\label{lemma:momentsofgaussianpolynomials}
\end{lemma}
\begin{proof}
We have $\mathbb{E}\big( |P+P'|^r \big) = \mathbb{E}\big( |P|^r \big) + \sum_{k=1}^{r}\binom{r}{k}\mathbb{E}\big( P^{r-k}(P')^{k} \big)$, and $\mathbb{E}\big( P^{r-k}(P')^{k} \big)$ is non-negative (for any integers $r,k$) since $P$ and $P'$ have non-negative coefficients and all moments of centered Gaussian random variables are non-negative.
\end{proof}

\subsection{Upper bound for moments of \texorpdfstring{$D_{0}$}{D0} (diagonal-free contribution)}
We start by noting that the function $f$ takes on values $0$ or $1$, dependent on the index sets $S,S',S'',j,k,l$ labeling the Majorana operators. We consider replacing $f$ in each term of $D_0$ (Eq.~\eqref{eq:A(1)2}) with $\delta_{a,b}\delta_{c,d}$, where either 
\begin{align}
 a\in(S'\cup k),\: b \neq c\in (S''\cup l),\;  d\in (S\cup j) \tag{option 1},
 \end{align}
 or 
 \begin{align}
a\in (S'\cup k),\: b \neq c\in (S\cup j),\; d\in  (S''\cup l) \tag{ option 2 }.
\end{align}
We denote this modified sum as $D_{0, \delta\delta}$.
By inspection, the index sets for which $f$ is non-zero all correspond to a non-zero contribution for $\delta_{a,b}\delta_{c,d}$. Note that those index sets for which $\delta_{a,b}\delta_{c,d}$ is non-zero also include index sets for which $f$ is zero.  Hence, the terms associated with non-zero $\delta_{a,b}\delta_{c,d}$ (for the two options listed above) are a superset of the terms that correspond to non-zero values of $f$. Therefore, by Lemma \ref{lemma:momentsofgaussianpolynomials}, the upper bounds on even moments of $D_{0}$ can be obtained by upper bounding the even moments of $D_{0, \delta\delta}$.

We will denote the part of the sum $D_{0, \delta\delta}$ corresponding to option $1$ as $D_{0,\min}$:
\begin{equation}
    D_{0,\min} := C\frac{1}{\sqrt{n}}\binom{n}{q-1}^{-3/2}\hspace{-3.5em}\sum_{\substack{j,k,l,\\S,S',S'',\\ \\ \text{s.t. }|(S'\cup k)\cap (S''\cup l)|\geq 1 \\ \text{and } |(S''\cup l)\cap (S\cup j)|\geq 1}}\hspace{-3em}J_{S,j}J_{S',k}J_{S'',l},\!
\label{eq:Amin}
\end{equation}
where the sum is over indices such that $(S,j)\neq (S',k)\neq (S'',l)\neq (S,j)$ (by definition of $D_{0}$) and such that $(S'\cup k)\cap (S''\cup l)$ and $(S''\cup l)\cap (S\cup j)$ differ by at least one element. Any bound for all even moments of $D_{0,\min}$ also holds for $D_{0, \delta\delta}-D_{0,\min}$ which corresponds to option 2, due to the symmetry $(S,j)\leftrightarrow (S'',l)$ between the two options.
An upper bound on all even moments of $D_{0,\delta\delta}$ (and, by implication, $D_{0}$) then follows from binomial expansion and application of the Cauchy-Schwarz inequality, similarly to Eq.~\eqref{eq:momentboundsplit}. Thus it only remains to prove $\mathbb{E}\big( |D_{0,\min}|^{r} \big)<(C\sqrt{n})^r$ for all even $r$.

To upper bound the even moments of $D_{0,\min}$, we are going to employ a decoupling technique. To that end, we will study the even moments of a related decoupled quantity. This decoupled quantity is defined as $D_{0,\min}$ but with the standard Gaussian random variables $J_{S,j}$, $J_{S',k}$ and $J_{S'',l}$ (selected from a single sequence of standard Gaussian random variables) being replaced by their \textit{decoupled} versions $J_{S,j}^{(1)}$, $J_{S',k}^{(2)}$ and $J_{S'',l}^{(3)}$ (selected from \textit{three} independent sequences of standard Gaussian random variables). The related decoupled quantity is given by (where the sum is again over indices $(S,j)\neq (S',k)\neq (S'',l)\neq (S,j)$ and again such that $(S'\cup k)\cap (S''\cup l)$ and $(S''\cup l)\cap (S\cup j)$ differ by at least one element):
\begin{equation}
    C\:\frac{1}{n^{3q/2-1}}\hspace{-2.5em}\sum_{\substack{j,k,l,\\S,S',S'',\\ \\ \text{s.t. }|(S'\cup k)\cap (S''\cup l)|\geq 1 \\ \text{and } |(S''\cup l)\cap (S\cup j)|\geq 1 }}\hspace{-2.8em}J_{S,j}^{(1)}J_{S',k}^{(2)}J_{S'',l}^{(3)},
\label{eq:amindecoupled}
\end{equation}
where we have additionally used that $(k/l)^{l}\leq\binom{k}{l}\leq (e\:k/l)^{l}$.

To upper bound the even moments of this decoupled quantity, we will make use of Lemma \ref{lemma:latala} below from \cite{GaussianChaoses}. The even moments of this decoupled sum are upper bounded by upper bounding the even moments of a decoupled sum whose terms are a superset of the terms in the sum in Eq. \ref{eq:amindecoupled}. Through Lemma \ref{lemma:momentsofgaussianpolynomials}, the even moments of the latter sum are larger than those of the former sum. For each $J^{(i)}_{S,j}$, we introduce $q!-1$ additional independent standard Gaussian random variables associated with the $q!$ permutations of the indices in the subsets of size $q$. Furthermore, we introduce additional independent standard Gaussian random variables for which some or all of the $q$ indices that label them are equal. We consider a sum over lists of $3q$ indices (which label the independent standard Gaussian random variables) $i_{1}^{(1)},\ldots,i_{q}^{(1)}$, $i_{1}^{(2)},\ldots,i_{q}^{(2)}$ and $i_{1}^{(3)},\ldots,i_{q}^{(3)}$ (with each index in $[2n]$), instead of the sum over subsets of $[2n]$ in Eq. \eqref{eq:amindecoupled}. Note that the sum over lists, by definition, can contain terms for which two (or three) of the Gaussian random variables have equal index sets. 

The index lists $i_{1}^{(1)},\ldots,i_{q}^{(1)}$ and $i_{1}^{(2)},\ldots,i_{q}^{(2)}$ that are summed over each have \textit{any} one index (denoted by resp. $x$ and $y$) that is equal to an index in $i_{1}^{(3)},\ldots,i_{q}^{(3)}$. If we additionally sum over all `positions' of the $x$ and $y$ indices (where $p_1$, $p_2$, $p_3$ and $p_4$ $\in\{1,\ldots,q\}$ label these positions), we obtain the sum (see Eq. \eqref{eq:amindecoupled3} below) whose terms are a superset of those in the sum in Eq. \eqref{eq:amindecoupled}. Note that this sum in Eq. \eqref{eq:amindecoupled3} contains all the contributions from a sum over lists of indices, and contains some terms multiple times that would occur only once in a sum over lists of indices: For example, in the hypothetical case $q=2$, one could have a contribution $J^{(1)}_{3,3}J^{(2)}_{6,7}J^{(3)}_{3,7}$ that would appear once in the sum over lists of indices but appears twice in the sum in Eq. \eqref{eq:amindecoupled3} (once for $p_1 = 1$ and once for $p_1 = 2$). Through Lemma \ref{lemma:momentsofgaussianpolynomials}, the even moments of the sum in Eq. \eqref{eq:amindecoupled3} will therefore be larger than those of the sum over lists of indices (and therefore larger than those of the sum in Eq. \eqref{eq:amindecoupled}), and it will thus suffice to upper bound the even moments of the sum in Eq. \eqref{eq:amindecoupled3}.
\begin{multline}
    D_{0,\min}^{\text{decoupled}} := C\:\frac{1}{n^{3q/2-1}}\sum_{\substack{p_1,\ldots,p_4=1\\\text{s.t. }p_3\neq p_4}}^{q}\Bigg[\\ \times
    \sum_{\substack{i^{(1)}_{1},..,i^{(1)}_{p_1-1},x,\\i^{(1)}_{p_1+1},..,i^{(1)}_{q}=1}}^{2n} \hspace{+1.0em}
    \sum_{\substack{i^{(2)}_{1},..,i^{(2)}_{p_2-1},y,\\i^{(2)}_{p_2+1},..,i^{(2)}_{q}=1}}^{2n} \hspace{+1.0em}
    \sum_{\substack{i^{(3)}_{1},..,i^{(3)}_{p_3-1},x,i^{(3)}_{p_3+1},..,\\i^{(3)}_{p_4-1},y,i^{(3)}_{p_4+1},..,i^{(3)}_{q}=1}}^{2n} \hspace{+0.5em}
    \bigg(J^{(1)}_{\substack{i^{(1)}_{1},..,i^{(1)}_{p_1-1},x,\\i^{(1)}_{p_1+1},..,i^{(1)}_{q}}} J^{(2)}_{\substack{i^{(2)}_{1},..,i^{(2)}_{p_2-1},y,\\i^{(2)}_{p_2+1},..,i^{(2)}_{q}}} J^{(3)}_{\substack{i^{(3)}_{1},..,i^{(3)}_{p_3-1},x,i^{(3)}_{p_3+1},..,\\i^{(3)}_{p_4-1},y,i^{(3)}_{p_4+1},...,i^{(3)}_{q}}}\bigg)
    \Bigg].
\label{eq:amindecoupled3}
\end{multline}

The free indices ($q-1$ indices of $i_{1}^{(1)},\ldots,i_{q}^{(1)}$ and $i_{1}^{(2)},\ldots,i_{q}^{(2)}$, and $q-2$ indices of $i_{1}^{(3)},\ldots,i_{q}^{(3)}$) can be summed over to obtain new independent standard Gaussian random variables denoted by $K^{(1)}_{x,p_1}$, $K^{(2)}_{y,p_2}$ and $K^{(3)}_{x,p_3;y,p_4}$:
\begin{subequations}
\begin{equation}
    K^{(1)}_{x,p_1} := 1/\sqrt{(2n)^{q-1}}\hspace{-2.5em}\sum_{i^{(1)}_{1},..,i^{(1)}_{p_1-1},i^{(1)}_{p_1+1},..,i^{(1)}_{q}=1}^{2n}\hspace{-1.5em}J^{(1)}_{i^{(1)}_{1},..,i^{(1)}_{p_1-1},x,i^{(1)}_{p_1+1},..,i^{(1)}_{q}},
\end{equation}
\begin{equation}
    K^{(2)}_{y,p_2} := 1/\sqrt{(2n)^{q-1}}\hspace{-2.5em}\sum_{i^{(2)}_{1},..,i^{(2)}_{p_2-1},i^{(2)}_{p_2+1},..,i^{(2)}_{q}=1}^{2n}\hspace{-1.5em}J^{(2)}_{i^{(2)}_{1},..,i^{(2)}_{p_2-1},y,i^{(2)}_{p_2+1},..,i^{(2)}_{q}},
\end{equation}
\begin{equation}
    K^{(3)}_{x,p_3;y,p_4} := 1/\sqrt{(2n)^{q-2}}\hspace{-2.5em}\sum_{\substack{i^{(3)}_{1},..,i^{(3)}_{p_3-1},i^{(3)}_{p_3+1},..,\\i^{(3)}_{p_4-1},i^{(3)}_{p_4+1},..,i^{(3)}_{q}=1}}^{2n}\hspace{-1.5em}J^{(3)}_{i^{(3)}_{1},..,i^{(3)}_{p_3-1},x,i^{(3)}_{p_3+1},..,i^{(3)}_{p_4-1},y,i^{(3)}_{p_4+1},...,i^{(3)}_{q}},
\end{equation}
\end{subequations}
where we have used that the normalized sum $1/\sqrt{m}\sum_{i=1}^{m}J_{i}$ of a sequence of standard Gaussian random variables $J_{1},\ldots,J_{m}$ is again a standard Gaussian random variable. We now obtain the following expression for $D_{0,\min}^{\text{decoupled}}$:
\begin{equation}
    D_{0,\min}^{\text{decoupled}} := C\sum_{\substack{p_1,\ldots,p_4=1\\\text{s.t. }p_3\neq p_4}}^{q}\bigg[\frac{1}{n}\sum_{x,y=1}^{2n}K_{x,p_1}^{(1)}K_{y,p_2}^{(2)}K_{x,p_3;y,p_4}^{(3)}\bigg].
\label{eq:amindecoupled2}
\end{equation}
The sum over all free indices gives an extra total factor of $n^{3q/2-2}$, which partially cancels against $n^{3q/2-1}$ in Eq. \eqref{eq:amindecoupled3}. Importantly, we note that now the random variables $K^{(1)}_{x,p_1}$ and $K^{(1)}_{x',p_1}$ are independent for $x\neq x'$ (and equivalently for $K_{y,p_2}^{(2)}$ and $K_{x,p_3;y,p_4}^{(3)}$). We will apply Lemma \ref{lemma:latala} from \cite{GaussianChaoses} separately to each contribution to $D_{0,\min}^{\text{decoupled}}$ in Eq. \eqref{eq:amindecoupled2} (with a contribution corresponding to one combination of $p_i$'s).

\begin{lemma}[Theorem 1 in \cite{GaussianChaoses}]
Let $Y\in \mathbb{R}^{N\times \ldots \times N}$ be a $d$-dimensional matrix and define:
\begin{equation}
    F\Big(\{K_{1}^{(j)},\ldots,K_{N}^{(j)}\}_{j=1}^{d}\Big) := \sum_{i_{1},\ldots,i_{d}=1}^{N}Y_{i_{1},\ldots,i_{d}}\prod_{j=1}^{d}K_{i_{j}}^{(j)},
\label{eq:Fgaussianchaos}
\end{equation}
where $\{K_{1}^{(j)},\ldots,K_{N}^{(j)}\}_{j=1}^{d}$ are $d$ independent sequences of $N$ standard Gaussian random variables. Then for any integer $k\geq 2$:
\begin{align}
    \mathbb{E}\Big( \bigl\lvert F\big(\{K_{1}^{(j)},\ldots,K_{N}^{(j)}\}_{j=1}^{d}\big)\bigr\rvert^{k}  \Big)\leq& \: \big( C\: \sum_{\mathcal{P}}k^{|\mathcal{P}|/2}\|Y\|_{\mathcal{P}} \big)^{k} \nonumber \\
    \leq&\: \big( C\: \max_{\mathcal{P}}\big[ k^{|\mathcal{P}|/2}\|Y\|_{\mathcal{P}} \big] \big)^{k},\notag
\end{align}
where $\mathcal{P}$ are partitions of $[d]$ into non-empty parts $(P_{1},\ldots,P_{s})$. The second inequality holds because the number of partitions of $[d]$ into non-empty parts is constant in $n$ (since $d$ is constant in $n$). The quantity $\|Y\|_{\mathcal{P}}$ is defined as:
\begin{align}
    \|Y\|_{\mathcal{P}} &= \|Y\|_{(P_{1},\ldots,P_{s})}:= \max\Bigg\{ \sum_{ i_{1},\ldots,i_{d} = 1}^{N} Y_{i_{1},\ldots,i_{d}}\:x_{i_{P_1}}^{(1)}\ldots x_{i_{P_s}}^{(s)} \: 
    :\! \sum_{i_{P_1}}\big(x_{i_{P_1}}^{(1)}\big)^{2}\leq 1,\ldots,\sum_{i_{P_k}}\big(x_{i_{P_k}}^{(k)}\big)^{2}\leq 1 \Bigg\},
\label{eq:ynorm}
\end{align}
with each $x\in \mathbb{R}$.
\label{lemma:latala}
\end{lemma}

\begin{remark}
If $F$ in Eq. \eqref{eq:Fgaussianchaos} is diagonal-free (i.e., $Y_{i_1,\ldots,i_d} = 0$ if $i_j=i_k$ for any $j\neq k$) then the moments of the `\textit{decoupled}' $F$ in Eq. \eqref{eq:Fgaussianchaos} are (up to constants only depending on $d$) an upper bound for the moments of its `\textit{coupled}' counterpart $F'\big(K_{1},\ldots,K_{N}\big) := \sum_{i_{1},\ldots,i_{d}=1}^{N}Y_{i_{1},\ldots,i_{d}}\prod_{j=1}^{d}K_{i_{j}}$ (i.e., where the random variables are all taken from the \textit{same} sequence of $N$ standard Gaussian random variables):
\begin{equation*}
    \mathbb{E}\big( (F')^{r} \big) \leq C \mathbb{E}\big( (F)^{r} \big).
\end{equation*}
See e.g. Theorem 2.1 in \cite{DecouplingInequalities}.
\end{remark}
The fact that this decoupling inequality only holds for diagonal-free polynomials is exactly the reason for differentiating between the diagonal-free contribution $D_{0}$ and the diagonal contributions $D_{1},D_{2},D_{3},D_{4}$ to $A(1)$.

For $\sum_{x,y=1}^{2n}K_{x,p_1}^{(1)}K_{y,p_2}^{(2)}K_{x,p_3;y,p_4}^{(3)}$ in Eq. \eqref{eq:amindecoupled2}, we see that $d=3$ and hence the possible partitions $\mathcal{P}$ are $\{1,2,3\}$, $\{1\}\{2,3\}$, $\{2\}\{1,3\}$, $\{1,2\}\{3\}$, $\{1\}\{2\}\{3\}$. The associated $\|Y\|_{\mathcal{P}}$ values can be (straightforwardly) calculated and are given in Table \ref{table:partitions}. Using Table \ref{table:partitions} and Lemma \ref{lemma:latala}, we find the following upper bound on $\mathbb{E}\Big(\Big[C/n\sum_{x,y=1}^{2n}K_{x,p_1}^{(1)}K_{y,p_2}^{(2)}K_{x,p_3;y,p_4}^{(3)}\Big]^{r}\Big)$ (for all even $r$):
\begin{equation}
    \Big( C/n\:\max\big(\sqrt{r}n,r\sqrt{n},r\sqrt{n},r,r^{3/2}\big) \Big)^{r} \leq \big(C \sqrt{n}\big)^{r}.
\end{equation}

\begin{table}[ht]
\centering
\begin{tabular}{||c c c||} 
 \hline
 $\mathcal{P}$ & $|\mathcal{P}|$ & $\|Y\|_{\mathcal{P}}$ \\ [0.5ex] 
 \hline\hline
 $\{1,2,3\}$ & $1$ & $2n$ \\
 $\{1\}\{2,3\}$ & $2$ & $\sqrt{2n}$ \\ 
 $\{2\}\{1,3\}$ & $2$ & $\sqrt{2n}$ \\ 
 $\{1,2\}\{3\}$ & $2$ & $1$ \\ 
 $\{1\}\{2\}\{3\}$ & $3$ & $1$ \\ [1ex] 
 \hline
\end{tabular}
\caption{The different partitions $\mathcal{P}$ of $[3]$ into non-empty parts, with the associated number of parts $|\mathcal{P}|$, and the associated $\|Y\|_{\mathcal{P}}$ for $\sum_{x,y=1}^{2n}K_{x,p_1}^{(1)}K_{y,p_2}^{(2)}K_{x,p_3;y,p_4}^{(3)}$ in Eq. \eqref{eq:amindecoupled2}. $\|Y\|_{\mathcal{P}}$ for the first four partitions can be straightforwardly evaluated by applying Eq. \eqref{eq:ynorm} to Eq. \eqref{eq:amindecoupled}, and the fifth $\|Y\|_{\mathcal{P}}$ can be evaluated by additional application of the Cauchy-Schwarz inequality.}
\label{table:partitions}
\end{table}

Note that $D_{0,\min}^{\text{decoupled}}$ in Eq. \ref{eq:amindecoupled2} consists of $q^4$ (with $q=O(1)$) contributions, each corresponding to a given combination of $p_i$'s. We can again use the multinomial expansion and successive application of the Cauchy-Schwarz inequality (together with the fact that the multinomial coefficients can be upper bounded by $C^{r}$ and that the number of $q^{4}$-tuples of non-negative integers whose sum equals $r$ is upper bounded by $C^{r}$ for some constant $C$) to conclude that the upper bounds of $(C\sqrt{n})^{r}$ for $r$th moments (for all even $r$) of these contributions imply an upper bound of $(C\sqrt{n})^{r}$ for $r$th moments (for all even $r$) of $D_{0,\min}^{\text{decoupled}}$.

We now employ the decoupling inequality from the above remark to obtain
\begin{equation*}
    \mathbb{E}\big( |D_{0,\min}|^{r} \big) \!\leq\! C\mathbb{E}\big( \bigl\lvert D_{0,\min}^{\text{decoupled}}\bigr\rvert^{r} \big) \!\leq\! \big(C \sqrt{n}\big)^{r}.
\end{equation*}
From the arguments given previously, this implies the desired bound $\mathbb{E}\big( |D_{0}|^{r} \big)\leq \big(C \sqrt{n}\big)^{r}$ for all even $r$, in particular for $r\leq 16 \cdot 2n$.

\subsection{Upper bound for moments of \texorpdfstring{$D_{1}$}{D1}, \texorpdfstring{$D_{2}$}{D2}, \texorpdfstring{$D_{3}$}{D3} and \texorpdfstring{$D_{4}$}{D4}}
In the previous section we used a decoupling inequality to upper bound the $r$th moments (for even $r\leq 16\cdot 2n$) of $D_{0}$. These decoupling inequalities hold for (Gaussian) polynomials for which each Gaussian monomial is a product of \textit{distinct} Gaussian random variables, i.e., \textit{diagonal-free} polynomials. This holds indeed -- by definition -- for the $D_{0}$ contribution to $A(1)$, but not for contributions $D_{1}$, $D_{2}$, $D_{3}$ and $D_{4}$. For that reason, we cannot make use of the same decoupling inequality for the $D_{1}$, $D_{2}$, $D_{3}$ and $D_{4}$ contributions. Therefore, we have to resort to other methods to bound their $r$th moments (for even $r\leq 16\cdot 2n$).

\begin{itemize}
\item 
The $D_{1}$ contribution can be written as:
\begin{equation}
    D_{1} := C\frac{1}{\sqrt{n}}\binom{n}{q-1}^{-3/2}\sum_{j,S}\big(J_{S,j}\big)^{3}.
\end{equation}
The $r$th moment (with $r$ even) of $D_{1}$ can be upper bounded as follows:
\begin{center}
\begin{align}
    \mathbb{E}\Big( |D_{1}|^{r} \Big) \leq&\: C^{r}\Bigg[ \frac{1}{\sqrt{n}}\binom{n}{q-1}^{-3/2}\Bigg]^{r} \mathbb{E}\bigg[\Big( \sum_{S,j} \big(J_{S,j}\big)^{3} \Big)^{r}\bigg] \nonumber \\ \leq&\: C^{r}\Bigg[ \frac{1}{\sqrt{n}}\binom{n}{q-1}^{-3/2}\Bigg]^{r} \bigg(\binom{n}{q-1}n\:r^{3/2}\bigg)^{r} \nonumber \\
    \leq&\: C^{r} n^{(5/2-q/2)r},
\end{align}
\end{center}
where we have used that $\mathbb{E}\Big[ \Big(\sum_{i=1}^{m}K_{i}\Big)^{r} \Big] =\sum_{k_1 + \ldots + k_m = r} \frac{r!}{k_1!\ldots k_m!} \mathbb{E}\big( \big(K_1)^{k_1}\big)\ldots \mathbb{E}\big((K_m)^{k_m} \big)$ (for $K_1,\ldots,K_m$ independent random variables), the fact that $(S,j)$ can take on $2n\binom{2n}{q-1}$ values and the fact that the $p$th moment of a standard Gaussian random variable is equal to $(p-1)!!$ ($\leq p^{p/2}$). For even $r$, we therefore conclude that $\mathbb{E}\Big( |D_{1}|^{r} \Big) \leq \big( C\sqrt{n} \big)^{r}$. 

\item
The $D_{2}$ and $D_{3}$ contributions are equivalent and can be written as:
\begin{equation}
    D_{2} = D_{3} := C\frac{1}{\sqrt{n}}\binom{n}{q-1}^{-3/2}\sum_{j,S,S'\text{ s.t. }S\neq S'}\big(J_{S,j}\big)^{2}J_{S',j}.
\end{equation}
The $r$th moment (with $r$ even) can be written as follows:
\begin{equation}
    \mathbb{E}\Big( |D_{2}|^{r} \Big), \mathbb{E}\Big( |D_{3}|^{r} \Big) = C^{r}\Bigg[ \frac{1}{\sqrt{n}}\binom{n}{q-1}^{-3/2}\Bigg]^{r} \mathbb{E}\Big(\Big( \sum_{j,S,S'\text{ s.t. }S\neq S'} \big(J_{S,j}\big)^{2}J_{S',j} \Big)^{r} \Big).
\label{eq:D2D3bound}
\end{equation}
We define 
\begin{equation}
g:=\sum_{j,S,S'\text{ s.t. }S\neq S'} \big(J_{S,j}\big)^{2}J_{S',j},
\label{eq:g}
\end{equation}
for which $\mathbb{E}(g) = 0$. We note that $g$ is a homogeneous polynomial in standard Gaussian random variables of degree $3$. To upper bound the moments of $g$, and thereby the moments of $D_{2}$ and $D_{3}$, we use the following result from \cite{upperbound_moments_diagonal_gaussianpolynomials}. This result is an extension of Lemma \ref{lemma:latala} from \cite{GaussianChaoses} to the setting where diagonal terms are allowed to appear in the polynomial. The extension also includes inhomogeneous polynomials, although in the current setting we are considering only homogeneous polynomials.

\begin{lemma}[Theorem 1.3 in \cite{upperbound_moments_diagonal_gaussianpolynomials}]
Let $\mathbf{K}:=K_{1}\ldots,K_{N}$ denote a sequence of $N$ independent standard Gaussian random variables and $g:\mathbb{R}^{N}\to \mathbb{R}$ a polynomial of degree $D$. Then, for all $r\geq 2$:
\begin{equation}
    \mathbb{E}\bigg(\Big[g(\mathbf{K}) - \mathbb{E}\big(g(\mathbf{K})\big)\Big]^{r}\bigg) \leq C^{r}\bigg[\sum_{1\leq d\leq 3}\sum_{\mathcal{P}([d])}r^{|\mathcal{P}|/2}\bigl\|\mathbb{E}\big(\mathbf{D}^{d}g(\mathbf{K})\big)\bigr\|_{\mathcal{P}}\bigg]^{r},
\label{eq:diagupperbound}
\end{equation}
where $\mathcal{P}$ are partitions of $[d]$ into non-empty parts, and $\|Y\|_{\mathcal{P}}$ (with $Y$ a $d$-way tensor) is defined in Eq. \eqref{eq:ynorm}. $\mathbf{D}^{d}g(\mathbf{K})$ denotes the $d$th derivative of $g(\mathbf{K})$, which corresponds to a $d$-way tensor with entries equal to $\big[\mathbf{D}^{d}g(\mathbf{K})\big]_{i_{1},\ldots,i_{d}} = \frac{\partial}{\partial K_{i_1}}\ldots \frac{\partial}{\partial K_{i_d}}g(\mathbf{K})$. For $d = D$, $\mathbf{D}^{d}g(\mathbf{K})$ is constant.
\label{lem:latalaextended}
\end{lemma}
For $g$ in Eq. \eqref{eq:g}, we have that $N=n\binom{n}{q-1}$, since the sequence of Gaussian random variables corresponds to $\{J_{S,j}\}$. To find an upper bound for the $r$th moment of $g$ using Eq. \eqref{eq:diagupperbound}, we first calculate $\mathbf{D}^{d}g$ for $d=1,2,3$. Then, for each $d$, we upper bound $\bigl\|\mathbb{E}\big(\mathbf{D}^{d}g\big)\bigr\|_{\mathcal{P}}$ for all partitions $\mathcal{P}$ of $[d]$. We will show that for all $d$ and associated partitions $\mathcal{P}([d])$, $\bigl\|\mathbb{E}\big(\mathbf{D}^{d}g\big)\bigr\|_{\mathcal{P}}$ can be upper bounded in such a way that $\mathbb{E}\Big( |D_{2}|^{r} \Big), \mathbb{E}\Big( |D_{3}|^{r} \Big) \leq \big( C\sqrt{n} \big)^{r}$ for all even $2\leq r\leq 16\cdot 2n$. Finally, the $0$th moment also (trivially) satisfies this upper bound, hence it holds for all even $r\leq 16\cdot 2n$.

The derivatives of $g$ are equal to:
\begin{equation}
    D\,g = \big( \sum_{S'\,:\, S'\neq S}J_{S',j}^{2}+2J_{S,j}\sum_{S'\,:\, S'\neq S}J_{S',j} \big)_{(S,j)} \Longrightarrow \mathbb{E}\big(D\,g\big) = \binom{n}{q-1}_{(S,j)},
\end{equation}
\begin{equation}
    D^{2}\,g =    \begin{pmatrix}
                2\sum_{S'\,:\, S'\neq S}J_{S',j}, & \text{if }(S,j)=(T,k) \\
                2(J_{T,j}+J_{S,j}), & \text{if } S\neq T\text{ and }j=k \\
                0, & \text{if }j\neq k
                \end{pmatrix}_{(S,j),(T,k)} \Longrightarrow \mathbb{E}\big(D^{2}g\big) = (0)_{(S,j),(T,k)},
\end{equation}
\begin{equation}
    D^{3}\,g =    \begin{pmatrix}
               2, & \text{if }(S=T\neq U \text{ or } S\neq T=U\text{ or }S=U\neq T)\text{ and }j=k=l \\
               0, & \text{if }(S=T=U\text{ or }S\neq T\neq U)\text{ and }j=k=l \\
               0, & \text{if }j,k,l\text{ are not all equal} \\ \end{pmatrix}_{(S,j),(T,k),(U,l)} \Longrightarrow \mathbb{E}\big(D^{3}\,g\big) = D^{3}\,g.
\end{equation}
In Table \ref{table:partitionsD2D3}, we give the values of $\bigl\|\mathbb{E}\big(\mathbf{D}^{d}g\big)\bigr\|_{\mathcal{P}}$ for all partitions $\mathcal{P}([d])$ for $d=1,2,3$. $\bigl\|\mathbb{E}\big(\mathbf{D}^{d}g\big)\bigr\|_{\mathcal{P}}$ for $d=1$ can be straightforwardly evaluated using Eq. \eqref{eq:ynorm} and for $d=2$ can be trivially evaluated by using $\mathbb{E}\big(\mathbf{D}^{2}g\big)=0$. For $d=3$, $\bigl\|\mathbb{E}\big(\mathbf{D}^{d}g\big)\bigr\|_{\mathcal{P}}$ can be upper bounded using Eq. \eqref{eq:ynorm}, and the triangle and Cauchy-Schwarz inequalities (for illustration purposes, we provide an example of the derivation of this upper bound for $\mathcal{P}=\{1,2\}\{3\}$ below).

\begin{table}[ht]
\centering
\captionsetup{width=.9\linewidth}
\begin{tabular}{||c c c||} 
 \hline
 $\mathcal{P}$ & $|\mathcal{P}|$ & $\bigl\|\mathbb{E}\big(\mathbf{D}^{d}g\big)\bigr\|_{\mathcal{P}}$ \\ [0.5ex] 
 \hline\hline
 $\{1\}$ & $1$ & $\binom{n}{q-1}^{3/2}$ \\ [0.5ex]
 \hline
 $\{1,2\}$ & $1$ & 0 \\
 $\{1\}\{2\}$ & $2$ & 0 \\ [0.5ex]
 \hline
 $\{1,2,3\}$ & $1$ & $C\sqrt{n}\binom{n}{q-1}$ \\
 $\{1\}\{2,3\}$ & $2$ & $Cn\binom{n}{q-1}$ \\ 
 $\{2\}\{1,3\}$ & $2$ & $Cn\binom{n}{q-1}$ \\ 
 $\{1,2\}\{3\}$ & $2$ & $Cn\binom{n}{q-1}$ \\ 
 $\{1\}\{2\}\{3\}$ & $3$ & $C\binom{n}{q-1}$ \\ [1ex] 
 \hline
\end{tabular}
\caption{The different partitions $\mathcal{P}$ of $[3]$ into non-empty parts, with the associated number of parts $|\mathcal{P}|$, and (the upper bounds for) the associated $\bigl\|\mathbb{E}\big(\mathbf{D}^{d}g\big)\bigr\|_{\mathcal{P}}$ for $g$ in Eq. \eqref{eq:g}.}
\label{table:partitionsD2D3}
\end{table}

Combining the upper bounds for $\bigl\|\mathbb{E}\big(\mathbf{D}^{d}g\big)\bigr\|_{\mathcal{P}}$ in Table \ref{table:partitionsD2D3} with the factor of $r^{|\mathcal{P}|/2}$($\leq Cn^{|\mathcal{P}|/2}$) in Eq. \eqref{eq:diagupperbound} and the normalization factor in Eq. \eqref{eq:D2D3bound}, we find -- using $\mathbb{E}(g) = 0$ -- that indeed $\mathbb{E}\Big( |D_{2}|^{r} \Big), \mathbb{E}\Big( |D_{3}|^{r} \Big) \leq \big( C\sqrt{n} \big)^{r}$ for all even $r\leq 16\cdot 2n$.

\textit{Example:} For illustration purposes, we give an explicit evaluation of $\bigl\|\mathbb{E}\big(\mathbf{D}^{d}g\big)\bigr\|_{\mathcal{P}}$ for $\mathcal{P}=\{1,2\}\{3\}$ (the evaluations for other $\mathcal{P}$'s follow using similar methods). By definition (Eq. \eqref{eq:ynorm}), we have:
\begin{multline}
    \bigl\|\mathbb{E}\big(\mathbf{D}^{3}g\big)\bigr\|_{\{1,2\}\{3\}} = \text{sup}\Bigl\{ \sum_{(S,j),(T,k),(U,l)} \mathbb{E}\big(\mathbf{D}^{3}g\big)_{(S,j),(T,k),(U,l)} x_{(S,j),(T,k)}y_{(U,l)}\\: \sum_{(S,j),(T,k)}x_{(S,j),(T,k)}^{2}\leq 1,\sum_{(U,l)}y_{(U,l)}^{2}\leq 1 \Bigr\}.
\end{multline}
Using the expression obtained for $\mathbb{E}\big(\mathbf{D}^{3}g\big)_{(S,j),(T,k),(U,l)}$, we obtain:
\begin{center}
\begin{align}
    \bigl\|\mathbb{E}\big(\mathbf{D}^{3}g\big)\bigr\|_{\{1,2\}\{3\}} =&\: \text{sup}\Bigl\{ \sum_{j,S,T} 2\: x_{(S,j),(T,j)}y_{(T,j)} + \sum_{j,S,T} 2\: x_{(S,j),(T,j)}y_{(S,j)} + \sum_{j,S,U} 2\: x_{(S,j),(S,j)}y_{(U,j)} \nonumber \\&\hspace{0.4\linewidth}: \sum_{(S,j),(T,k)}x_{(S,j),(T,k)}^{2}\leq 1,\sum_{(U,l)}y_{(U,l)}^{2}\leq 1 \Bigr\} \nonumber \\
    \leq&\: \text{sup}\Bigl\{ 2\sum_{S}\Bigl\lvert\sum_{j,T}  x_{(S,j),(T,j)}y_{(T,j)}\Bigr\rvert +2 \sum_{T}\Bigl\lvert\sum_{j,S} x_{(S,j),(T,j)}y_{(S,j)}\Bigr\rvert \nonumber \\&\hspace{0.1\linewidth}+ 2\sum_{j} \Bigl\lvert \sum_{S}x_{(S,j),(S,j)}\Bigr\rvert \: \Bigl\lvert \sum_{U}y_{(U,j)}\Bigr\rvert: \sum_{(S,j),(T,k)}x_{(S,j),(T,k)}^{2}\leq 1,\sum_{(U,l)}y_{(U,l)}^{2}\leq 1 \Bigr\} \nonumber \\
    \leq&\: \text{sup}\Bigl\{ 2\sum_{S}\|x\|\:\|y\| +2 \sum_{T}\|x\|\:\|y\| \nonumber \\&\hspace{0.05\linewidth}+ 2\sum_{j} \|x\|\sqrt{\binom{n}{q-1}}\:\|y\|\sqrt{\binom{n}{q-1}} : \sum_{(S,j),(T,k)}x_{(S,j),(T,k)}^{2}\leq 1,\sum_{(U,l)}y_{(U,l)}^{2}\leq 1 \Bigr\} \nonumber \\
    \leq&\: Cn\binom{n}{q-1},
\end{align}
\end{center}
where we have used the triangle inequality in the first inequality, and the Cauchy-Schwarz inequality for the second inequality (and we note that e.g. $\sum_{U}y_{(U,j)}$ is simply equal to the inner product of $y_{(j)}:=(y_{(U_1,j)},y_{(U_2,j)},\ldots)$ with the all-ones vector).

\item
The $D_{4}$ contribution can be written as:
\begin{equation}
    D_{4} := C\frac{1}{\sqrt{n}}\binom{n}{q-1}^{-3/2}\sum_{\substack{j,k,S,S' \\ \text{ s.t. }0<|S\cap S'|<q-1 \\\text{is odd}}} J_{S,j}\big(J_{S',k}\big)^{2}.
\end{equation}
We note that the main difference with the $D_{2}$ and $D_{3}$ contributions is that, for $D_{4}$, the sum is over the double index $j,k$ (instead of over the single index $j$), and over a restricted sum over sets $S,S'$ (instead of over a free sum over sets $S,S'$). To bound the moments of $D_{4}$, we will employ a similar method as for $D_{2}$ and $D_{3}$. The $r$th moment (with $r$ even) can be upper bounded as follows (where we drop the `$|S\cap S'|$ is odd' constraint using Lemma \ref{lemma:momentsofgaussianpolynomials} and denote the collection of subsets $S'$ such that $0<|S\cap S'|<q-1$ by $\sigma(S)$):
\begin{equation}
    \mathbb{E}\Big( |D_{4}|^{r} \Big) \leq C^{r}\Bigg[ \frac{1}{\sqrt{n}}\binom{n}{q-1}^{-3/2}\Bigg]^{r} \mathbb{E}\Big(\Big( \sum_{\substack{j,k,S,\\ S'\in \sigma(S)}} J_{S,j}\big(J_{S',k}\big)^{2} \Big)^{r} \Big).
\label{eq:D4bound}
\end{equation}
We note that $|\sigma(S)|$ can be upper bounded and lower bounded by $Cn^{q-2}$ (for some constants $C$). We define
\begin{equation}
    h:=\sum_{\substack{j,k,S,\\S'\in \sigma(S)}} J_{S,j}\big(J_{S',k}\big)^{2},
\label{eq:h}
\end{equation}
for which $\mathbb{E}(h) = 0$. We note that $h$ is a homogeneous polynomial in standard Gaussian random variables of degree $3$. To upper bound the moments of $g$, and thus the moments of $D_4$, we again use Lemma \ref{lemma:momentsofgaussianpolynomials} from \cite{upperbound_moments_diagonal_gaussianpolynomials}. We use Eq. \eqref{eq:diagupperbound} to find an upper bound for the $r$th moment of $h$. We first calculate $\mathbf{D}^{d}h$ for $d=1,2,3$. Then, for each $d$, we upper bound $\bigl\|\mathbb{E}\big(\mathbf{D}^{d}h\big)\bigr\|_{\mathcal{P}}$ for all partitions $\mathcal{P}$ of $[d]$. Thereby, we show that for all $d$ and associated partitions $\mathcal{P}([d])$, $\bigl\|\mathbb{E}\big(\mathbf{D}^{d}h\big)\bigr\|_{\mathcal{P}}$ can be upper bounded such that $\mathbb{E}\Big( |D_{4}|^{r} \Big)\leq \big(C\sqrt{n}\big)^{r}$ for all even $2\leq r\leq 16\cdot 2n$. The $0$th moment trivially satisfies this bound, and therefore it holds for all even $r\leq 16\cdot 2n$.

The derivatives of $h$ are equal to:
\begin{equation}
    D\,h = \big( \sum_{S'\in \sigma(S),p}J_{S',p}^{2}+2J_{S,j}\sum_{S'\in\sigma(S),p}J_{S',p} \big)_{(S,j)} \Longrightarrow \mathbb{E}\big(Dh\big) \leq \big(Cn^{q-1}\big)_{(S,j)},
\end{equation}
where the sum over $p$ runs from $0$ to $n$ and we have used the bounds on $|\sigma(S)|$. Note that this is a pointwise upper bound on the entries of the vector $\mathbb{E}\big(Dh\big)$, which will be enough to bound the corresponding norm. 
\begin{equation}
    D^{2}\,h =    \begin{pmatrix}
                2J_{T,k}+2J_{S,j}, & \text{if }T\in\sigma(S)\text{ and }\forall k \\
              2\sum_{S'\in\sigma(S),p}J_{S',p}, & \text{if } (T,k)=(S,j) \\
                0, & \text{otherwise}
                \end{pmatrix}_{(S,j),(T,k)} \Longrightarrow \mathbb{E}\big(D^{2}h\big) = (0)_{(S,j),(T,k)},
\end{equation}

\begin{equation}
    D^{3}\,h =    \begin{pmatrix}
               2, & \text{if }((U,l)=(T,k)\:(U,T\in \sigma(S)))\text{ or}\\
               & ((U,l)=(S,j)\text{ and }T\in \sigma(S)\text{ and }\forall k) \text{ or}\\
               & ((T,k)=(S,j)\text{ and }U\in \sigma(S)\text{ and }\forall l)\\
               0, & \text{otherwise} \end{pmatrix}_{(S,j),(T,k),(U,l)} \Longrightarrow \mathbb{E}\big(D^{3}h\big) = D^{3}h.
\end{equation}

In Table \ref{table:partitionsD4}, we give the values of $\bigl\|\mathbb{E}\big(\mathbf{D}^{d}h\big)\bigr\|_{\mathcal{P}}$ for all partitions $\mathcal{P}([d])$ for $d=1,2,3$. $\bigl\|\mathbb{E}\big(\mathbf{D}^{d}h\big)\bigr\|_{\mathcal{P}}$ for $d=1$ can be straightforwardly evaluated using Eq. \eqref{eq:ynorm} and for $d=2$ can be trivially evaluated by using $\mathbb{E}\big(\mathbf{D}^{2}h\big) = 0$. For $d=3$, $\bigl\|\mathbb{E}\big(\mathbf{D}^{d}h\big)\bigr\|_{\mathcal{P}}$ can be upper bounded using Eq. \eqref{eq:ynorm}, and the triangle and Cauchy-Schwarz inequalities. To obtain these upper bounds, we have again used the bounds on $|\sigma(S)|$.
\begin{table}[ht]
\centering
\captionsetup{width=.9\linewidth}
\begin{tabular}{||c c c||} 
 \hline
 $\mathcal{P}$ & $|\mathcal{P}|$ & $\bigl\|\mathbb{E}\big(\mathbf{D}^{d}h\big)\bigr\|_{\mathcal{P}}$ \\ [0.5ex] 
 \hline\hline
 $\{1\}$ & $1$ & $C\big( n^{q-1} \big)^{3/2}$ \\ [0.5ex]
 \hline
 $\{1,2\}$ & $1$ & 0 \\
 $\{1\}\{2\}$ & $2$ & 0 \\ [0.5ex]
 \hline
 $\{1,2,3\}$ & $1$ & $Cn^{q-1/2}$ \\
 $\{1\}\{2,3\}$ & $2$ & $Cn^{q}$ \\ 
 $\{2\}\{1,3\}$ & $2$ & $Cn^{q}$ \\ 
 $\{1,2\}\{3\}$ & $2$ & $Cn^{q}$ \\ 
 $\{1\}\{2\}\{3\}$ & $3$ & $Cn^{q/2}$ \\ [1ex] 
 \hline
\end{tabular}
\caption{The different partitions $\mathcal{P}$ of $[3]$ into non-empty parts, with the associated number of parts $|\mathcal{P}|$, and (the upper bounds for) the associated $\bigl\|\mathbb{E}\big(\mathbf{D}^{d}h\big)\bigr\|_{\mathcal{P}}$ for $h$ in Eq. \eqref{eq:h}.}
\label{table:partitionsD4}
\end{table}

Combining the upper bounds for $\bigl\|\mathbb{E}\big(\mathbf{D}^{d}h\big)\bigr\|_{\mathcal{P}}$ in Table \ref{table:partitionsD4} with the factor of $r^{|\mathcal{P}|/2}$($\leq Cn^{|\mathcal{P}|/2}$) in Eq. \eqref{eq:diagupperbound} and the normalization factor in Eq. \eqref{eq:D4bound}, we find -- using $\mathbb{E}(h) = 0$ -- that indeed $\mathbb{E}\Big( |D_{4}|^{r} \Big) \leq \big( C\sqrt{n} \big)^{r}$ for all even $r\leq 16\cdot 2n$.

\end{itemize}

In conclusion, we have shown that the $r$th moments (for even $r\leq 16\cdot 2n$) of $D_{0}$, $D_{1}$, $D_{2}$, $D_{3}$ and $D_{4}$ can be upper bounded by $\big(C\sqrt{n}\big)^{r}$, and hence, by Eq.~\eqref{eq:momentboundsplit}, the $(2n)$th moment of $A(1)$ can be upper bounded by $\big(C\sqrt{n}\big)^{2n}$. Thereby, we have also established that the second condition in Eq.~\eqref{eq:conditions} is satisfied.

\section{Two-colored SYK to standard SYK}
\label{app:2coltostandard}
In this Appendix, we give the proof of Lemma \ref{lem:2colouredtostandard}.
\begin{proof}
To establish Lemma \ref{lem:2colouredtostandard}, we now show that the state $\rho$ that achieves $\text{Tr}(H^{(2)}\rho) \geq C\sqrt{n}$ (with $H^{(2)}$ defined in Eq. \eqref{SYK_2}), with probability at least $1-\exp\big(-\Omega(n)\big)$ also achieves $\text{Tr}(H\rho)\geq C\sqrt{n}$ for the standard SYK Hamiltonian. To that end, we consider a standard SYK model with $2n$ Majorana operators and partition these Majorana operators into a subset of size $\frac{2n(q-1)}{q}$ and a complementary subset of size $\frac{2n}{q}$. The standard SYK model Hamiltonian $H$, see Eq.~\eqref{SYK_H}, consists of $\binom{2n}{q}$ terms. These terms are labeled by all ordered subsets $\{j_{1}<\ldots<j_{q}\}$, and $\mathcal{I}$ denotes the collection of these subsets. The terms in $H$ for which $q-1$ Majorana operators are in the first subset, and the other Majorana operator is in the complementary subset, are labelled by ordered subsets $\{j_{1}<\ldots<j_{q}:j_{1}<\ldots<j_{q-1}<\frac{q-1}{q}\leq j_{q}\}$. We denote the collection of these subsets by $\mathcal{T}$. The collection of other subsets is denoted by $\mathcal{T}' = \mathcal{I}\backslash \mathcal{T}$. $\mathcal{T}$ and $\mathcal{T}'$ thus correspond to collections of terms in the Hamiltonian. We denote the Hamiltonian consisting of the collection $\mathcal{T}$ by $H_{\mathcal{T}}$ and the Hamiltonian consisting of terms $\mathcal{T}'$ by $H_{\mathcal{T}'}$, hence $H=H_{\mathcal{T}}+H_{\mathcal{T}'}$. $H_{\mathcal{T}}$ corresponds exactly to the $2$-colored Hamiltonian in Eq. \eqref{SYK_2} when multiplied by 
\begin{equation}
    e^{-(q-1)/2} \leq \frac{\sqrt{\binom{2n}{q}}}{\sqrt{\frac{2n}{q}}\sqrt{\binom{\frac{q-1}{q}2n}{q-1}}} \leq e^{q/2},
\label{eq:const2col_to_standard}
\end{equation}
which, importantly, is lower bounded and upper bounded by a constant in $n$. Note that the sizes of the two subsets into which the Majorana operators are partitioned can in fact be \textit{any} $c2n$ and $(1-c)2n$ for $c=O(1)$ (instead of just $\frac{2n(q-1)}{q}$ and $\frac{2n}{q}$). The factor in Eq.~\eqref{eq:const2col_to_standard} is lower bounded and upper bounded by a constant in $n$ as well for all of these other partitions. Hence $n$ is not constrained to be a multiple of $q$.

For any state $\rho$, $\mathbb{E}\big(\text{Tr}(H\rho)\big) = 0$, where the expectation value is w.r.t. the couplings in $H$ since the couplings are random variables with zero mean. The state $\rho_{\theta}$ defined in Eq.~ \eqref{eq:defrhotheta} is able to achieve $\text{Tr}(H^{(2)}\rho_{\theta} )\geq C\sqrt{n}$ (with high probability) since it is constructed using a circuit that itself \textit{depends} on the random couplings $J_{I}$ ($I\in \mathcal{T}$) appearing in $H^{(2)}$. Since $\rho_{\theta}$ does \textit{not} depend on the couplings $J_{I}$ with $I\in \mathcal{T}'$, we have $\text{Tr}(H_{\mathcal{T}'}\rho_{\theta} ) = 0$. Since: (i) $\lvert \text{Tr}(C_I \rho) \rvert \leq 1$ (for any $\rho$) for $I\in \mathcal{T}'$, (ii) that each $J_{I}$ is a standard Gaussian random variable, and (iii) that $|\mathcal{T}'|\leq \binom{2n}{q}$, the quantity 
\begin{equation}
\text{Tr}\big(H_{\mathcal{T}'}\rho \big) = \binom{2n}{q}^{-1/2}\sum_{I\in \mathcal{T}'}J_{I}\text{Tr}(C_I \rho)
\end{equation}
is a Gaussian random variable with zero mean and variance at most one, for any $\rho$. Then, $\mathbb{E}\big[ \exp(t\:\text{Tr}(H_{\mathcal{T}'}\rho)) \big] \leq \exp(\frac{1}{2}t^{2})$ for all $t\geq 0$. Applying Chernoff's bound to $\text{Tr}(H_{\mathcal{T}'}\rho)$, and choosing $t = C\sqrt{n}$, we obtain:
\begin{equation}
    \text{Pr}\Big[ \lvert \text{Tr}(H_{\mathcal{T}'}
    \rho) \rvert \geq C\sqrt{n} \Big] \leq 2\exp\big(-\Omega(n)\big),
\label{eq:boundT'}
\end{equation}
for any constant $C$.

Using Eq.~\eqref{eq:boundT'} and $\text{Tr}(H\rho) = \text{Tr}(H_{\mathcal{T}}\rho) + \text{Tr}( H_{\mathcal{T}'}\rho)$, we conclude that the state $\rho_{\theta}$ which achieves $\text{Tr}( H^{(2)}\rho_{\theta})\geq C\sqrt{n}$ (i.e., for the $2$-colored SYK Hamiltonian) with probability at least $1-\exp\big( -\Omega(n) \big)$, also achieves $\text{Tr}(H\rho_{\theta})\geq C\sqrt{n}$ (where $H$ is the standard SYK Hamiltonian in Eq.~\eqref{SYK_H}) with probability at least $1-\exp\big( -\Omega(n) \big)$. Therefore, $\lambda_{\max}(H)\geq C\sqrt{n}$ with probability at least $1-\exp\big( -\Omega(n) \big)$.
\end{proof}


\begin{thebibliography}{34}
\providecommand{\natexlab}[1]{#1}
\providecommand{\url}[1]{\texttt{#1}}
\expandafter\ifx\csname urlstyle\endcsname\relax
  \providecommand{\doi}[1]{doi: #1}\else
  \providecommand{\doi}{doi: \begingroup \urlstyle{rm}\Url}\fi

\bibitem[Haldar et~al.(2021)Haldar, Tavakol, and Scaffidi]{SYKScaffidi}
Arijit Haldar, Omid Tavakol, and Thomas Scaffidi.
\newblock Variational wave functions for {Sachdev-Ye-Kitaev} models.
\newblock \emph{Phys. Rev. Research}, 3:\penalty0 023020, 2021.
\newblock URL \url{https://doi.org/10.1103/PhysRevResearch.3.023020}.

\bibitem[Hastings and O'Donnell(2022)]{HO:approxferm}
Matthew~B. Hastings and Ryan O'Donnell.
\newblock Optimizing strongly interacting fermionic {H}amiltonians.
\newblock In \emph{Proceedings of the 54th Annual ACM SIGACT Symposium on
  Theory of Computing}, STOC 2022, page 776–789, New York, NY, USA, 2022.
  ACM.
\newblock URL \url{https://doi.org/10.1145/3519935.3519960}.

\bibitem[Gharibian et~al.(2015)Gharibian, Huang, Landau, and
  Shin]{Gharibian2015}
Sevag Gharibian, Yichen Huang, Zeph Landau, and Seung~Woo Shin.
\newblock {Quantum Hamiltonian Complexity}.
\newblock \emph{Foundations and Trends in Theoretical Computer Science},
  10\penalty0 (3):\penalty0 159--282, 2015.
\newblock URL \url{https://doi.org/10.1561\%2F0400000066}.

\bibitem[Bravyi and Kitaev(2002)]{BK:ferm}
Sergey~B. Bravyi and Alexei~Yu. Kitaev.
\newblock Fermionic quantum computation.
\newblock \emph{Annals of Physics}, 298\penalty0 (1):\penalty0 210--226, 2002.
\newblock URL \url{https://doi.org/10.1006/aphy.2002.6254}.

\bibitem[Khanna et~al.(2001)Khanna, Sudan, Trevisan, and
  Williamson]{khanna2001approximability}
Sanjeev Khanna, Madhu Sudan, Luca Trevisan, and David~P Williamson.
\newblock The approximability of constraint satisfaction problems.
\newblock \emph{SIAM Journal on Computing}, 30\penalty0 (6):\penalty0
  1863--1920, 2001.
\newblock URL \url{https://doi.org/10.1137/S0097539799349948}.

\bibitem[Goemans and Williamson(1995)]{GW:semidef}
Michel~X. Goemans and David~P. Williamson.
\newblock Improved approximation algorithms for maximum cut and satisfiability
  problems using semidefinite programming.
\newblock \emph{J. ACM}, 42\penalty0 (6):\penalty0 1115–1145, 1995.
\newblock URL \url{https://doi.org/10.1145/227683.227684}.

\bibitem[Kraus and Cirac(2010)]{KC:ferm}
Christina~V Kraus and J~Ignacio Cirac.
\newblock Generalized {H}artree–{F}ock theory for interacting fermions in
  lattices: numerical methods.
\newblock \emph{New Journal of Physics}, 12\penalty0 (11):\penalty0 113004,
  2010.
\newblock URL \url{https://doi.org/10.1088/1367-2630/12/11/113004}.

\bibitem[Bravyi and Gosset(2017)]{BG:impurity}
Sergey Bravyi and David Gosset.
\newblock Complexity of quantum impurity problems.
\newblock \emph{Communications in Mathematical Physics}, 356\penalty0
  (2):\penalty0 451–500, 2017.
\newblock URL \url{https://doi.org/10.1007/s00220-017-2976-9}.

\bibitem[Bravyi and Koenig(2012)]{BK:simul}
Sergey Bravyi and Robert Koenig.
\newblock Classical simulation of dissipative fermionic linear optics.
\newblock \emph{Quant. Inf. Comp.}, 12\penalty0 (11--12):\penalty0 925--943,
  2012.
\newblock URL \url{https://doi.org/10.48550/arXiv.1112.2184}.

\bibitem[de~Melo et~al.(2013)de~Melo, Ćwikliński, and Terhal]{MCT:gaussian}
Fernando de~Melo, Piotr Ćwikliński, and Barbara~M Terhal.
\newblock The power of noisy fermionic quantum computation.
\newblock \emph{New Journal of Physics}, 15\penalty0 (1):\penalty0 013015,
  2013.
\newblock URL \url{https://doi.org/10.1088/1367-2630/15/1/013015}.

\bibitem[Lieb(1973)]{lieb}
Elliott~H. Lieb.
\newblock {The classical limit of quantum spin systems}.
\newblock \emph{Communications in Mathematical Physics}, 31\penalty0
  (4):\penalty0 327 -- 340, 1973.
\newblock URL \url{https://doi.org/10.1007/BF01646493}.

\bibitem[Bravyi et~al.(2019)Bravyi, Gosset, K{\"o}nig, and
  Temme]{BGKT:manybody}
Sergey Bravyi, David Gosset, Robert K{\"o}nig, and Kristan Temme.
\newblock Approximation algorithms for quantum many-body problems.
\newblock \emph{Journal of Mathematical Physics}, 60\penalty0 (3):\penalty0
  032203, 2019.
\newblock URL \url{https://doi.org/10.1063/1.5085428}.

\bibitem[Bansal et~al.(2009)Bansal, Bravyi, and Terhal]{BBT:approx}
Nikhil Bansal, Sergey Bravyi, and Barbara~M Terhal.
\newblock Classical approximation schemes for the ground-state energy of
  quantum and classical {Ising spin Hamiltonians} on planar graphs.
\newblock \emph{Quantum Inf. Comp.}, 9\penalty0 (7-8):\penalty0 701--720, 2009.
\newblock URL \url{https://doi.org/10.48550/arXiv.0705.1115}.

\bibitem[Gharibian and Kempe(2012)]{GK:dense}
Sevag Gharibian and Julia Kempe.
\newblock Approximation algorithms for {QMA}-complete problems.
\newblock \emph{SIAM Journal on Computing}, 41\penalty0 (4):\penalty0
  1028--1050, 2012.
\newblock URL \url{https://doi.org/10.1137/110842272}.

\bibitem[Brandao and Harrow(2013)]{BH:dense-approx}
Fernando~G.S.L. Brandao and Aram~W. Harrow.
\newblock Product-state approximations to quantum ground states.
\newblock In \emph{Proceedings of the Forty-Fifth Annual ACM Symposium on
  Theory of Computing}, STOC '13, page 871–880, New York, NY, USA, 2013.
  Association for Computing Machinery.
\newblock ISBN 9781450320290.
\newblock URL \url{https://doi.org/10.1145/2488608.2488719}.

\bibitem[Harrow and Montanaro(2017)]{HM:extremal}
Aram~W. Harrow and Ashley Montanaro.
\newblock Extremal eigenvalues of local {H}amiltonians.
\newblock \emph{{Quantum}}, 1:\penalty0 6, April 2017.
\newblock URL \url{https://doi.org/10.22331/q-2017-04-25-6}.

\bibitem[Bergamaschi(2022)]{berga:improved}
Thiago Bergamaschi.
\newblock Improved product-state approximation algorithms for quantum local
  {H}amiltonians, 2022.
\newblock URL \url{https://doi.org/10.48550/arXiv.2210.08680}.

\bibitem[Gharibian and Parekh(2019)]{Parekh19}
Sevag Gharibian and Ojas Parekh.
\newblock Almost optimal classical approximation algorithms for a quantum
  generalization of max-cut.
\newblock 2019.
\newblock URL \url{https://doi.org/10.4230/LIPICS.APPROX-RANDOM.2019.31}.

\bibitem[Alon et~al.(2006)Alon, Makarychev, Makarychev, and
  Naor]{alon2006quadratic}
Noga Alon, Konstantin Makarychev, Yury Makarychev, and Assaf Naor.
\newblock Quadratic forms on graphs.
\newblock \emph{Inventiones mathematicae}, 163\penalty0 (3):\penalty0 499--522,
  2006.
\newblock URL \url{https://doi.org/10.1007/s00222-005-0465-9}.

\bibitem[Xu et~al.(2020)Xu, Susskind, Su, and Swingle]{sparseSYK}
Shenglong Xu, Leonard Susskind, Yuan Su, and Brian Swingle.
\newblock A sparse model of quantum holography, 2020.
\newblock URL \url{https://doi.org/10.48550/arXiv.2008.02303}.

\bibitem[Garc\'{\i}a-Garc\'{\i}a et~al.(2021)Garc\'{\i}a-Garc\'{\i}a, Jia,
  Rosa, and Verbaarschot]{SYK-sparse}
Antonio~M. Garc\'{\i}a-Garc\'{\i}a, Yiyang Jia, Dario Rosa, and Jacobus J.~M.
  Verbaarschot.
\newblock Sparse {Sachdev-Ye-Kitaev} model, quantum chaos, and gravity duals.
\newblock \emph{Phys. Rev. D}, 103:\penalty0 106002, 2021.
\newblock URL \url{https://doi.org/10.1103/PhysRevD.103.106002}.

\bibitem[Frieze and Karo{\'n}ski(2016)]{frieze2016introduction}
Alan Frieze and Micha{\l} Karo{\'n}ski.
\newblock \emph{Introduction to random graphs}.
\newblock Cambridge University Press, 2016.
\newblock URL \url{https://doi.org/10.1017/CBO9781316339831}.

\bibitem[Aharonov et~al.(2013)Aharonov, Arad, and Vidick]{AV:QPCP}
Dorit Aharonov, Itai Arad, and Thomas Vidick.
\newblock The quantum {PCP} conjecture.
\newblock \emph{ACM SIGACT News}, 44:\penalty0 47--79, 2013.
\newblock URL \url{https://doi.org/10.48550/arXiv.1309.7495}.

\bibitem[Bravyi et~al.(2021)Bravyi, Chowdhury, Gosset, and Wocjan]{BCGW}
Sergey Bravyi, Anirban Chowdhury, David Gosset, and Pawel Wocjan.
\newblock On the complexity of quantum partition functions, 2021.
\newblock URL \url{https://doi.org/10.48550/arXiv.2110.15466}.

\bibitem[Boucheron et~al.(2003)Boucheron, Lugosi, and
  Massart]{bouch:concentration}
Stéphane Boucheron, Gábor Lugosi, and Pascal Massart.
\newblock Concentration inequalities using the entropy method.
\newblock \emph{The Annals of Probability}, 31\penalty0 (3):\penalty0
  1583--1614, 2003.
\newblock URL \url{https://doi.org/10.1214/aop/1055425791}.

\bibitem[Tomioka and Suzuki(2014)]{spectralnormrandomtensors}
Ryota Tomioka and Taiji Suzuki.
\newblock Spectral norm of random tensors.
\newblock \emph{arXiv}, 2014.
\newblock URL \url{https://doi.org/10.48550/arXiv.1407.1870}.

\bibitem[Anshu et~al.(2021)Anshu, Gosset, Morenz~Korol, and
  Soleimanifar]{ImprovedApproxAlg}
Anurag Anshu, David Gosset, Karen~J. Morenz~Korol, and Mehdi Soleimanifar.
\newblock Improved approximation algorithms for bounded-degree local
  {H}amiltonians.
\newblock \emph{Phys. Rev. Lett.}, 127:\penalty0 250502, 2021.
\newblock URL \url{https://doi.org/10.1103/PhysRevLett.127.250502}.

\bibitem[Laurent and Massart(2000)]{chisquaredtailbounds}
B.~Laurent and P.~Massart.
\newblock {Adaptive estimation of a quadratic functional by model selection}.
\newblock \emph{The Annals of Statistics}, 28\penalty0 (5):\penalty0 1302 --
  1338, 2000.
\newblock URL \url{https://doi.org/10.1214/aos/1015957395}.

\bibitem[Bollobás(1998)]{Bollobas}
Béla Bollobás.
\newblock \emph{Modern Graph Theory}.
\newblock Graduate Texts in Mathematics 184. Springer-Verlag New York, 1998.
\newblock URL \url{https://doi.org/10.1007/978-1-4612-0619-4}.

\bibitem[Dirac(1952)]{dirac1952some}
Gabriel~Andrew Dirac.
\newblock Some theorems on abstract graphs.
\newblock \emph{Proceedings of the London Mathematical Society}, 3\penalty0
  (1):\penalty0 69--81, 1952.
\newblock URL \url{https://doi.org/10.1112/plms/s3-2.1.69}.

\bibitem[Abramowitz and Stegun(1964)]{SA:book}
Milton Abramowitz and Irene~A. Stegun.
\newblock \emph{Handbook of Mathematical Functions with Formulas, Graphs, and
  Mathematical Tables}.
\newblock Dover, New York, 1964.
\newblock ISBN 0-486-61272-4.

\bibitem[Latala(2006)]{GaussianChaoses}
Rafal Latala.
\newblock Estimates of moments and tails of {G}aussian chaoses.
\newblock \emph{The Annals of Probability}, 34\penalty0 (6), nov 2006.
\newblock URL \url{https://doi.org/10.1214\%2F009117906000000421}.

\bibitem[de~la Pena et~al.(1994)de~la Pena, Montgomery-Smith, and
  Szulga]{DecouplingInequalities}
V.~H. de~la Pena, S.~J. Montgomery-Smith, and Jerzy Szulga.
\newblock Contraction and decoupling inequalities for multilinear forms and
  u-statistics.
\newblock \emph{The Annals of Probability}, 22\penalty0 (4):\penalty0
  1745--1765, 1994.
\newblock URL \url{https://doi.org/10.48550/arXiv.math/9406214}.

\bibitem[Adamczak and
  Wolff(2015)]{upperbound_moments_diagonal_gaussianpolynomials}
Rados{\l}aw Adamczak and Pawe{\l} Wolff.
\newblock Concentration inequalities for non-lipschitz functions with bounded
  derivatives of higher order.
\newblock \emph{Probability Theory and Related Fields}, 162\penalty0
  (3):\penalty0 531--586, 2015.
\newblock URL \url{https://doi.org/10.1007/s00440-014-0579-3}.

\end{thebibliography}
\end{document}